\documentclass[acmsmall,nonacm]{acmart} 

\usepackage[dvipsnames]{xcolor}
\newcommand{\Dee}{\mathcal{D}}
\usepackage{flushend}
\usepackage[a-1b]{pdfx}   
\usepackage{balance}

\usepackage{color}
\usepackage[labelfont=bf]{caption}
\usepackage{subcaption}
\usepackage{hyperref}
\usepackage{makecell}
\usepackage{fixltx2e}
\usepackage{enumitem}
\usepackage{rotating,multirow}
\usepackage{etoolbox,xspace}
\usepackage{algorithmicx}
\usepackage{algorithm}
\usepackage{fontawesome}
\usepackage{algpseudocode}
\usepackage{pgfplots}

\algtext*{EndWhile}
\algtext*{EndIf}
\algtext*{EndFunction}
\algtext*{EndFor}
\algrenewcommand\algorithmicrequire{\textbf{Input:}}
\algrenewcommand\algorithmicensure{\textbf{Output:}}

\usepackage{amsthm}

\usepackage{tabularx}
\usepackage{blindtext}
\usepackage{ragged2e}  
\usepackage{booktabs}  
\usepackage{textcomp}
\usepackage{url}
\usepackage{comment}
\usepackage{enumitem}
\usepackage{nicematrix}
\usepackage[simplified]{pgf-umlcd}
\usepackage{tikz}
\usepackage{ifthen}
\usepackage{diagbox}

\newcolumntype{Y}{>{\RaggedRight\arraybackslash}X}
\DeclareMathOperator*{\argmin}{arg\,min}

\newtheorem{definition}{Definition}
\newtheorem{theorem}{Theorem}
\newtheorem{lemma}{Lemma}
\usepackage{booktabs}
\usepackage{xspace}

\newenvironment{proof}{\paragraph{Proof:}}{\hfill$\square$}
\newcounter{example}
\renewcommand{\theexample}{\arabic{example}}
\newenvironment{example}{
        \vspace{1ex}
        \refstepcounter{example}
        {\noindent\bf \textit{Example} \theexample:}}
	{\vspace{1ex}}

\newcounter{runningexample}
\renewcommand{\therunningexample}{\arabic{runningexample}}
\newenvironment{runningexample}{
        \vspace{1ex}
        \refstepcounter{runningexample}
        {\noindent\bf Running Example \therunningexample:}}
	{\vspace{1ex}}

\newcommand{\squishlist}{
 \begin{list}{$\bullet$}
  { \setlength{\itemsep}{0pt}
     \setlength{\parsep}{1pt}
     \setlength{\topsep}{1pt}
     \setlength{\partopsep}{0pt}
     \setlength{\leftmargin}{1em}
     \setlength{\labelwidth}{1em}
     \setlength{\labelsep}{0.5em} } }

\newcommand{\squishend}{
  \end{list}
}

\definecolor{americanrose}{rgb}{1.0, 0.01, 0.24}
\definecolor{airforceblue}{rgb}{0.36, 0.54, 0.66}
\definecolor{ao(english)}{rgb}{0.0, 0.5, 0.0}
\definecolor{ao}{rgb}{0.0, 0.0, 1.0}

\newcommand{\aryan}[1]{\textcolor{blue}{Aryan: #1}}

\newcommand{\eat}[1]{}

\newcommand{\submit}[1]{}

\usepackage{xspace}
\newcommand{\stitle}[1]{\vspace{1mm}\noindent{\bf #1:}}

\newcommand{\Gee}{\mathcal{G}}
\newcommand{\Ze}{\mathbb{Z}}

\newcommand{\dee}{\mathcal{D}}
\newcommand{\eh}{\mathcal{H}}
\newcommand{\Aa}{\mathcal{A}}
\newcommand{\Es}{\mathcal{S}}

\newcommand{\opt}{\mathsf{opt}}
\newcommand{\ef}{\mathscr{f}}
\newcommand{\gi}{\mathscr{g}}

\newcommand{\Qu}{Q}

\renewcommand{\Re}{\mathbb{R}}%
\newcommand{\eps}{\varepsilon}

\renewcommand{\marginpar}[2][]{}

\newcommand{\prob}{fair minimum-weight cover\xspace}

\usepackage{bbm, dsfont}
\usepackage{esvect}

  \allowdisplaybreaks

  \newcommand{\groupsize}{\ell}

\AtBeginDocument{%
  }

\makeatletter

\newenvironment{proofoflemma}[1]{%
  \par\pushQED{\qed}%
  \normalfont\topsep6\p@\@plus6\p@\relax
  \trivlist
  \item[\hskip\labelsep\itshape Proof of Lemma~\ref{#1}.]%
}{%
  \popQED\endtrivlist\@endpefalse
}
\makeatother

\newcommand{\revone}[1]{\textcolor{black}{#1}}
\newcommand{\revtwo}[1]{\textcolor{black}{#1}}

\newcommand{\revtwofour}[1]{\textcolor{black}{#1}}

\newenvironment{envrevone}{\color{black}}{}
\newenvironment{envrevtwo}{\color{black}}{}
\newenvironment{envrevfour}{\color{black}}{}
\newenvironment{envrevtwofour}{\color{black}}{}

\begin{document}

\title{Weighted Set Multi-Cover on Bounded Universe and Applications in Package Recommendation}
\titlenote{This work was partially supported by NSF grant IIS-2348919.}


\author{Nima Shahbazi}
\orcid{0000-0001-7016-3807}
\authornote{The first two authors contributed equally to this research.}
\affiliation{%
  \institution{Department of Computer Science, University of Illinois Chicago}
  \city{Chicago}
  \state{Illinois}
  \country{USA}
}
\email{nshahb3@uic.edu}

\author{Aryan Esmailpour}
\orcid{0009-0000-3798-9578}
\authornotemark[2]
\affiliation{%
  \institution{Department of Computer Science, University of Illinois Chicago}
  \city{Chicago}
  \state{Illinois}
  \country{USA}
}
\email{aesmai2@uic.edu}

\author{Stavros Sintos}
\orcid{0000-0002-2114-8886}
\affiliation{%
  \institution{Department of Computer Science, University of Illinois Chicago}
  \city{Chicago}
  \state{Illinois}
  \country{USA}
}
\email{stavros@uic.edu}

\begin{abstract}
\begin{envrevtwofour}
The weighted set multi-cover problem is a fundamental generalization of set cover that arises in data-driven applications where one must select a small, low-cost subset from a large collection of candidates under coverage constraints. In data management settings, such problems arise naturally either as expressive database queries or as post-processing steps over query results, for example, when selecting representative or diverse subsets from large relations returned by database queries for decision support, recommendation, fairness-aware data selection, or crowd-sourcing.
While the general weighted set multi-cover problem is NP-complete, many practical workloads involve a \emph{bounded universe} of attributes or items that must be covered, leading to the Weighted Set Multi-Cover with Bounded Universe (WSMC-BU) problem, where the universe size is constant. Despite its relevance in large-scale data processing pipelines, little is known about efficient algorithms for this special case.
In this paper, we develop exact and approximation algorithms for WSMC-BU. We first discuss a dynamic programming algorithm that solves WSMC-BU exactly in $O(n^{\ell+1})$ time, where $n$ is the number of input sets and $\ell=O(1)$ is the universe size. We then present a $2$-approximation algorithm based on linear programming and rounding, running in $O(\mathcal{L}(n))$ time, where $\mathcal{L}(n)$ denotes the complexity of solving a linear program with $O(n)$ variables. To further improve efficiency for large datasets, we propose a faster $(2+\varepsilon)$-approximation algorithm with running time $O(n \log n + \mathcal{L}(\log W))$, where $W$ is the ratio of the total weight to the minimum weight, and $\varepsilon$ is an arbitrary constant specified by the user.
Our results provide the first practical constant-factor approximation algorithms for WSMC-BU, with approximation guarantees independent of the universe size. Extensive experiments on real and synthetic datasets demonstrate that our methods consistently outperform greedy and standard LP-rounding baselines in both solution quality and runtime, making them suitable for data-intensive selection tasks over large query outputs.
\end{envrevtwofour}
\end{abstract}

%

\maketitle

\newcommand{\sol}{\mathcal{H}}
\renewcommand{\prob}{\textsf{WSMC-BU}}

\section{Introduction}\label{sec:intro}

\emph{Set cover} is one of the most fundamental problems in computer science, with applications across diverse domains such as network design, information retrieval, resource allocation, sensor networks, VLSI
design, machine learning, crew scheduling, facility location, computational biology, and network security~\cite{revelle1976applications, vemuganti1998applications, rubin1973technique, cho2012chapter, ge2010application, hoffman1993solving, lourencco2001multiobjective}.
In databases, the set cover problem models several data management and query processing tasks, including pipelined filter optimization~\cite{munagala2005pipelined}, query routing in big data architectures~\cite{narayan2016efficient}, and package-to-group recommendations~\cite{serbos2017fairness,qi2016recommending}.
Formally, given a universe $\Gee$ of $\ell$ items and a family $\Dee$ of $n$ subsets of $\Gee$, the goal is to select the smallest number of sets from $\Dee$ whose union covers all items in $\Gee$. Despite being one of Karp’s 21 NP-complete problems, both a greedy algorithm and an LP-based rounding algorithm achieve $O(\log \ell)$-approximation~\cite{chvatal1979greedy}, which is known to be optimal~\cite{dinur2014analytical}.

Numerous variants of the set cover problem have been studied in both the theory and database communities. Examples include the \emph{set multi-cover problem}~\cite{dobson1982worst} 
the \emph{multi-set multi-cover problem}~\cite{hua2009exact} 
and the \emph{fair set cover problem}~\cite{dehghankar2025fair}
. 
For all these variants, weighted versions have also been investigated, where each set is associated with a weight and the objective is to minimize the total weight of the selected sets while satisfying the covering constraints.

In this work, we focus on the \emph{weighted set multi-cover problem}. More formally, let $\Gee$ be a universe of $\ell$ items where each item $g\in \Gee$ has a non-negative demand $\Qu_g$, and let $\Dee$ be a family of $n$ subsets of $\Gee$, where each set $t\in \Dee$ has an associated non-negative weight $w_t$. The goal is to select a collection of sets from $\Dee$ with minimum total weight such that every item $g\in\Gee$ is covered at least $\Qu_g$ times. This problem is NP-complete, yet, interestingly, it also admits a greedy $O(\log \ell)$-approximation algorithm~\cite{dobson1982worst}, where in each iteration it adds in the solution the set minimizing the ratio of its weight over the number of items with positive demands it covers.  Furthermore, the standard LP-rounding algorithm for the weighted set multi-cover problem returns an $O(\log \sum_{g\in \Gee}(Q_g))$-approximation solution~\cite{Vazirani2010-hn}.


\begin{envrevtwofour}
In data management settings, the weighted set multi-cover problem captures the algorithmic core of a class of expressive selection tasks over large relations. Conceptually, these tasks can be viewed as \emph{package recommendation} or \emph{constrained selection queries}~\cite{brucato2016scalable, brucato2017scalable, drosou2012disc}, where the goal is to retrieve a package (subset) of tuples that satisfies global coverage or representation constraints over attributes, while minimizing a cost function, such as total weight, resource usage, or acquisition cost. Such constraints go beyond what is directly expressible in standard SQL, and therefore are often handled either through specialized query operators or as post-processing steps over query results in data analytics pipelines. 

The weighted set multi-cover problem has a wide range of applications in data management, decision support, and data selection. 
To demonstrate this, we next describe some representative data management settings where this problem arises naturally. 

{\bf Fairness-aware data selection: \cite{Nargesian2021fair, stoyanovich2018fair, shahbazi2024fairness, chameleon}} A common problem in data management and ML pipelines is to construct a training dataset that satisfies specified representation requirements over demographic groups. 

\begin{example} \label{ex1}
    In the Chameleon system~\cite{chameleon}, they are given a dataset of face images, where each image belongs to one or more demographic groups (e.g., ``female'', ``Black'', ``young''). To ensure that the downstream model is not biased, someone must select sufficiently many images from each demographic group. Since an image can simultaneously belong to multiple groups (e.g., an image may be both ``female'' and ``young''), every selected image contributes coverage to multiple demographic requirements. The goal is to compute a small number of images to satisfy all requirements.
\end{example}

Similar problems appear in many other fairness applications. For instance, consider student admissions in a university, where the goal is to select a cohort of applicants while minimizing the total financial aid or scholarship budget. Each applicant is characterized by protected attributes such as gender and ethnicity, and institutions are often subject to diversity constraints. For instance, a Hispanic-Serving Institution in the U.S. may be required to admit at least $25\%$ Hispanic students, and, due to persistent gender imbalance in STEM fields, a university may additionally require that at least $25\%$ of admitted students be women. In this formulation, each applicant corresponds to a set with an associated cost (the offered financial aid), and each protected attribute category corresponds to an item with a specified demand.

Our techniques are applicable in all these examples, where the goal is to select a subset of data to satisfy selection requirements across different demographics. 


{\bf Crowd-Sourced Task Assignment / Data Enrichment: \cite{crowd6, Crowd2, Crowd3}} Generally, the problem can be formulated as follows: standard DB work on task assignment and crowdsourcing frames the problem as selecting workers (each with a cost and set of skills) to cover a set of tasks. This is exactly an instance of the weighted set multi-cover problem: Each worker defines a set that covers the skill/task types they can do,  each skill/task represents an item in the universe, each demand is the number of worker-assignments needed per task/type, and the weight is defined as the worker cost.
To highlight it further, we show the following example, with an overview shown in Figure~\ref{fig:example-figure}.

\begin{example}
\label{ex:CS}
In crowd-sourced data enrichment, the data is often collected from unreliable sources on the internet and needs verification using a human in the loop \cite{Marcus2011}. For this, companies often use expert pools like \textit{Amazon Mechanical Turk} to access specialized experts who possess varying proficiencies in domains such as medicine, law, and finance \cite{mturk}. Each expert has an hourly rate, and the goal is to systematically verify specific attributes—such as the biosafety level of a research lab or the regulatory compliance of a financial institution- to ensure they are accurate before they are finalized in a knowledge graph \cite{Dong2014}. Since these attributes often carry high stakes, the task requires a predefined level of redundancy where multiple independent experts must reach a consensus on the same data point to establish it as ground truth \cite{Sheng2008}. For example, a task may pose the specific question, ``Does the provided documentation confirm that this medical facility is licensed for radiology?'' This question is distributed to a set of multiple experts, and their responses are aggregated to ensure reliability; the answer is only integrated into the database if a strict consensus is met, such as through a simple majority vote or a requirement that at least $90\%$ of the experts agree on the response. This process involves identifying the required expertise for each attribute, evaluating the available solver pool for overlapping skill sets that can satisfy several domain requirements simultaneously, and selecting the most cost-effective group of individuals to meet the total redundancy constraints \cite{Chai2016, Trushkowsky2013}. The final goal of this pipeline is the production of an enriched knowledge graph, which serves as a highly reliable and structured database where raw, noisy information has been transformed into a network of verified entities and relationships. An overview of this pipeline is shown in Figure~\ref{fig:example-figure}.
\end{example}

{\bf Constrained package recommendations: \cite{serbos2017fairness,qi2016recommending}} 
Given a user's preferences, the goal is to recommend a package of items that minimizes cost while satisfying the specified constraints.

\begin{example}
\label{ex4}
Consider a recommendation application that has access to a dataset containing a collection of restaurants (e.g., the Yelp dataset \cite{yelp2025dataset}). Suppose a user specifies the following cuisine preferences: \emph{very interested} in Italian and omnivore cuisines, \emph{interested} in Japanese cuisine, and \emph{somewhat interested} in Mediterranean and vegan cuisines.
Each restaurant may offer multiple cuisines; for example, \emph{Orsa \& Winston} in Los Angeles serves Japanese, Italian, and omnivore dishes. In addition, each restaurant is associated with an average price per person.
The objective is to select a set of restaurants with a minimum total cost that includes at least 3 Italian, 2 Japanese, 1 Mediterranean, 3 omnivore, and 1 vegan restaurants.
\end{example}

\begin{example}
\label{ex3}
Consider a software company that must hire developers with experience in multiple programming languages. Specifically, the company requires at least $10$ developers with experience in Python, $7$ with experience in C, $5$ with experience in Java, and at least $3$ with experience in PHP. Each developer lists the programming languages they are proficient in, along with an associated salary. A developer who is proficient in multiple languages can simultaneously satisfy the requirements for each of those languages (for example, a developer proficient in both C and Python can count toward both the C and Python requirements). The objective is to select a set of developers with a minimum total salary that satisfies all the language requirements.
\end{example}

\end{envrevtwofour}

\begin{figure*}[t]
  \centering
  \includegraphics[scale = 0.4]{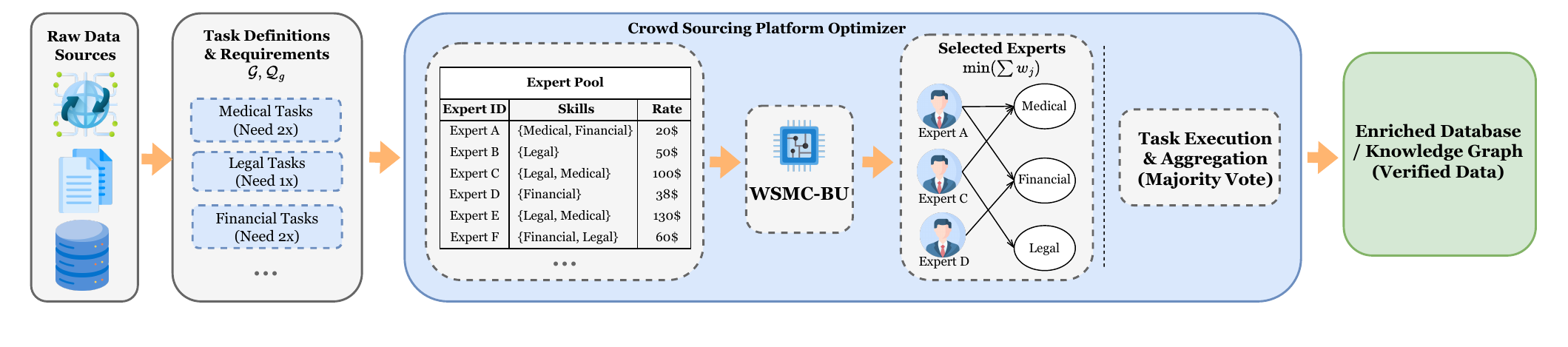}
  \vspace{-3.0em}
  \caption{\begin{envrevtwofour}
  Overview of crowd-sourced data enrichment using WSMC-BU.
  \end{envrevtwofour}}
  \label{fig:example-figure}
  \vspace{-1em}
\end{figure*}

\vspace{-0.2em}


All the examples above are instances of the weighted set multi-cover problem. 
\begin{envrevtwofour}
Since this problem is NP-complete, in the above applications, the existing literature often relies on a greedy approach to solve the problem \cite{chameleon, serbos2017fairness}. This is due to the simplicity of the greedy algorithm, while offering a reasonable approximation. Indeed, it is known that the approximation factor of the greedy algorithm cannot be improved in the general settings (unbounded universe) \cite{dobson1982worst}. However, we observe that a common feature of all these applications is that while the number of sets is large (e.g., images, students, restaurants, experts), the universe size $\ell=|\Gee|$ is relatively small (e.g., number of skills, tasks, demographic groups, or cuisines). Hence, it is reasonable to assume that in many data management and package recommendation tasks, the universe size is bounded by a constant, i.e., $\ell=O(1)$.
\end{envrevtwofour}
This motivates the study of the \emph{Weighted Set Multi-Cover problem on Bounded Universe}, or \prob\ for short. Its definition is identical to the weighted set multi-cover problem, except that $|\Gee|=O(1)$. 

Despite the many applications of \prob, research on this special case remains limited. Bredereck et al.~\cite{bredereck2020mixed}, showed that the weighted set multi-cover problem is fixed parameter tractable, when parameterized by the size of the universe, implying that \prob\ is not NP-hard. However, while being theoretically significant, their approach is of limited practical relevance:
the exact polynomials are not analyzed, and polynomial exponents appear to be significantly large, and to the best of our knowledge, no implementation exists due to its complexity. This motivates the search for simpler exact algorithms and for efficient approximation algorithms. Recall that for the weighted set multi-cover problem (unbounded universe), the greedy algorithm achieves an $O(\log \ell)$-approximation. 
This directly implies an $O(1)$-approximation for \prob\ in $O(n \log n)$ time. 
Similarly, the standard LP-rounding algorithm achieves an $O\!\left(\log \sum_{g \in \Gee} Q_g\right)$-approximation, but it is outperformed by the greedy algorithm in both approximation guarantee and running time.
\footnote{The standard LP formulation for the weighted set multi-cover problem, shown in LP~\eqref{standardLP}, serves as the baseline for our experiments.}
Beyond these results, no approximation algorithms are specifically known for \prob.

This leaves several fundamental open questions:  
Is there a simple exact algorithm for \prob\ that runs in polynomial time and is practical to implement?  
Is there a practical approximation algorithm with theoretical guarantees that outperforms the greedy algorithm both in theory and in practice?  
More specifically, is there an efficient truly constant-factor approximation for \prob\ whose guarantee is independent of the input parameters $\ell$ and $n$?
Can such an algorithm also outperform the greedy and the standard LP-rounding method in practice, both in efficacy (total weight of the returned sets) and efficiency (runtime)?
We answer all of these questions affirmatively.

\vspace{-0.6em}
\paragraph{Our contributions}
Our contributions are summarized as follows.


    \begin{itemize}[leftmargin=20pt]
        \item In Section~\ref{sec:opt}, we discuss exact polynomial time algorithms for the \prob\ problem. As a warm-up, we show a simple exact Dynamic Programming algorithm that runs in $O(n^{\ell+1})$ time.
        \item In Section~\ref{sec:approx}, we study approximation algorithms for \prob. 
First, in Section~\ref{subsec:biglp} we design a $2$-approximation algorithm that runs in $O(\mathcal{L}(n))$ time, where $\mathcal{L}(n)$ is the time to solve an LP with $n$ variables and $O(n)$ constraints.\footnote{Using the fastest currently known algorithm~\cite{jiang2021faster}, we have $\mathcal{L}(n)=O(n^{2.38})$.}.

Next, in Section~\ref{subsec:smalllp}, we modify our previous approximation algorithm to improve its running time. For any given constant $\eps>0$, we design a $(2+\eps)$-approximation algorithm in $O(n\log n + \mathcal{L}(\log W))$ time, where $W=\tfrac{\sum_{t\in \Dee}w_t}{\min_{t\in \Dee}w_t}$.

\item In Section~\ref{sec:exp}, we implement our algorithms and compare them against the Greedy and the standard LP-rounding baselines on real and synthetic datasets. Our approximation algorithm consistently produces solutions of smaller weight than baselines: in real datasets, baselines' solutions are up to $1.3$ times larger, and in synthetic datasets up to $2$ times larger. Moreover, the $(2+\eps)$-approximation algorithm typically runs faster than the LP baseline, and competes with the greedy algorithm while being faster in many cases.
\end{itemize}

An additional desirable property in applications is to avoid
\emph{over-satisfying} the demands of individual items. That is,
while every demand must be met, selecting sets that exceed the
required coverage by a large margin may lead to ``overfitting'' on
certain items, thereby wasting resources or biasing the solution
towards a few categories. 
\revtwo{The notion of oversatisfying demands refers to situations
where, in order to satisfy the demands of certain groups, the solution
is forced to include many elements belonging to another group
whose demand is low. This is not directly controlled
by a boundary on the number of items but is rather a consequence
of the structure of the sets. An example is shown in Appendix~\ref{appndx:experiments}.}
Ideally, algorithms should satisfy each demand closely, without introducing much unnecessary coverage.
As we show in our experimental evaluation (Section~\ref{sec:exp}),
our algorithms naturally achieve this property, producing solutions
that satisfy demands without significantly overshooting them, in contrast to the
greedy baseline, which often overcovers certain items.

\vspace{-0.6em}
\section{Preliminaries}
\label{sec:prelim}
\vspace{-0.45em}
\newcommand{\dom}{\mathsf{dom}}
\subsection{Problem definition}
First, we formally define the Weighted Set Multi-Cover problem with Bounded Universe (\prob) and give some useful definitions that we will use throughout the paper.

\vspace{-0.2em}
\begin{definition}[Weighted Set Multi-Cover problem with Bounded Universe --- \prob]
\label{probdef}
    Let $\Gee$ be a universe of $\ell=O(1)$ items, where each item $g\in \Gee$ is associated with non-negative demand $\Qu_g$, and let $\Dee$ be a family of $n$ sets where each set $t\in \Dee$ is defined as a subset of $\Gee$, i.e., $t\subseteq \Gee$, and it is associated with a non-negative weight $w_t$. The goal is to select a family of sets $\sol\subseteq \Dee$ \revone{(each set in $\Dee$ may be selected at most once in $\sol$)} with minimum total weight such that every item $g\in \Gee$ is covered at least $\Qu_g$ times.
    Formally, that is:
\begin{align}
   \label{Probob1} \text{minimize:} \quad & \sum_{t \in \eh} w_t \\
   \label{Probcon1} \text{subject to:} \quad 
    & |\{t\in \sol\mid g\in t\}|\geq \Qu_g, \quad \forall g \in \Gee\\
   \label{Probcon2} & \eh \subseteq \dee.
\end{align}
\end{definition}
\vspace{-0.5em}
Throughout the paper, for simplicity, we always assume $Q_g \leq n$, for all $g \in \Gee$, since otherwise, the instance becomes infeasible.
\revone{Note that $\Dee$ may contain multiple distinct elements that represent the same underlying set, possibly with identical or different weights. Each such element is treated as a separate copy. In the solution $\sol$, at most one instance of each copy may be selected; that is, no copy can be used more than once, but different copies containing the same set of items can be used simultaneously.}
For a positive integer $z$, we define $[z]=\{1,\ldots, z\}$.
\revtwo{
We summarize all commonly used notation in Table~\ref{Table:Notation}.}


An optimum solution to the \prob\ problem is a set of tuples $\sol^*$ that minimizes the objective function~\eqref{Probob1} satisfying the constraints~\eqref{Probcon1},~\eqref{Probcon2}.

An algorithm for the \prob\ problem is called an $\alpha$-approximation algorithm (for any $\alpha\geq 1$), if it returns a set $\sol'$ that satisfies the constraints~\eqref{Probcon1},~\eqref{Probcon2} and $\sum_{t'\in \sol'}w_{t'}\leq \alpha\cdot \sum_{t^*\in \sol^*}w_{t^*}$. The set $\sol'$ is also called an $\alpha$-approximation solution for the \prob\ problem.

\begin{table}[h]
\small
\centering
\begin{envrevtwo}
 \begin{tabular}{c|c} 
 \hline
 Notation & Definition\\\hline
 \hline
 $\Gee$ & Universe of items\\ \hline
 $\ell$ & Size of the universe which is constant\\ \hline
 $\Qu_g$ & Demand of item $g\in \Gee$\\ \hline
  $\Dee$ & Family of sets (subsets of $\Gee$)\\ \hline
  $n$ & Number of sets in $\Dee$\\ \hline
  $w_t$ & Weight of set $t\in \Dee$\\ \hline
  $W$ & $\frac{\sum_{t\in \Dee}w_t}{\min_{t\in \Dee} w_t}$\\ \hline
  $\mathcal{L}(X)$ & LP solve time with 
$X$ variables and $O(X)$ constraints \\ \hline
   $[X]$ & $\{1,\ldots, X\}$\\ \hline
   $B[H]$ & All sets in $\Dee$ that contain the subset of items $H\subseteq \Gee$\\ \hline
   $\eps$ & Arbitrary positive constant\\ \hline
\end{tabular}
\caption{\revtwo{Table of Notations}}\label{Table:Notation}
\end{envrevtwo}
\vspace{-2em}
\end{table}

\subsection{Overview}
\label{subsec:overview}
\vspace{-0.2em}
In the next sections, we present a simple exact algorithm and two more involved approximation algorithms for the \prob\ problem. We briefly show the high-level ideas of our methods

For the exact algorithm, we discuss the natural dynamic programming algorithm for the \prob\ problem.

Next, we study approximation algorithms.
We begin by formulating \prob\ as an Integer Program and relaxing it to an optimization problem that minimizes the sum of non-decreasing, convex, piecewise linear functions subject to linear constraints. We show that this relaxation can be solved exactly by formulating it as a Linear Program (LP). We solve this LP (using an LP-solver) and design a novel rounding technique that yields a $2$-approximation algorithm for the \prob\ problem. The running time of our algorithm is $O(\mathcal{L}(n))$, where $\mathcal{L}(n)$ is the time to solve an LP with $n$ variables and $O(n)$ constraints.

The LP we obtain is related to the standard LP~\eqref{standardLP} for the weighted set multi-cover problem, although it has a different formulation. Furthermore, while randomized rounding on the standard LP~\eqref{standardLP} yields an $O\!\left(\log \sum_{g\in \Gee} Q_g\right)$-approximation, our novel rounding technique can also be applied to the standard formulation and improves the guarantee to~$2$. The motivation for formulating our problem as the minimization of a sum of piecewise linear functions under linear constraints will become clear in the following paragraph.

Our algorithm is the first truly constant-factor approximation (independent of $\ell$) for the \prob\ problem. However, the running time is super-quadratic in~$n$. To address this, we further refine our approach. Instead of directly solving the relaxed optimization problem with the original piecewise linear functions, we approximate these functions within a $(1+\eps)$ factor by constructing new non-decreasing, convex, piecewise linear functions with fewer linear pieces. The motivation for introducing this formulation, rather than working directly with the standard LP, is that it allows us to describe the piecewise functions more compactly, leading to smaller LPs and thus faster running time. The resulting LP is solved optimally, and rounding its solution yields a $(2+\eps)$-approximation algorithm for the \prob\ problem with running time $O(n\log n + \mathcal{L}(\log W))$, where $W=\tfrac{\sum_{t\in \Dee}w_t}{\min_{t\in \Dee}w_t}$.

\section{Exact Algorithms}
\label{sec:opt}
As previously mentioned, \cite{bredereck2020mixed} showed that the Weighted Set Multi-Cover problem is fixed parameter tractable when parameterized by the size of the universe of elements to cover. This implies that when the size of the universe is a constant, there exists an algorithm for solving the Weighted Set Multi-Cover problem in polynomial time with respect to the size of the input, or equivalently, there exists a polynomial time algorithm for solving \prob\ problem. Although theoretically intriguing, their approach lacks practical applicability, as the exact polynomials are not analyzed, and both the constant factors and polynomial exponents appear to be significantly large.

In this section, as a warm-up, we present a simple exact algorithm to our problem using Dynamic Programming (DP). 
While the DP approach is not always faster than the optimum algorithm in~\cite{bredereck2020mixed}, it is significantly simpler and easier to implement. In contrast, we are not aware of an implementation of the optimum algorithm in~\cite{bredereck2020mixed}.
The DP algorithm runs in time exponential to $|\Gee|$. Since $|\Gee|$ is constant, the running time is polynomial on $n$.

We aim to develop an optimal DP-based strategy that minimizes the total weight of the \prob\ problem. This process involves a sequence of iterative steps in which the algorithm evaluates the contribution of each set by updating the residual demands and the DP value at each step. In each iteration, the optimal solution records the minimum weight needed to fulfill the demands for the given state, up to the $i$-th set (arbitrary order). Each state is defined as the current level of demands satisfied for each item in $\Gee$.

For simplicity, we present the algorithm to compute the weight of the optimum solution. It is trivial how to restore the actual optimum family of subsets that achieve the optimum weight by storing the history of updates.
Let $t_1,\ldots, t_n$ be an arbitrary ordering of the sets in $\Dee$.
We define the table $\mathsf{DP}$ with dimensions $(n+1)\times (n+1) \times \ldots \times (n+1)$ ($\ell+1$ times).
Let $\mathsf{DP}[i][q_1]\ldots[q_\ell]$ be the weight of the optimum solution of the \prob\ problem considering the first $i$ sets with demands $q_1, \ldots, q_\ell$ on items $g_1,\ldots, g_\ell\in \Gee$, respectively.
We first initialize all cells in table $\mathsf{DP}$ to $\infty$. 
By definition, we set $\mathsf{DP}[i][0]\ldots[0]=0$ for every $i\in\{0\}\cup [n]$.
For every $i\in[n]$, $q_1\in\{0\}\cup[n],\ldots, q_\ell\in\{0\}\cup[n]$, it holds,
\vspace{-0.2em}
\begin{align}
\label{eq:DP}
\notag& \mathsf{DP}[i][q_1]\ldots[q_\ell]=\min\{\mathsf{DP}[i-1][q_1]\ldots[q_\ell], w_{t_i} + \\& \!\mathsf{DP}[i\!-\!1][\max\{0,q_1\!-\!\mathbbm{1}\![g_1\!\in\! t_i]\}]\ldots[\max\{0,q_\ell\!-\!\mathbbm{1}\![g_\ell\!\in\!t_i]\}]\},
\end{align}
\vspace{-0.2em}
where $\mathbbm{1}[\textsf{condition}]$ is the indicator function that returns $1$ if the $\textsf{condition}$ holds, and $0$ otherwise.
In other words, in every step, we decide whether we include the set $t_i$ in the current solution.
The optimum weight among the first $i$ sets and demands $q_1,\ldots, q_\ell$ is the minimum between i) Do not take set $t_i$: the optimum weight among the first $i-1$ sets and demands $q_1,\ldots, q_\ell$, and ii) Take set $t_i$: weight of set $t_i$ plus the optimum weight among the first $i-1$ sets and demands $q_1-\mathbbm{1}[g_1\in t_i],\ldots, q_\ell-\mathbbm{1}[g_\ell\in t_i]$.
After filling all cells in table $\mathsf{DP}$, 
we compute the optimum weight as $\min_{j_1\geq \Qu_{g_1}, \ldots, j_\ell\geq \Qu_{g_\ell}}\mathsf{DP}[n][j_1]\ldots[j_\ell]$. 
By tracing the optimum path in the DP algorithm, we get the family $\sol$ of the optimum selected sets.
\revone{The pseudocode of the algorithm is presented in Algorithm~\ref{alg:dp} in Appendix~\ref{appndx:DPpseudocode}.}

The correctness of the algorithm follows from Equation~\eqref{eq:DP}.
For the running time, we note that $\mathsf{DP}$ has $O(n^{|\Gee|+1})$ cells. In each cell we spend $O(\ell)=O(1)$ time to set the weight according to Equation~\ref{eq:DP}.
Overall, the running time of the algorithm is $O(n^{|\Gee|+1})$.
We conclude with the following theorem.
\vspace{-0.5em}
\begin{theorem}
    There exists an exact algorithm for the \prob\ problem that runs in $O(n^{|\Gee|+1})$ time.
\end{theorem}
\vspace{-0.5em}

\begin{envrevone}
    \begin{runningexample}
        \label{runex0}
        We show the execution of our DP algorithm using a simple example. Assume an instance of the $\prob$ problem with $\Gee=\{g_1,g_2\}$. There exists two sets $A_1=A_2=\{g_1\}$ with weights $1, 8$, respectively. The demand of $g_1$ is $Q_{g_1}=2$. There exists two sets $A_3=A_4=\{g_2\}$ with weights $2, 9$, respectively. The demand of $g_1$ is $Q_{g_2}=2$. Finally, there exists two sets $A_5=A_6=\{g_1,g_2\}$ with weights $3, 5$, respectively. Initially $\mathsf{DP}[0][0][0]=0$ and all other cells of table $\mathsf{DP}$ are set to $\infty$. We show the $\mathsf{DP}$ table after the execution of the DP algorithm in Table~\ref{table:ex0}. The notation $(x,y)\rightarrow z$ in row $A_i$ denotes $\mathsf{DP}[i][x][y]=z$.
        For example consider $\mathsf{DP}[5][1][2]=\min\{\mathsf{DP}[4][1][2], 3+\mathsf{DP}[4][0][1]\}$. We have that $\mathsf{DP}[4][1][2]=12$, and $3+\mathsf{DP}[4][0][1]=5$, so $\mathsf{DP}[5][1][2]=5$. Indeed, the sets $A_3, A_5$ have total weight $2+3=5$ and they cover $g_1$ once and $g_2$ twice.
        The optimum weight is $\mathsf{DP}[6][2][2]=6$ using the sets $A_1, A_3, A_5$ with total weight $6$.
    \end{runningexample}
\end{envrevone}
\begin{table}[ht]
\vspace{-1em}
\centering
\begin{envrevone}
\begin{tabular}{|c|c|}
\hline
$0$ & $(0,0)\rightarrow 0$ \\ \hline

$A_1$ & $(0,0)\rightarrow 0$, $(1,0)\rightarrow 1$ \\ \hline

$A_2$ & $(0,0)\rightarrow 0$, $(1,0)\rightarrow 1$, $(2,0)\rightarrow 9$ \\ \hline

\multirow{2}{*}{$A_3$} &
$(0,0)\rightarrow 0$, $(1,0)\rightarrow 1$, $(2,0)\rightarrow 9$, $(1,1)\rightarrow 3$, $(2,1)\rightarrow 11$,\\ 
& $(0,1)\rightarrow 2$ \\ \hline

\multirow{2}{*}{$A_4$} &
$(0,0)\rightarrow 0$, $(1,0)\rightarrow 1$, $(2,0)\rightarrow 9$, $(1,1)\rightarrow 3$, $(2,1)\rightarrow 11$,\\ 
& $(0,1)\rightarrow 2$,  $(0,2)\rightarrow 11$, $(1,2)\rightarrow 12$, $(2,2)\rightarrow 20$\\ \hline

\multirow{2}{*}{$A_5$} &
$(0,0)\rightarrow 0$, $(1,0)\rightarrow 1$, $(2,0)\rightarrow 4$, $(1,1)\rightarrow 3$, $(2,1)\rightarrow 4$,\\ 
& $(0,1)\rightarrow 2$,  $(0,2)\rightarrow 5$, $(1,2)\rightarrow 5$, $(2,2)\rightarrow 6$\\ \hline

\multirow{2}{*}{$A_6$} &
$(0,0)\rightarrow 0$, $(1,0)\rightarrow 1$, $(2,0)\rightarrow 4$, $(1,1)\rightarrow 3$, $(2,1)\rightarrow 4$,\\ 
& $(0,1)\rightarrow 2$,  $(0,2)\rightarrow 5$, $(1,2)\rightarrow 5$, $(2,2)\rightarrow 6$\\ \hline

\end{tabular}
\caption{\label{table:ex0}\revone{The non-infinite entries of the $\mathsf{DP}$ table for the $\prob$ instance in Example~\ref{runex0}. The notation $(x,y)\rightarrow z$ in row $A_i$ denotes $\mathsf{DP}[i][x][y]=z$.}}
\end{envrevone}
\vspace{-1em}
\end{table}

\vspace{-1em}
\section{Approximation Algorithms}
\label{sec:approx}

While the straightforward DP algorithm from the previous section runs in polynomial time, it quickly becomes inefficient when $|\Gee|$ is a large constant. Consequently, we turn our attention to developing efficient approximation algorithms and demonstrate their theoretical guarantees and practicality.
We first show a $2$-approximation algorithm that runs in $O(\mathcal{L}(n))$ time, where $\mathcal{L}(x)$ is the running time to solve a Linear Program with $O(x)$ variables and $O(x)$ constraints. Then, we show a $(2+\eps)$-approximation algorithm that runs in $O(n\cdot \log n+\mathcal{L}(\log(W))$ time, where $\eps$ is an arbitrarily small constant given by the user, say $0.01$, and $W=\frac{\sum_{t\in \Dee}w_t}{\min_{t\in \Dee}w_t}$.

\vspace{-0.5em}
\subsection{Polynomial 2-approximation algorithm}\label{subsec:biglp}

In this section, we show an LP-based approximation algorithm for the \prob\ problem. We first model \prob\ as an Integer Program (IP), and then relax the constraints to allow fractional solutions. We then show that the relaxed program can be written as a Linear Program (LP). We use an LP-solver to compute the optimum solution of the LP, and we use rounding to get a $2$-approximation solution.
While we could directly formulate the \prob\ problem as an Integer Linear Program (ILP) and then solve its LP relaxation, we use a different formulation that will allow us to improve the running time of the algorithm in the next sections. We note that the main goal of the algorithm discussed in this subsection is to provide a foundation for the efficient algorithm discussed in Subsection~\ref{subsec:smalllp}.


We first partition the sets from $\Dee$ in buckets based on the items they contain. For every subset $H\subseteq \Gee$, let $B(H)=\{t\in \Dee\mid t=H\}$, i.e., the bucket $B(H)$ contains all sets in $\Dee$ that contain exactly the subset of items $H$. Let $\mathcal{B}=\{B(H)\mid H\subseteq \Gee\}$. Since $|\Gee|=\groupsize$, there are at most $2^\groupsize$ different buckets, $|\mathcal{B}|=O(2^\groupsize)=O(1)$.
For every subset $H \subseteq \Gee$, we sort the sets in $B(H)$ by their weights in ascending order and we define the function $f_{B(H)}(x)$ that returns the sum of the weights of the first $x$ sets in $B(H)$, for every $x \in \{0, 1, \dots, |B(H)|\}$. Note that $f_{B(H)}(0) = 0$, for all $H \subseteq \Gee$. For simplicity, we refer to $f_{B(H)}(\cdot)$ as $f_H(\cdot)$. In other words, $f_H(x)$ returns the cost of choosing $x$ sets from $B(H)$ with the minimum weight.


With these definitions, we can now model the \prob\ as an IP. For each $H\subseteq \Gee$, let $x_H$ be the integer variable that denotes the number of sets selected from $B(H)$ in the returned solution.
The \prob\ problem is equivalent to the following Integer Program.

\vspace{-1.1em}
\begin{align}
    \label{obfunc}\text{minimize:} \quad & \sum_{H \subseteq \Gee} f_H(x_H) \\
   \label{con1} \text{subject to:}\quad& \sum_{H \subseteq \Gee, \, g \in 
    H}x_H \geq Q_g, \quad \forall g \in \Gee \\
   \label{con2} & 0 \leq x_H \leq |B(H)|, \quad \forall H \subseteq \Gee \\
   \label{con3} & \quad x_H \in \Ze, \quad \forall H\subseteq \Gee.
\end{align}

Constraints~\eqref{con1} ensure that the demands of all items in $\Gee$ are satisfied. Constraints~\eqref{con2} ensure that we select at most $|B(H)|$ sets from each subset $H\subseteq \Gee$.
If we select $\sigma_H\in \Ze$ sets from the bucket $B(H)$, it is always optimal to use the sets with minimum weight possible, and since $f_H(\sigma_H)$ returns the total weight of the lightest $\sigma_H$ sets from $B(H)$, the objective function in~\eqref{obfunc} returns the total weight of the solution. Therefore, the problem formulation~\eqref{obfunc} correctly captures the \prob\ problem.

Next, we relax the IP~\eqref{obfunc} to allow fractional solutions, and then show a rounding scheme to round the fractional solution into an integral solution while preserving the accuracy.
In the next paragraphs, let $\mathsf{opt}$ be the value of the objective function in the optimum solution of IP~\eqref{obfunc}.


\paragraph{Relaxation} We note that every function $f_H:\mathbb{Z}\rightarrow \Re$ is only defined on integral values $\{0, 1, \dots, B(H)\}$. For each $H \subseteq \Gee$, we define a real function $\hat{f}_H:\Re\rightarrow\Re$ to evaluate fractional solutions:

\vspace{-1.2em}
$$\hat{f}_H(x) =\left(f_H(z+1)\!-\!f_H(z)\right)\cdot (x-z)+f_H(z), \quad\text{if } z\!\leq\!x\!\leq\!z\!+\!1,  z\!\in\!\mathbb{Z}.$$


Using the function $\hat{f}_H(\cdot)$, we define the following optimization problem.

\vspace{-1em}
\begin{align}
    \label{LPobfunc}\text{minimize:} \quad & \sum_{H \subseteq \Gee} \hat{f}_H(x_H) \\
    \notag\text{subject to:}\quad& \sum_{H \subseteq \Gee, g \in 
    H}x_H \geq Q_g, \quad \forall g \in \Gee, \\
    \notag& 0 \leq x_H \leq |B(H)|, \quad \forall H \subseteq \Gee \\
    \notag& \quad x_H \in \Re, \quad \forall H\subseteq \Gee.
\end{align}
Let $\hat{\mathsf{opt}}$ be the value of the objective function in the optimum solution of Problem~\eqref{LPobfunc}.

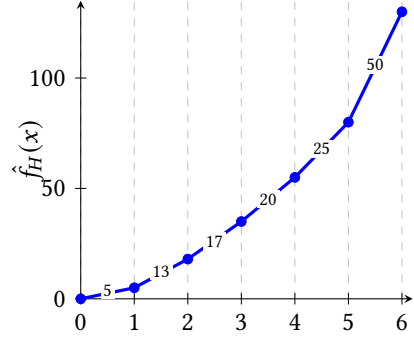
\begin{figure}[t]
\centering

\begin{minipage}{0.42\textwidth}
\centering

\small
\renewcommand{\arraystretch}{1.15}
\begin{tabular}{@{} c | c | c @{}}
\hline
\textbf{Expert \#} & \textbf{Skills} & \textbf{Hourly Rate} \\
\hline
1 & (Legal, Medical) & 5\$ \\ \hline
2 & (Legal, Medical) & 13\$ \\ \hline
3 & (Legal, Medical) & 17\$ \\ \hline
4 & (Legal, Medical) & 20\$ \\ \hline
5 & (Legal, Medical) & 25\$ \\ \hline
6 & (Legal, Medical) & 50\$ \\ \hline
\hline
\end{tabular}

\end{minipage}
\hfill
\begin{minipage}{0.55\textwidth}
\centering

\begin{tikzpicture}[scale=1.0]
  \begin{axis}[
      width=6cm, height=5.5cm,
      axis lines=left,
      xmin=0, xmax=6.2,
      ymin=0, ymax=135,
      xtick={0,1,2,3,4,5,6},
      xmajorgrids,
      grid style={dashed, very thin},
      xlabel={$x$}, ylabel={$\hat{f}_H(x)$},
      ylabel style={yshift=-15pt},
      tick style={black},
      clip=false
    ]

    \addplot[very thick, blue]
      coordinates {
        (0,0)
        (1,5)
        (2,18)
        (3,35)
        (4,55)
        (5,80)
        (6,130)
      };

    \addplot[only marks, mark=*, mark size=1.8pt, blue] 
      coordinates {
        (0,0)
        (1,5)
        (2,18)
        (3,35)
        (4,55)
        (5,80)
        (6,130)
      };

    \node[font=\scriptsize, fill=white, inner sep=1pt] at (axis cs:0.5,4) {5};
    \node[font=\scriptsize, fill=white, inner sep=1pt] at (axis cs:1.5,12.5) {13};
    \node[font=\scriptsize, fill=white, inner sep=1pt] at (axis cs:2.5,26) {17};
    \node[font=\scriptsize, fill=white, inner sep=1pt] at (axis cs:3.5,45) {20};
    \node[font=\scriptsize, fill=white, inner sep=1pt] at (axis cs:4.5,68) {25};
    \node[font=\scriptsize, fill=white, inner sep=1pt] at (axis cs:5.5,106) {50};

  \end{axis}
\end{tikzpicture}

\end{minipage}

\vspace{-2em}
\caption{\small Sample of experts with the same skill set. In this example, $H = \{Legal, Medical\}$, and the rate is considered as the weight. The figure shows the piecewise linear function $\hat{f}_H(x)$, based on the shown table.}
\label{fig:pl-table-stacked}
\vspace{-1.2em}

\end{figure}

An example of $\hat{f}_H(\cdot)$ is shown in Figure~\ref{fig:pl-table-stacked}. As we will show later, for every $H\subseteq \Gee$, function $\hat{f}_H(\cdot)$ is non-decreasing, convex and piecewise linear.
Hence, we argue that Problem~\eqref{LPobfunc} is an abbreviated representation of a pure LP.
There are ways to translate Problem~\ref{LPobfunc} into a pure Linear Program (see for example the $\lambda$ method from \cite{wolsey1999integer, lee2001polyhedral}). More specifically, for each $H\subseteq \Gee$, we define the binary variable $x_{H,i}$ which is $1$ if the set with the $i$-th smallest weight in $B(H)$ is selected in the solution, and $0$ otherwise. Let $w_{H,i}$ be the weight of the $i$-th smallest weight in $B(H)$.
We define the following LP.
    \begin{align}
    \label{FinalLPobfunc}\text{minimize:} \quad & \sum_{H \subseteq \Gee} \quad\sum_{i\in [|B(H)|]}w_{H,i}\cdot x_{H,i} \\
    \label{LPC1}\text{subject to:}\quad& \sum_{H \subseteq \Gee, \, g \in 
    H}\sum_{i\in[|B(H)|]}x_{H,i} \geq Q_g, \quad \forall g \in \Gee, \\
    \label{LPC2}& 0 \leq \sum_{i\in[|B(H)|]}x_{H,i} \leq |B(H)|, \quad \forall H \subseteq \Gee \\
    \label{LPC3}& \quad 0\leq x_{H,i}\leq 1, \quad \forall H\subseteq \Gee, i\in[|B(H)|].
\end{align}
We note that for every $H\subseteq \Gee$, there are $|B(H)|$ variables, and every set in $\Dee$ belongs to exactly one bucket $B(H)$. Hence the total number of variables are $n$. Since $|\Gee|=\ell=O(1)$ there are $O(1)$ non-trivial type~\eqref{LPC1} constraints. Furthermore, there are $O(2^\ell)=O(1)$ subsets of $\Gee$ so there are $O(1)$ non-trivial type~\eqref{LPC2} constraints. Finally, there are $O(n)$ trivial type~\eqref{LPC3} constraints.

From now on, given a set of real numbers $\{y_H \mid H \subseteq \Gee\}$, we use the notation $\vv{y}_{H \subseteq \Gee} \in \Re^{2^\ell}$ to denote the vector whose coordinate corresponding to a subset $H \subseteq \Gee$ takes the value $y_H$.

\revtwofour{In Appendix~\ref{appndx:lem1},} we show some properties of Problems~\eqref{obfunc},~\eqref{LPobfunc} and~\eqref{FinalLPobfunc}.

\begin{lemma}
\label{lem:1}
(i) $\hat{\mathsf{opt}}\leq \mathsf{opt}$. (ii) For every $H\subseteq \Gee$, the function $\hat{f}_H(\cdot)$ is a non-decreasing piecewise linear and convex function. (iii) Problem~\eqref{LPobfunc} is equivalent to LP~\eqref{FinalLPobfunc}.
\end{lemma}


Using the observations from Lemma~\ref{lem:1}, we describe the algorithm as follows. We first form the LP~\eqref{FinalLPobfunc}. Using an LP solver, we get the optimum solution.
For each $H\subseteq \Gee$ and $i\in[|B(H)|]$, let $\hat{x}_{H,i}$ be the value of variable $x_{H,i}$ in the optimum solution. From the proof of Lemma~\ref{lem:1}, LP~\eqref{FinalLPobfunc} and Problem~\eqref{LPobfunc} are equivalent, so without loss of generality, assume that for each $H\subseteq \Gee$,  $\hat{x}_{H}$ is the value of variable $x_{H}$ in the optimum solution of Problem~\eqref{LPobfunc}.

\begin{envrevtwofour}
\begin{runningexample}
\label{exAlg}
    Assume an instance of the \prob\ problem with $\Gee=\{g_1,g_2\}$. For $H_1=\{g_1\}\in\Gee$ there exists three sets $\Gamma_1=\Gamma_2=\Gamma_3=\{g_1\}$ in $\Dee$ with weights $1, 3, 4$, respectively. The demand of $g_1$ is $\Qu_{g_1}=3$. For $H_2=\{g_2\}\in \Gee$ there exists two sets $\Delta_1=\Delta_2=\{g_2\}$ in $\Dee$ with weights $2, 8$ respectively. The demand of $g_2$ is $\Qu_{g_2}=2$.
    Finally, for $H_{1,2}=\{g_1,g_2\}\in \Gee$ there exist three sets $E_1=E_2=E_3=\{g_1,g_2\}$ in $\Dee$ with weights $4, 5, 6$, respectively. 

    We first present the formulation of Problem~\ref{LPobfunc} under the setting described above.
    In the toy example, Problem~\ref{LPobfunc} is defined over $3$ variables, i.e., one for every subset of $\Gee$. Namely, we have the variables $x_{H_1}$, $x_{H_2}$, and $x_{H_{1,2}}$.
    
    In order to compute the objective functions we should define $\hat{f}_H(x)$ for every $H\subseteq \Gee$.
    For $H_1=\{g_1\}\subset \Gee$, $f_{H_1}(0)=0$, $f_{H_1}(1)=1$, $f_{H_1}(2)=4$, and $f_{H_1}(3)=8$, so the function $\hat{f}_{H_1}(x)$ is defined as follows:
    i) If $0\leq x<1$, then $\hat{f}_{H_1}(x)=x$, ii) if $1\leq x<2$, then $\hat{f}_{H_1}(x)=3(x-1)+1=3x-2$, and iii) if $2\leq x\leq 3$, then $\hat{f}_{H_1}(x)=4(x-2)+4=4x-4$.
    For $H_2=\{g_2\}\subset \Gee$, then $f_{H_2}(0)=0$, $f_{H_2}(1)=2$, and $f_{H_2}(2)=10$, so the function $\hat{f}_{H_2}(x)$ is defined as follows:
    i) If $0\leq x<1$, then $\hat{f}_{H_2}(x)=2x$, and ii) if $1\leq x\leq 2$, then $\hat{f}_{H_2}(x)=8(x-1)+2=8x-6$.
    For $H_{1,2}=\{g_1,g_2\}= \Gee$, then $f_{H_{1,2}}(0)=0$, $f_{H_{1,2}}(1)=4$, $f_{H_{1,2}}(2)=9$, and $f_{H_{1,2}}(3)=15$, so the function $\hat{f}_{H_{1,2}}(x)$ is defined as follows:
    i) If $0\leq x<1$, then $\hat{f}_{H_{1,2}}(x)=4x$, ii) if $1\leq x<2$, then $\hat{f}_{H_{1,2}}(x)=5(x-1)+4=5x-1$, and iii) if $2\leq x\leq 3$, then $\hat{f}_{H_{1,2}}(x)=6(x-2)+9=6x-3$.

    Next, we show the constraints of Problem~\ref{LPobfunc}.
    For $g_1\in \Gee$, we must satisfy $x_{H_1}+x_{H_{1,2}}\geq 3$, while for $g_2\in \Gee$ we must satisfy $x_{H_2}+x_{H_{1,2}}\geq 2$. Finally, we must have $0\leq x_{H_1}\leq 3$, $0\leq x_{H_2}\leq 2$, and $0\leq x_{H_{1,2}}\leq 3$.

    The first phase of our algorithm computes the optimum solution to this optimization problem using an \textsf{LP}-solver.
    
\end{runningexample}

\end{envrevtwofour}

\paragraph{Rounding scheme} 
Next, we present the rounding scheme of our algorithm. While our rounding algorithm is theoretically efficient, as shown in Section~\ref{sec:exp}, in practice the LP formulated in the previous paragraphs almost always produces small fractional residuals, so our rounding scheme remains efficient in practical settings.

Let $\sol$ be an initially empty set denoting the family of sets we choose in our solution. We define $\bar{x}_H = \lfloor \hat{x}_H \rfloor$ for all $H \subseteq \Gee$, and from $B(H)$, we add the $\bar{x}_H$ sets with the smallest weight to the solution set $\sol$. By rounding down the values $\hat{x}_H$, some demands may remain unsatisfied. Next, we describe how we add more sets to the set $\sol$ to make sure all the demands are satisfied.

Let $\bar{B}(H) = B(H) - \sol$, be the family of unused sets in each bucket for all $H \subseteq \Gee$, and let $r = \left\lceil\sum_{H \subseteq \Gee}(\hat{x}_H - \bar{x}_H)\right\rceil$. From each bucket $\bar{B}(H)$, we add at most $r$ more sets to the solution set. To find the sets we add to the solution, we brute force all the possible ways of choosing between $0$ to $r$ sets from each bucket $\bar{B}_H$, and among the possible family of sets that satisfy all the demands, we choose the sets having the minimum total weight and add them to the solution set $\sol$. We will later show why it is sufficient to consider at most $r$ additional sets from each bucket to satisfy all the remaining demands. The pseudo-code of our algorithm is shown in Algorithm~\ref{alg:lprounding}.

\revtwofour{In Appendix~\ref{appndx:thm2},} we prove the following theorem.
\begin{theorem}
\label{thm:1}
    Given an instance $(\Dee, \Gee)$ of the \prob\ problem with $|\Dee|=n$ and $|\Gee|=O(1)$, there exists a $2$-approximation algorithm for the \prob\ problem that runs in $O(\mathcal{L}(n))$ time.
\end{theorem}

\begin{envrevtwofour}
    \noindent\textbf{Running Example~\ref{exAlg} (cont.)}
    Next, we show an execution of our rounding scheme using the toy example we defined in Running Example~\ref{exAlg}. For simplicity, assume that the optimum solution found by solving Problem~\ref{LPobfunc} is $\hat{x}_{H_1}=1.5$, $\hat{x}_{H_2}=0.5$, and $\hat{x}_{H_{1,2}}=1.5$. We note that this is not the optimum solution, however we will use this feasible solution to demonstrate our rounding scheme. We have $\bar{x}_{H_1}=1$, $\bar{x}_{H_2}=1$, and $\bar{x}_{H_{1,2}}=0$. In the set $\mathcal{H}$ we add the set $H_1$ with the smallest weight, which is $\Gamma_1$ and the set $H_{1,2}$ with the smallest weight, which is $E_1$. Hence $\mathcal{H}=\{\Gamma_1,E_1\}$. We also have $\bar{B}(H_1)=\{\Gamma_2,\Gamma_3\}$, $\bar{B}(H_2)=\{\Delta_1,\Delta_2\}$, and $\bar{B}(H_{1,2})=\{E_2,E_3\}$. By the definition of $r$ we get $r=\lceil 1.5\rceil=2$. Next, we brute force all combinations choosing between $0$ and $r=2$ sets from each bucket $\bar{B}(H_1)$, $\bar{B}(H_2)$, $\bar{B}(H_{1,2})$, and among the family of sets that satisfy the demands we choose the set with the minimum total weight. For example, the family that contains the pre-selected sets in $\mathcal{H}$ along with $2$ sets from $\bar{B}(H_1)$ (sets $\Gamma_2, \Gamma_3$), and no set from $\bar{B}(H_{1,2})$ and $\bar{B}(H_2)$ does not satisfy the  demands. Notice that while the demand of $g_2$ is $\Qu_{g_2}=2$ there is only one set in $\mathcal{H}\cup \{\Gamma_2\}\cup\{\Gamma_3\} =\{\Gamma_1, \Gamma_2, \Gamma_3, E_1\}$ that contains $g_2$. 
    On the other hand, the family that contains the pre-selected sets in $\mathcal{H}$ along with $1$ set from $\bar{B}(H_1)$ (set $\Gamma_2$), $1$ set from $\bar{B}(H_{1,2})$ (set $E_2$), and no set from $\bar{B}(H_2)$ satisfy the demands. Notice that in $\mathcal{H}\cup \{\Gamma_2\}\cup\{E_2\} =\{\Gamma_1,\Gamma_2, E_1, E_2\}$ there exist $4>\Qu_{g_1}$ sets that contain $g_1$ and $2\geq \Qu_{g_2}$ sets that contain $g_2$. However, the total weight of the selected family of sets $\{\Gamma_1,\Gamma_2, E_1, E_2\}$ is $13$, which is not the minimum. By trying all combinations, we find the the optimum family is the one that contains the pre-selected sets in $\mathcal{H}$ along with $1$ set from $\bar{B}(H_1)$ (set $\Gamma_2$), $1$ set from $\bar{B}(H_2)$ (set $\Delta_1$), and no set from $\bar{B}(H_{1,2})$, i.e, the family $\mathcal{H}\cup \{\Gamma_2\}\cup\{\Delta_1\} =\{\Gamma_1, \Gamma_2, E_1, \Delta_1\}$ satisfies all the demands and its total weight is $10$.
\end{envrevtwofour}

\begin{algorithm}[t]
    \caption{LP-based $2$-approximation Algorithm}
    \label{alg:lprounding}
    \begin{algorithmic}[1]
        \Require $\Dee, \Gee, \{Q_g\mid g\in \Gee\}$
        \Ensure Family of selected sets $\sol$
        \State Form and solve LP~\eqref{FinalLPobfunc}
        \State Convert the solution from LP~\eqref{FinalLPobfunc} to a solution for Problem~\ref{LPobfunc}
        \State For each $H\subseteq \Gee$, let $\bar{x}_H$ be the value of the variable $x_{H}$ in the optimum solution of Problem~\ref{LPobfunc}
        \State $\sol \gets \emptyset$
        \ForAll{$H \subseteq \Gee$}
            \State $\bar{x}_H \gets \lfloor \hat{x}_H \rfloor$
            \State Add $\bar{x}_H$ sets from $B(H)$ with minimum weight to $\sol$
        \EndFor
        \ForAll{$H \subseteq \Gee$}
            \State $\bar{B}(H) \gets B(H) \setminus \sol$
        \EndFor
        \State $r \gets \left\lceil\sum_{H \subseteq \Gee}(\hat{x}_H - \bar{x}_H)\right\rceil$
        \State $W' \gets \infty$, $\sol' \gets \emptyset$
        \ForAll{integer vectors $\vv{y}_{H \subseteq \Gee}$ such that \\ $0 \leq y_H \leq \min\{r, |\bar{B}(H)|\}$}
            \State $S \gets \emptyset$
            \ForAll{$H \subseteq \Gee$}
                \State Add $y_H$ sets with minimum weight from $\bar{B}(H)$ to $S$
            \EndFor
            \ForAll{$g \in \Gee$}
                \State $q_g \gets \sum_{H \ni g} \bar{x}_H + y_H$
            \EndFor
            \If{$q_g \geq Q_g$ for all $g \in \Gee$}
                \State $w \gets$ total weight of $S$
                \If{$w < W'$}
                    \State $W' \gets w$, $\sol' \gets S$
                \EndIf
            \EndIf
        \EndFor
        \State $\sol \gets \sol \cup \sol'$
        \State \Return $\sol$
    \end{algorithmic}
\end{algorithm}

\subsection{Fast $(2 + \epsilon)$-approximation}\label{subsec:smalllp}
In this section, we show how we can use a slightly different LP formulation to achieve a much faster algorithm. As discussed, the running time bottleneck of Algorithm~\ref{alg:lprounding} is solving the LP~\eqref{FinalLPobfunc}. The main complexity of the LP formulation comes from the number of variables, equivalently, the number of linear pieces in functions $\hat{f}_H(\cdot)$ from Problem~\eqref{LPobfunc}. As discussed in Subsection~\ref{subsec:biglp}, the total number of different linear pieces across all $H \subseteq \Gee$ is $O(n)$. This makes the size of LP~\eqref{FinalLPobfunc} to be $O(n)$, and solving this program takes super-quadratic time in terms of the number of sets $n$.
The idea of the algorithm we discuss in this section is to describe each function $\hat{f}_H(\cdot)$ with a smaller number of linear pieces bounding the error. This will allow us to significantly reduce the size of the LP instance and solve it efficiently.


Recall that every function $\hat{f}_H$ is convex, non-decreasing, and piecewise linear. We start by describing a more general result of approximating any continuous, non-decreasing, convex function (not necessarily piecewise linear) with a continuous, non-decreasing, convex, and piecewise linear function. Then, we use this result to approximate the piecewise linear function $\hat{f}_H(\cdot)$ with another continuous, non-decreasing, convex, and piecewise linear function having less number of linear pieces.

\begin{figure}[t]
\centering
\begin{tikzpicture}[scale=0.8]]
  \begin{axis}[
      width=6cm, height=6cm,
      axis lines=left,
      xmin=0, xmax=5.2,
      ymin=0, ymax=27,
      xtick={0,1,2,3,4, 5},
      xlabel={$x$}, ylabel={$\ef(x), \gi(x)$},
      ylabel style={yshift=-15pt},
      xlabel style={xshift=70pt, yshift = 15pt},
      tick style={black},
      clip=false,
      legend style={draw=none, font=\scriptsize, at={(0.97,0.03)}, anchor=south east}
    ]

    \addplot[blue, domain=0:5, samples=200] {x^2 + 1};
    \addlegendentry{$\ef(x)=x^2 + 1$}

    \addplot[red, dashed]
      coordinates {
        (0,1)
        (1.73,4)
        (3.87,16)
        (5, 26)
      };
    \addlegendentry{$\gi(x), \epsilon = 3$}

    \addplot[only marks, mark=*, mark size=1.2pt, red]
      coordinates {
        (0,1)
        (1.73,4)
        (3.87,16)
        (5, 26)
      };

  \end{axis}
\end{tikzpicture}
\vspace{-1em}
\caption{\small Approximation of the $\ef(x) = x^2 + 1$ curve by a piecewise linear function $\gi(\cdot)$, generated by the described algorithm, where $\epsilon = 3$. The function $\gi(\cdot)$ approximates the values of the function $\ef(\cdot)$ within a factor of $4$.
}
\label{fig:convex-approx}
\vspace{-1.5em}
\end{figure}
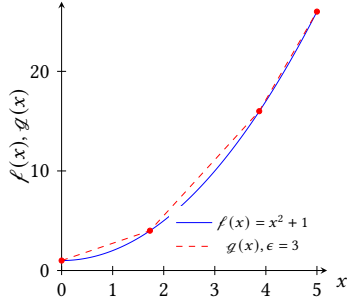

\paragraph{Approximating functions by piecewise linear functions}
Let $\ef: [\alpha, \beta] \to [\gamma, \delta]$ be any continuous, convex, and non-decreasing function such that $0 < \alpha \leq \beta$ and $0 < \gamma \leq \delta$, and let $\eps>0$ be any constant. We describe an algorithm to approximate the function $\ef(\cdot)$ by a piecewise linear function $\gi(\cdot)$, such that $\ef(x) \leq \gi(x) \leq (1 + \epsilon)\ef(x)$, for any value of $\alpha \leq x \leq \beta$.

We start by $x_1 = \alpha$ and get $\ef(x_1)$. Then, we find the largest value $x'_1$, such that $\ef(x'_1) \leq (1+\epsilon)\ef(x_1)$ and $x'_1 \leq \beta$. Since the function $\ef(\cdot)$ is non-decreasing and continuous it is sufficient to only consider $x_1'>x_1$.
Let $\mathsf{LargestVal}_{\ef}(x_1, (1+\eps)\ef(x_1))$ be a procedure over the function $\ef$ that returns such value $x_1'$.
Next, we set $x_2 = x'_1$, and recursively repeat the same process using $x_2$ instead of $x_1$, calling $\mathsf{LargestVal}_{\ef}(x_2, (1+\eps)\ef(x_2))$.
More generally, at the $i$'th step, we find the largest value $x'_i \leq \beta$, such that $\ef(x'_i) \leq (1+\epsilon)\ef(x_i)$, and then set $x_{i + 1} = x'_i$. 
We continue this process until $x_s = \beta$ for some step $s$.
Finally, we set the function $\gi(\cdot)$ to contain $s - 1$ linear pieces where the $i$-th piece of $\gi(\cdot)$ is the line segment connecting the points $(x_i, \ef(x_i))$ and $(x_{i+1}, \ef(x_{i+1}))$, for every $i\in[s-1]$. An illustration is shown in Figure~\ref{fig:convex-approx}. We prove the next lemma in Appendix~\ref{appndx:sec4b}.

\begin{lemma}\label{lem:approx}
    Given a parameter $\eps>0$, any continuous, convex and non-decreasing function $\ef(\cdot): [\alpha, \beta] \to [\gamma, \delta]$, where $0<\alpha\leq \beta$ and $0<\gamma\leq \delta$, can be approximated by a continuous, convex non-decreasing piecewise linear function $\gi(\cdot): [\alpha, \beta] \to [\gamma, \delta]$ having $O(\log(\frac{\delta}{\gamma})\log^{-1}(1+\eps))$ pieces, such that for every $\alpha \leq x \leq \beta$, it holds that $\ef(x) \leq \gi(x) \leq (1 + \epsilon) \ef(x)$. Moreover, the function $\gi(\cdot)$ can be constructed in $O(\log(\frac{\delta}{\gamma})\cdot\log^{-1}(1+\eps)\cdot T_{\ef}^>\cdot T_{\ef}^=)$ time, where $T_{\ef}^=$ is the time needed to calculate $\ef(x)$, for any $\alpha \le x \leq \beta$, and $T_{\ef}^>$ is the time needed to compute the smallest value $x'\geq \bar{x}$ such that $\ef(x')\leq (1+\eps)\ef(\bar{x})$, for any  values $\bar{x}, x'$ satisfying $\alpha\leq \bar{x}\leq x'\leq\beta$.
\end{lemma}

\paragraph{Our algorithm}
We move back to our original problem. For all $H \subseteq \Gee$, we keep the first linear piece, while for the rest pieces we use Lemma~\ref{lem:approx} to approximate $\hat{f}_H(\cdot)$. Let $\hat{g}_H(\cdot)$ be the resulting piecewise linear function we obtain from Lemma~\ref{lem:approx} along with the first linear piece from $\hat{f}_H(\cdot)$.
By Lemma~\ref{lem:approx}, for all $x \in [0,|B(H)|]$, it holds that $\hat{f}_H(x) \leq \hat{g}_H(x) \leq (1 + \epsilon)\hat{f}_H(x)$, where $\eps>0$ is a constant chosen by the user. We then define the optimization problem, similarly to Problem~\ref{LPobfunc} as shown in Subsection~\ref{subsec:biglp}, but we use the functions $\hat{g}_H(\cdot)$ instead of $\hat{f}_H(\cdot)$. 
\begin{align}
    \label{smallProb}\text{minimize:} \quad & \sum_{H \subseteq \Gee} \hat{g}_H(x_H) \\
    \notag\text{subject to:} \quad &\sum_{H \subseteq \Gee, \, g \in 
    H}x_H \geq Q_g, \quad \forall g \in \Gee, \\
   \notag & 0 \leq x_H \leq |B(H)|, \quad \forall H \subseteq \Gee \\
   \notag & \quad x_H \in \Re.
\end{align}
Using our machinery from Subsection~\ref{subsec:biglp} along with the proof of Lemma~\ref{lem:1}, we can formulate Problem~\eqref{smallProb} as an LP.
We then solve the new LP to get the fractional solution. Finally, we round its solution using the same rounding scheme as in Algorithm~\ref{alg:lprounding}.

\begin{envrevtwofour}
    \noindent\textbf{Running Example~\ref{exAlg} (cont.)} We continue Example~\ref{exAlg} executing our new algorithm. For every subset $H=\{H_1,H_2,H_{1,2}\}$, we construct the function $\hat{g}_H(x)$. For simplicity we focus on $H_1=\{g_1\}$, however the same approach works for the rest subsets. Recall the definition of $\hat{f}_{H_1}(x)$ which is a piecewise linear function with three pieces as shown in Figure~\ref{fig:piecewiseEstimation}. For $x\in [0,1]$ by the description of our algorithm we have $\hat{g}_{H_1}(x)=\hat{f}_{H_1}(x)=x$. Then for the rest of the pieces we execute the algorithm from Lemma~\ref{lem:approx}. Assume that $\eps=7$. We have $\alpha=1, \beta=3, \gamma=1,$ and $\delta=8$. Starting from $x_1=1$ we compute the largest $x_1'\leq 3$ such that $\hat{f}_{H_1}(x_1')\leq (1+\eps)\hat{f}_{H_1}(x_1)=8\cdot\hat{f}_{H_1}(x_1)$ and $x_1'\leq 3$. We observe that $x_1'=3$ so $\hat{g}_{H_1}(x)$ for $x\in[1,8]$ is defined as the linear segment that connects $(1,1)$ and $(3,8)$, as shown with the dashed segment in Figure~\ref{fig:piecewiseEstimation}. Repeating the same procedure for every subset $H$ of $\Gee$ we define all functions $\hat{g}_H(x)$. 
    The constraints of Problem~\ref{smallProb} are the same with the constraints of Problem~\ref{LPobfunc}.
\end{envrevtwofour}

\begin{figure}[t]
    \centering
    \includegraphics[width=0.45\linewidth]{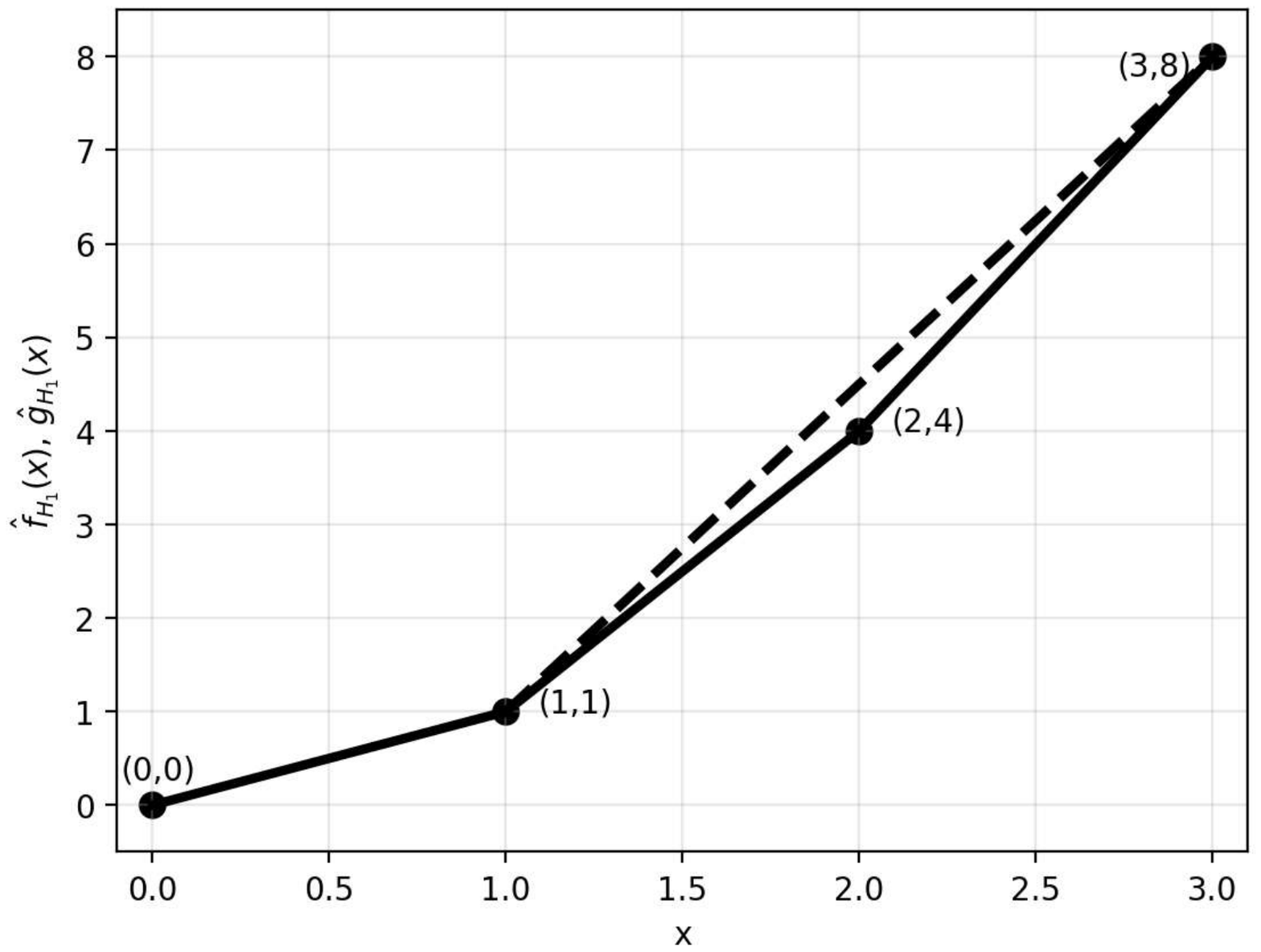}
    \vspace{-2em}
    \caption{\revtwofour{The function $\hat{f}_{H_1}(x)$ is a piecewise linear function consisting of three solid segments. For $\eps = 7$, the function $\hat{g}_{H_1}(x)$ is a piecewise linear function with two segments: for $x \in [0,1]$, $\hat{g}_{H_1}(x) = \hat{f}_{H_1}(x)$, while for $x \in [1,3]$, $\hat{g}_{H_1}(x)$ is given by the dashed linear segment connecting $(1,1)$ and $(3,8)$.}}
    \label{fig:piecewiseEstimation}
\end{figure}

\paragraph{Analysis}
For every $H\subseteq \Gee$, let $\tilde{x}_H$ be the integer value of the variable $x_H$ that is returned by the rounding scheme.
Recall that $\mathsf{opt}$ is the value of the optimum solution in IP~\eqref{obfunc}.
The next lemmas are shown in Appendix~\ref{appndx:sec4b}.
\begin{lemma}
\label{lem:ApproxFactor}
    $\sum_{H\subseteq \Gee}f_H(\tilde{x}_H)\leq 2\cdot(1+\eps)\cdot\opt$.
\end{lemma}

Notice that if we run the algorithm from Lemma~\ref{lem:approx}, setting $\eps\leftarrow \eps/2$, we will get directly $\sum_{H\subseteq \Gee}f_H(\tilde{x}_H)\leq (2+\eps)\opt$.

\begin{lemma}
\label{lem:runtimefaster}
    \revone{For any arbitrary constant $\eps>0$,} the algorithm runs in $O(n\log n + \mathcal{L}(\log (W)))$ time, where $W=\frac{\sum_{t\in \Dee}w_t}{\min_{t\in \Dee}w_t}$.
\end{lemma}

Putting everything together, we conclude to the next theorem.

\begin{theorem}
\label{thm:2}
    Given an instance $(\Dee, \Gee)$ of the \prob\ problem with $|\Dee|=n$ and $|\Gee|=O(1)$, and an arbitrary constant $\eps>0$, there exists a $(2+\eps)$-approximation algorithm for the \prob\ problem that runs in $O(n\log n+\mathcal{L}(\log(W)))$ time, where $W=\tfrac{\sum_{t\in \Dee}w_t}{\min_{t\in \Dee}w_t}$.
\end{theorem}

\vspace{-1.3em}
\section{Evaluation}
\label{sec:exp}
\vspace{-0.2em}
\subsection{Experiment Plan}
We aim to evaluate our algorithms with respect to both effectiveness and efficiency. The effectiveness experiments assess the quality of the solutions obtained, specifically in terms of total solution weight, while ensuring that the required demands are satisfied. We further assess the extent to which the obtained solution deviates from the specified demand requirements.
The efficiency experiments, on the other hand, measure the computational performance of each algorithm, i.e., the time required to obtain a solution. For each experiment, we evaluate the impact of varying the following parameters:  
number of sets in the dataset ($n$), number of items of interest ($\ell$), \revtwo{and the distribution and magnitude of the demands ($\Qu$)}.
In our experiments, we evaluate the performance of our $2$-approximation and $(2+\epsilon)$-approximation algorithms, and compare them against two baselines: 1) the greedy approach and 2) RR-LP, the classical LP formulation for covering problems followed by a randomized rounding scheme. \revone{We also compare our main results against the exact DP algorithm in Appendix~\ref{appndx:experiments}}. Across all settings, each experiment is repeated 30 times, and the reported results correspond to the average values. 

\vspace{-0.8em}
\subsection{Experiment Setting}
\subsubsection{Setup}
The experiments were conducted on an Apple M2 Pro processor, 16 GB memory, running macOS. The algorithms were implemented in Python 3.12. We use \textsf{Gurobi}~\cite{gurobi} LP solver in the implementation of our algorithms.\footnote{Code available at: https://anonymous.4open.science/r/WSMC-BU-7B4B.}
\subsubsection{Datasets}
To evaluate our algorithms in practical settings, we employ the following four real-world datasets:
\begin{itemize}[leftmargin=*]
    \item \textsf{Census~\cite{us_census_1990}:} This dataset contains a one percent sample of the {\em Public Use Microdata Samples} person records derived from the complete 1990 census dataset. It includes around 2.5 million records with 68 categorical attributes. We preprocess the dataset by transforming each attribute value into a distinct binary attribute. 
    Additionally, we assign a randomly generated weight between 1 and 1000 to each tuple.
    \item \textsf{Music~\cite{magellandata}:} This dataset is part of the Magellan data repository and comprises approximately 56,000 songs from the Amazon Music dataset. Each song is associated with a set-valued attribute called Genre. We used the Genre attribute to represent the sets in our problem. The duration of each song, measured in seconds, is used as the weight associated with the corresponding tuple. Notably, these weights are fairly similar across entries.
    \item \textsf{Stack Overflow~\cite{stackoverflow2024survey}:} This dataset is derived from the 2024 Stack Overflow Annual Survey of Developers and contains rich categorical information on technologies, roles, education, and demographics. We selected 15 attributes from the dataset and transformed them into 114 binary attributes. Each tuple is additionally assigned a randomly generated weight between 1 and 1000.

    \item \textsf{Yelp~\cite{yelp2025dataset}:} The Yelp business dataset contains approximately 150,000 business entries. From the Categories attribute, we extract items and generate distinct binary attributes corresponding to the unique category values. We then select the 90 most frequent attributes across all entries. Additionally, we use the Number of Reviews column as the weight attribute.
\end{itemize}
\vspace{-2mm}
\subsubsection{Parameter Defaults and Ranges}\label{sec:params}
The default demand distribution is defined as a randomly generated integer uniformly sampled from the range 1 to 10. 
\revtwo{In the experiments evaluating the impact of demands, we generate each group’s demands uniformly at random within a specified window, and we increase the window range in successive powers of two, ranging from $2^0$ to $2^{12}$.} 
In cases where the generated demand cannot be satisfied by the data, we assign the item its maximum feasible demand value (the total number of sets that contain this item). The default setting assigns the number of items $\ell$ to 20 and the number of sets $n$ to 50,000. The parameter $\epsilon$ in our $(2+\epsilon)$-approximation algorithm is fixed to $0.2$ across all settings.
In experiments assessing the impact of $n$, the number of sets is varied from 1,024 up to the full size of the dataset. In experiments evaluating the impact of $\ell$, the number of items is varied from 3 up to the total number of items present in the dataset.

\subsection{Implementation}
\subsubsection{Baselines}
There is relatively little prior work directly addressing the special case of the weighted set multi-cover problem we studied in this paper. To the best of our knowledge, the greedy algorithm is the main practical baseline, as it is the most commonly used algorithm in this area. As another baseline, we also evaluate the classical randomized rounding algorithm used for the set cover problem and its variants. Further details for each of the two baselines are discussed below.

\stitle{Greedy} The greedy baseline extends the well-known approximation algorithm for the unweighted Set Cover problem, which at each step selects the set covering the largest number of uncovered elements, achieving an $O(\ln(\ell))$ approximation. Dobson \cite{dobson1982worst} generalized this approach to the weighted Set Multi-Cover problem by introducing a cover-to-weight ratio, defined as the number of elements with positive residual demand in a set divided by its weight, and iteratively selecting the set with the highest ratio. A straightforward implementation requires $O(\ell n^2)$ time, as ratios must be recomputed for all remaining sets at each step. By using a priority queue and updating only the necessary values, the running time can be reduced to $O(\ell n \log(n))$, which becomes near-linear $O(n \log(n))$ when the number of items $\ell$ is constant, while preserving the same $O(\ln(\ell)) = O(1)$ approximation guarantee. We note that, indeed, for the sake of a fair comparison, our greedy baseline is implemented using the faster approach.

\stitle{RR-LP} The RR-LP baseline implements the classical LP formulation for the covering problems, followed by a randomized rounding scheme. In this approach, for each set $t \in \Dee$, a variable $x_t$ is introduced, and the LP formulation is as follows:
\begin{align}
    \label{standardLP}\text{minimize:} \quad & \sum_{t \in \Dee} w_{t}\cdot x_t \\
    \notag\text{subject to:} \quad &\sum_{t \in \Dee, \, g \in 
    t}x_t \geq Q_g, \quad \forall g \in \Gee, \\
   \notag & 0 \leq x_t \leq 1, \quad \forall t \in \Dee.
\end{align}
After solving the above LP, each tuple $t \in \Dee$ is chosen in the returned solution with probability $\min\{1,\ln(n) \cdot x_t\}$. It is known \cite{Vazirani2010-hn} that this algorithm returns an $O(\ln(n))$-approximation solution to \prob. The runtime of this baseline is equal to the runtime of solving the LP instance, and hence it runs in $\mathcal{L}(n)=O(n^{2.38})$.

\vspace{-0.5em}
\subsubsection{Our algorithms}
We implement the two algorithms described in Sections \ref{subsec:biglp} and \ref{subsec:smalllp}. Once the LP instance is solved, we round all values down and compute the residual demands that remain after rounding. The final solution is then obtained by adding extra tuples as specified in Algorithm \ref{alg:lprounding}. In the rare event that the total residual demand is large after the rounding, we replace this procedure with a greedy algorithm to fulfill the remaining item demands. The key distinction between the $2$-approximation algorithm from Section \ref{subsec:biglp} and the $(2+\epsilon)$-approximation algorithm from Section \ref{subsec:smalllp} lies in the significantly smaller LP instance generated and solved in the latter. To assess the impact of the optimizations introduced in the $(2+\epsilon)$-approximation algorithm on execution time, we evaluate both algorithms in the experiments. However, the $(2+\epsilon)$-approximation algorithm is considered our main result.
\revone{In Appendix~\ref{appndx:experiments}, we evaluate the exact DP algorithm. Since the DP algorithm is computationally expensive, the related experiments are evaluated over smaller data as described in Appendix~\ref{appndx:experiments}, and the primary goal is to assess the quality of the solutions produced by our approaches against the best possible (optimal) solution.}

\begin{figure*}[!tb] 
\centering
    \includegraphics[width=0.5\textwidth]{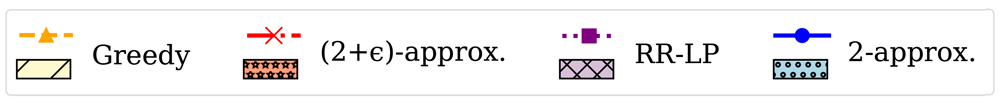}\\
    \begin{minipage}[t]{0.32\linewidth}
        \centering
        \includegraphics[width=\textwidth]{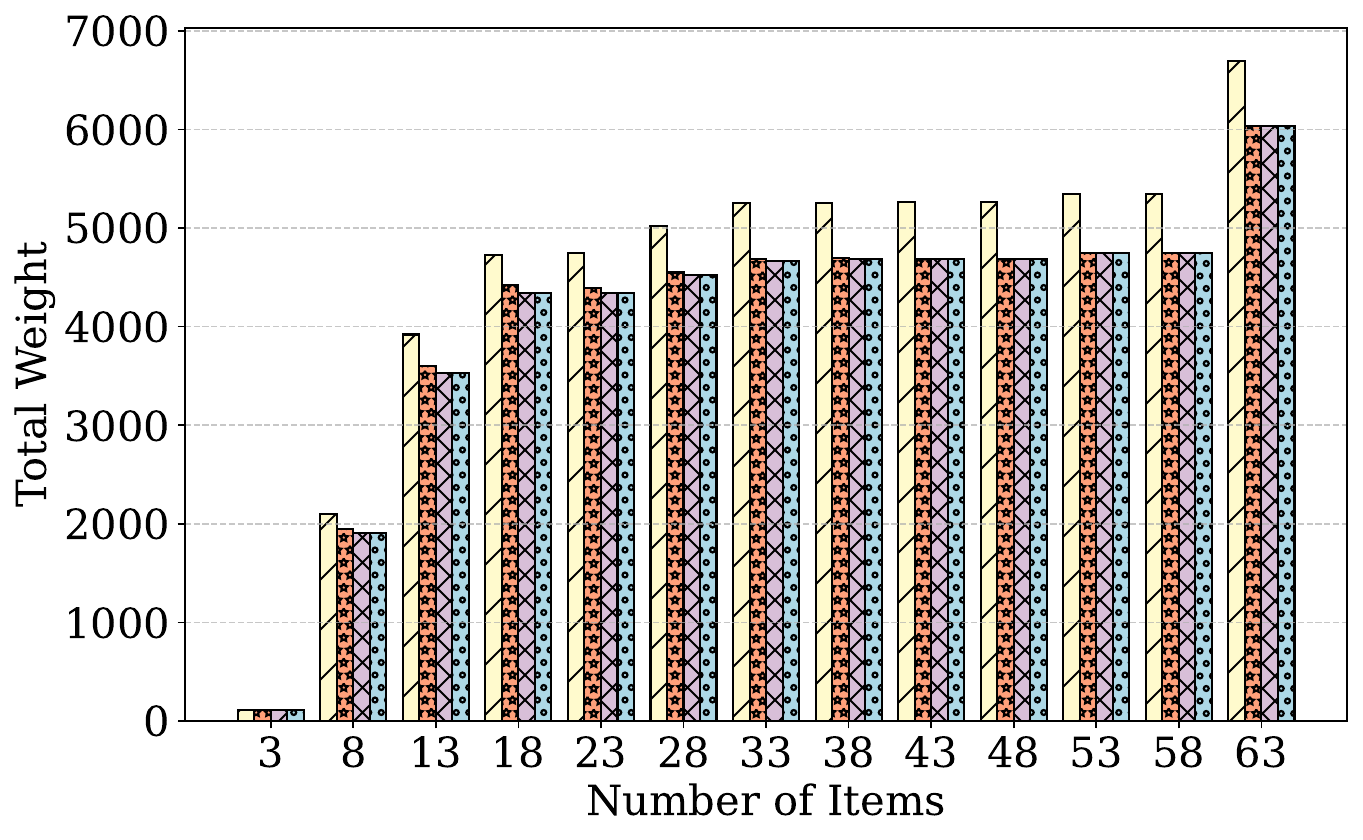}
        \vspace{-2.5em}
        \caption{\small Impact of varying number of items $\ell$ on the solution’s total weight, {\sf Census}}
        \label{fig:weight_varying_l_census}
    \end{minipage}
    \hfill
    \begin{minipage}[t]{0.32\linewidth}
        \centering
        \includegraphics[width=\textwidth]{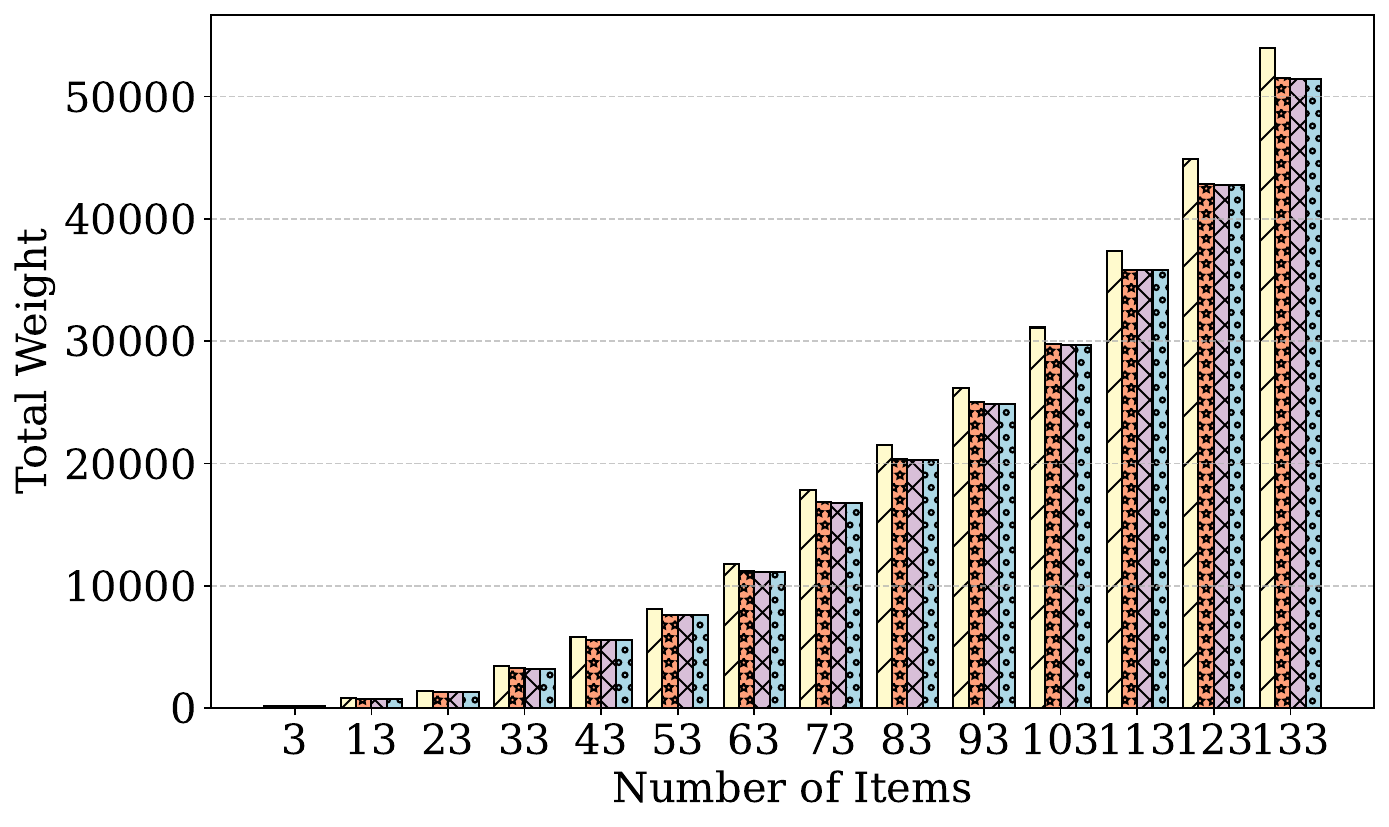}
        \vspace{-2.5em}
        \caption{\small Impact of varying number of items $\ell$ on the solution’s total weight, {\sf Music}}
        \label{fig:weight_varying_l_music}
    \end{minipage}
    \hfill
    \begin{minipage}[t]{0.32\linewidth}
        \centering
        \includegraphics[width=\textwidth]{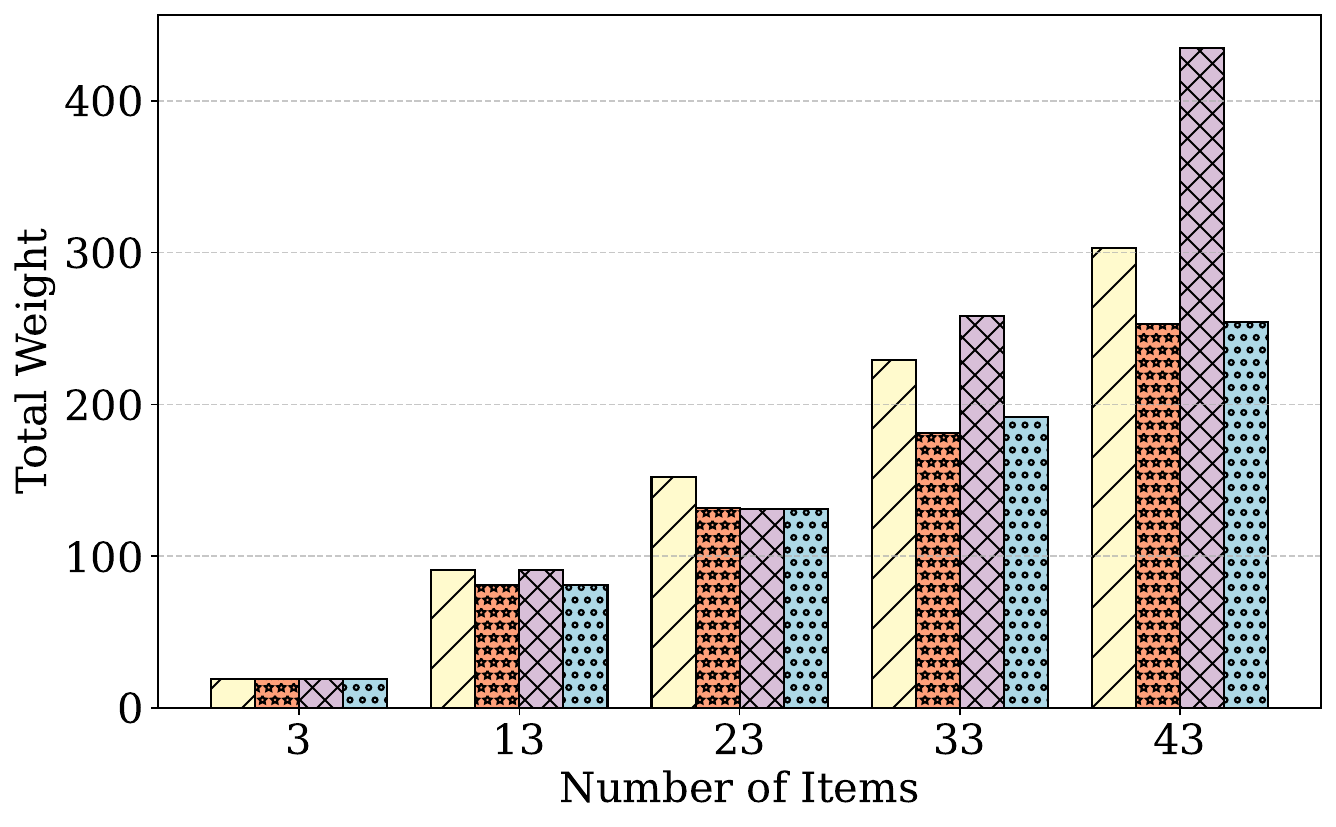}
        \vspace{-2.5em}
        \caption{\small Impact of varying number of items $\ell$ on the solution’s total weight, {\sf Stack Overflow}}
        \label{fig:weight_varying_l_stack}
    \end{minipage}
    \hfill
    \begin{minipage}[t]{0.32\linewidth}
        \centering
        \includegraphics[width=\textwidth]{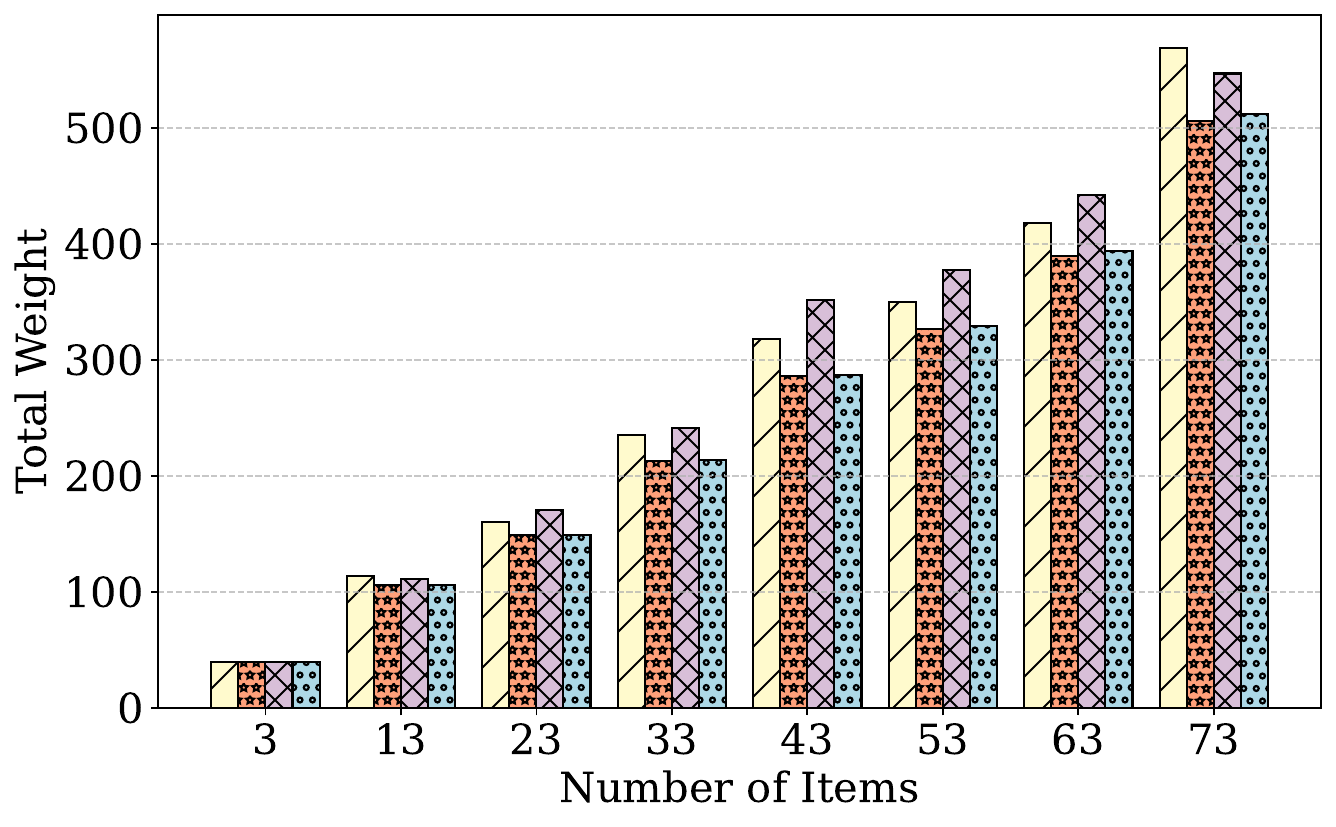}
        \vspace{-2.5em}
        \caption{\small Impact of varying number of items $\ell$ on the solution’s total weight, {\sf Yelp}}
        \label{fig:weight_varying_l_yelp}
    \end{minipage}
    \hfill
    \begin{minipage}[t]{0.32\linewidth}
        \centering
        \includegraphics[width=\textwidth]{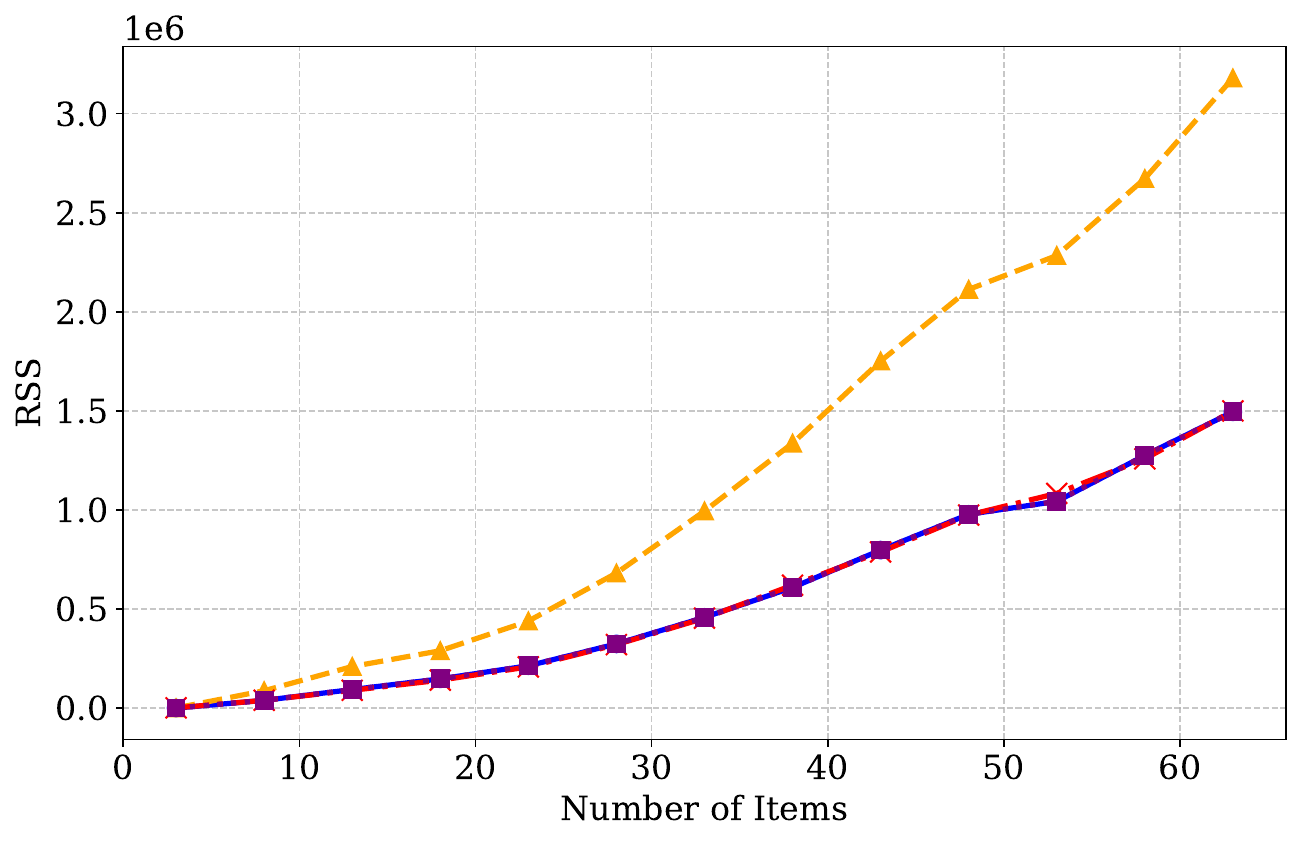}
        \vspace{-2.5em}
        \caption{\small Impact of varying number of items $\ell$ on deviation from demands, {\sf Census}}
        \label{fig:rss_varying_l_census}
    \end{minipage}
    \hfill
    \begin{minipage}[t]{0.32\linewidth}
        \centering
        \includegraphics[width=\textwidth]{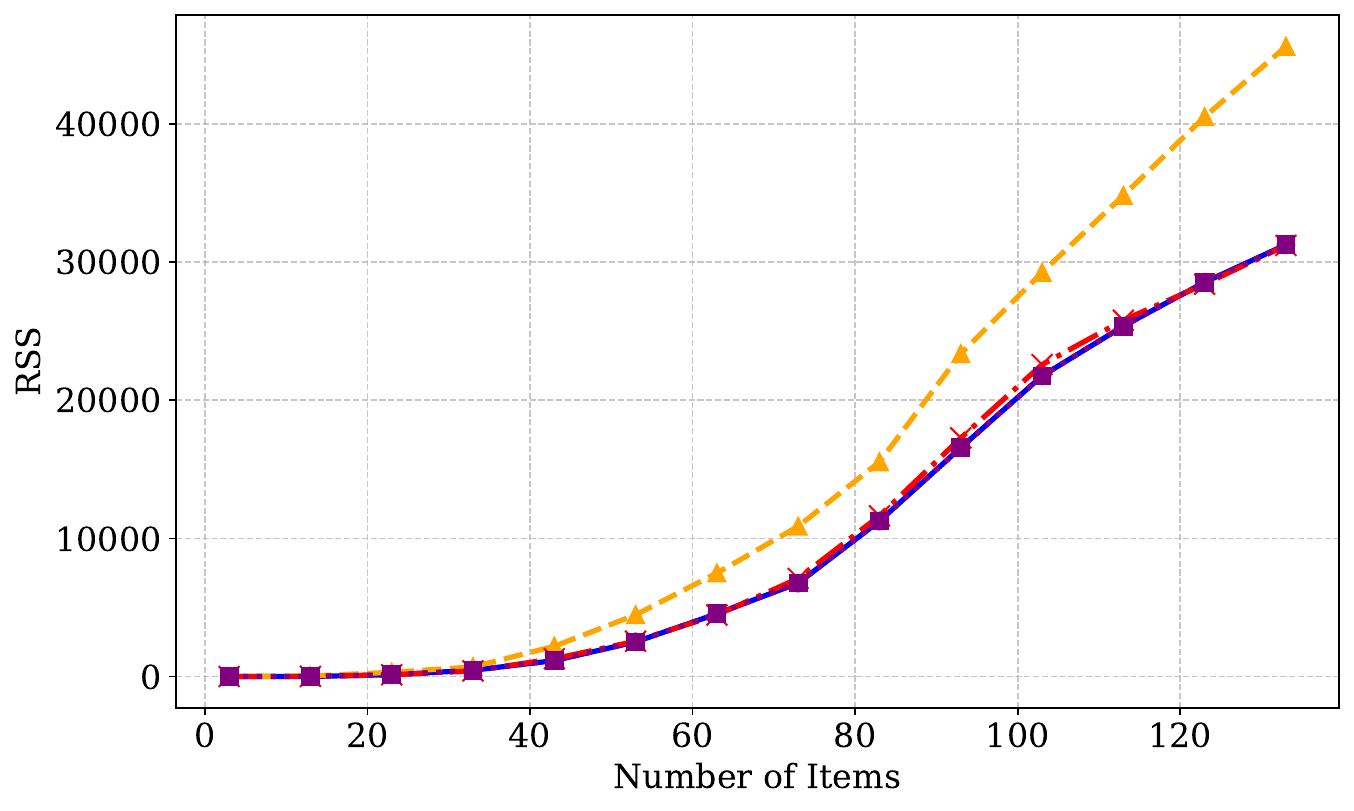}
        \vspace{-2.5em}
        \caption{\small Impact of varying number of items $\ell$ on deviation from demands, {\sf Music}}
        \label{fig:rss_varying_l_music}
    \end{minipage}
    \hfill
    \begin{minipage}[t]{0.32\linewidth}
        \centering
        \includegraphics[width=\textwidth]{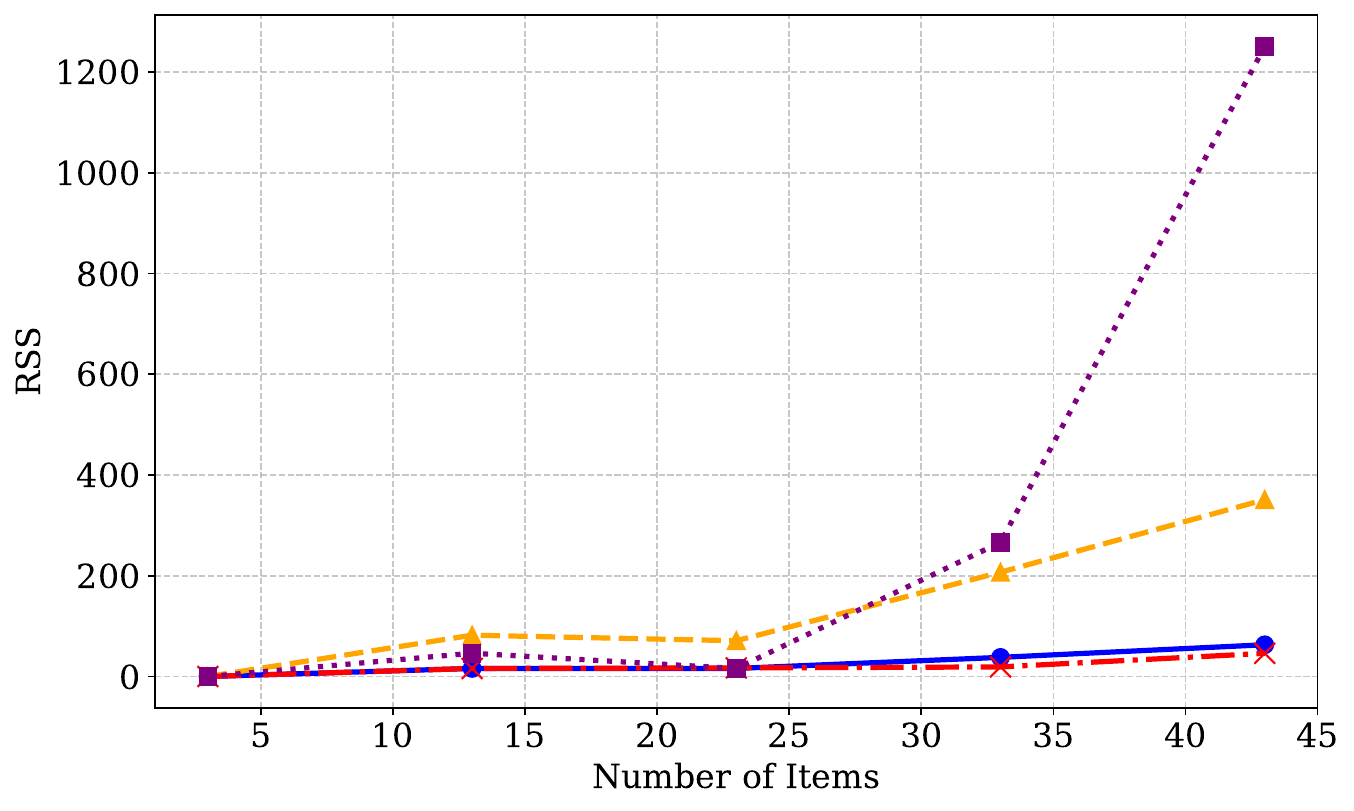}
        \vspace{-2.5em}
        \caption{\small Impact of varying number of items $\ell$ on deviation from demands, {\sf Stack Overflow}}
        \label{fig:rss_varying_l_stack}
    \end{minipage}
    \hfill
    \begin{minipage}[t]{0.32\linewidth}
        \centering
        \includegraphics[width=\textwidth]{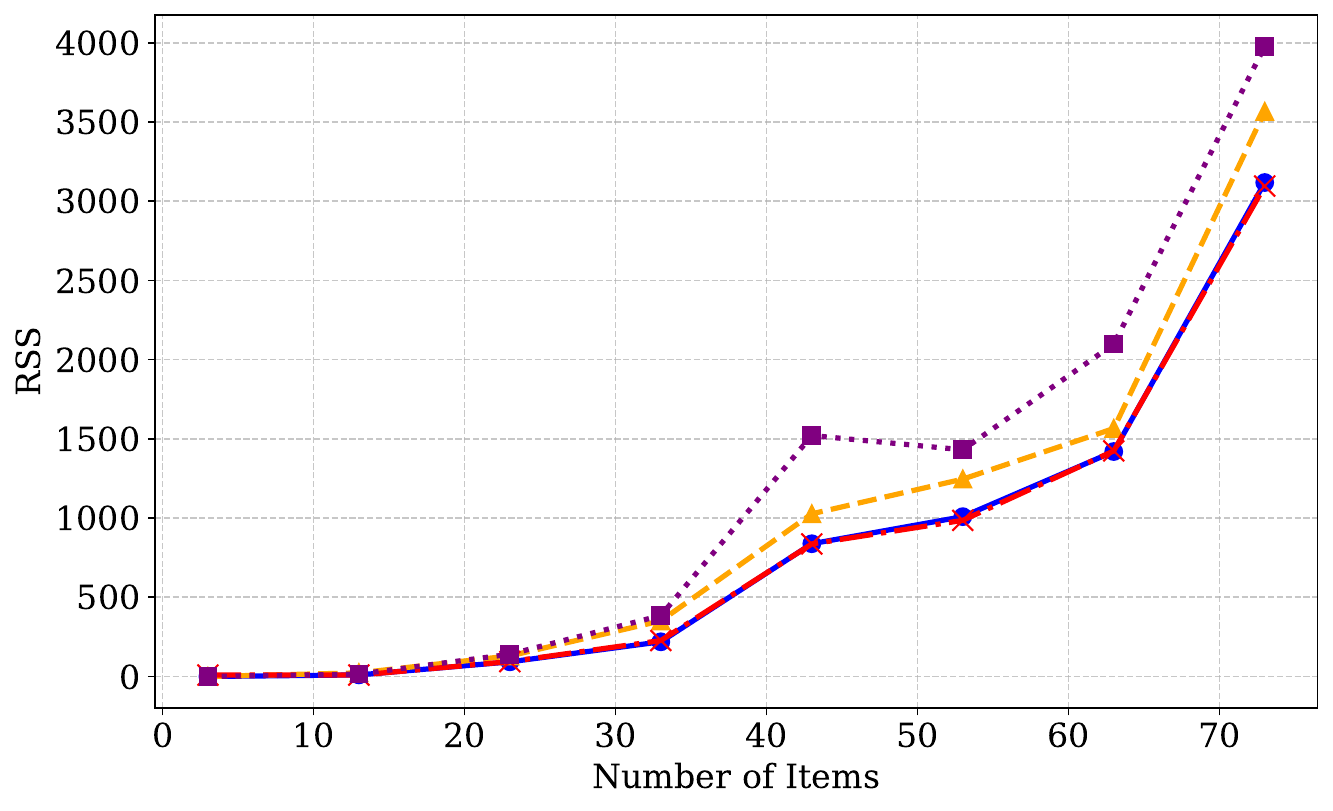}
        \vspace{-2.5em}
        \caption{\small Impact of varying number of items $\ell$ on deviation from demands, {\sf Yelp}}
        \label{fig:rss_varying_l_yelp}
    \end{minipage}
    \hfill
    \begin{minipage}[t]{0.32\linewidth}
        \centering
        \includegraphics[width=\textwidth]{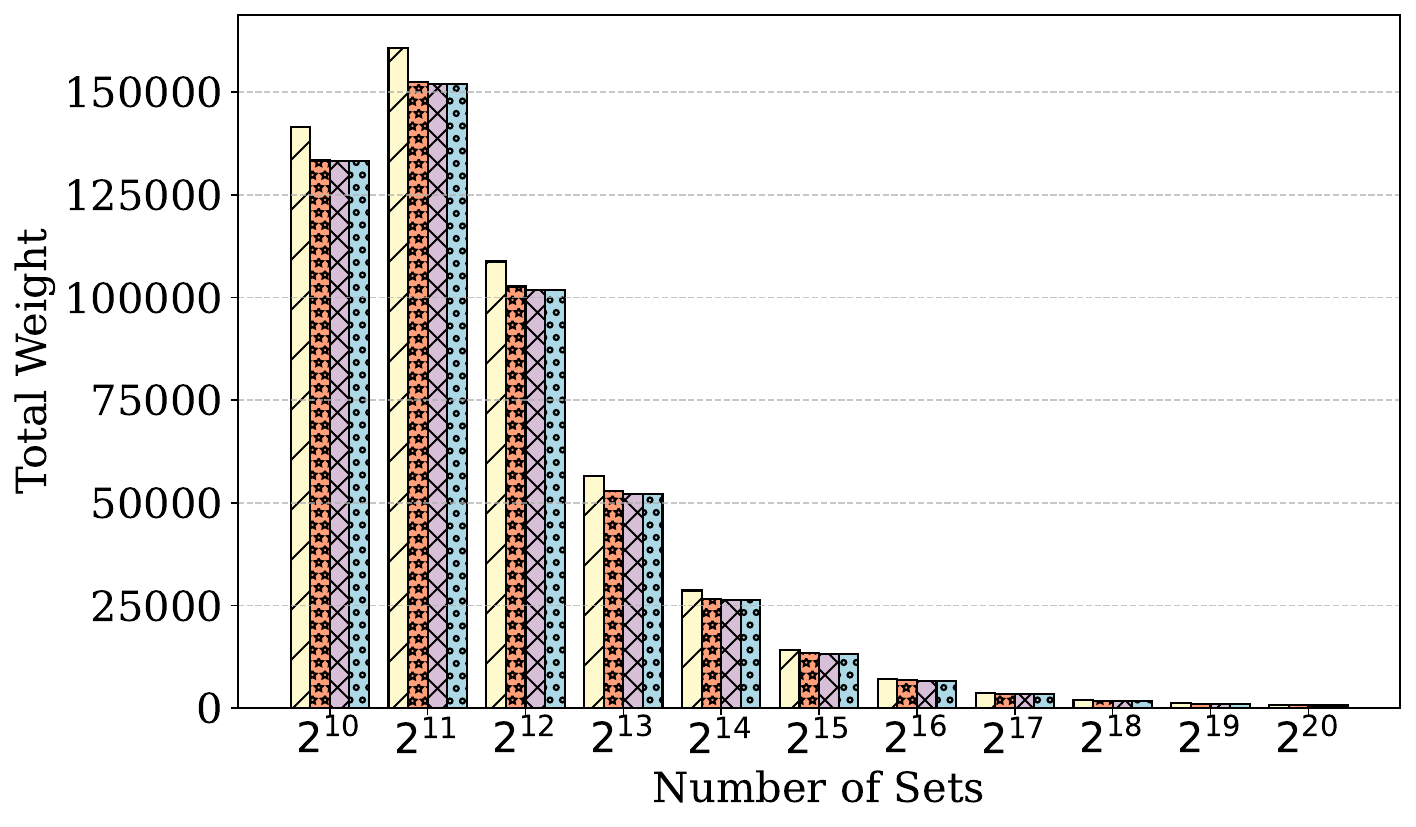}
        \vspace{-2.5em}
        \caption{\small Impact of varying number of sets $n$ on the solution’s total weight, {\sf Census}}
        \label{fig:weight_varying_n_census}
    \end{minipage}
    \hfill
    \begin{minipage}[t]{0.32\linewidth}
        \centering
        \includegraphics[width=\textwidth]{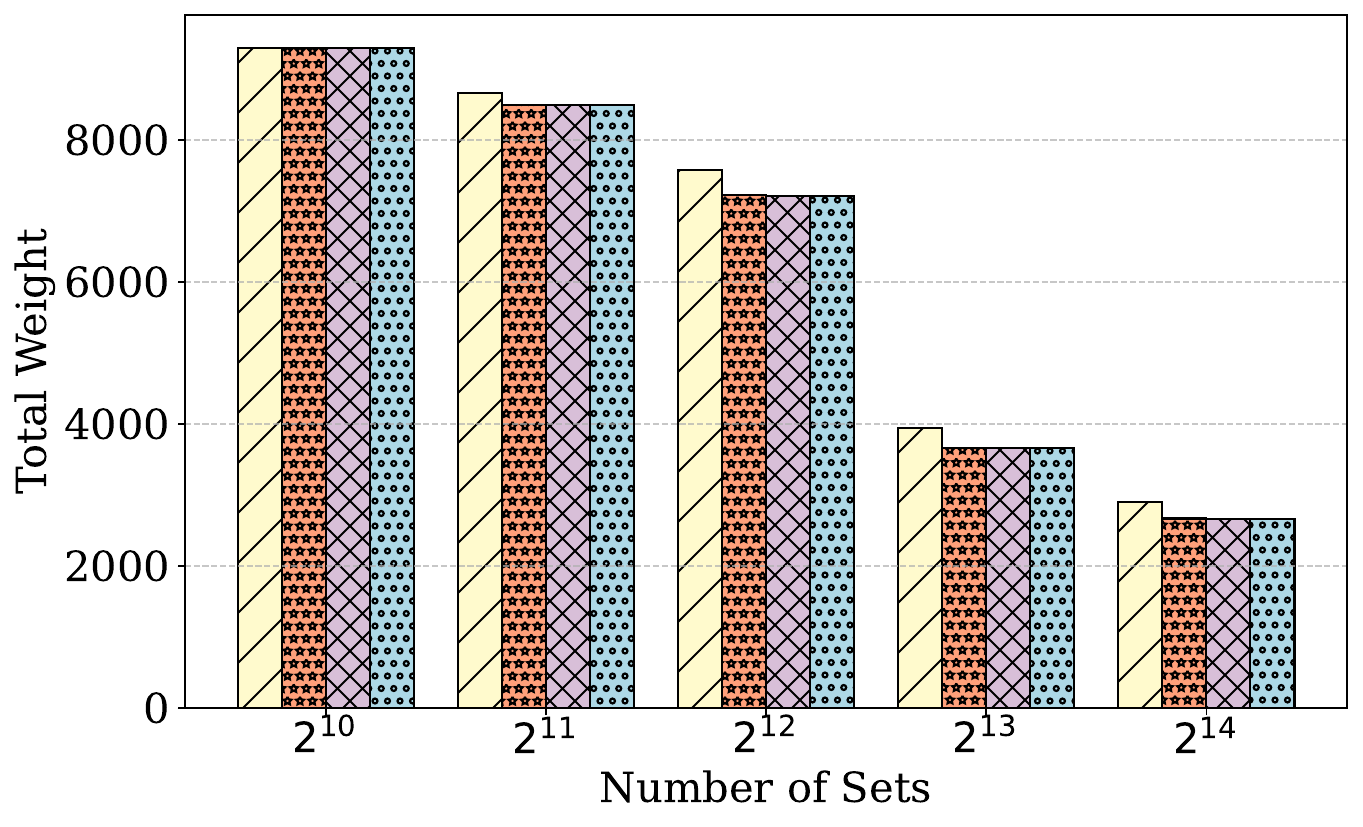}
        \vspace{-2.5em}
        \caption{\small Impact of varying number of sets $n$ on the solution’s total weight, {\sf Music}}
        \label{fig:weight_varying_n_music}
    \end{minipage}
    \hfill
    \begin{minipage}[t]{0.32\linewidth}
        \centering
        \includegraphics[width=\textwidth]{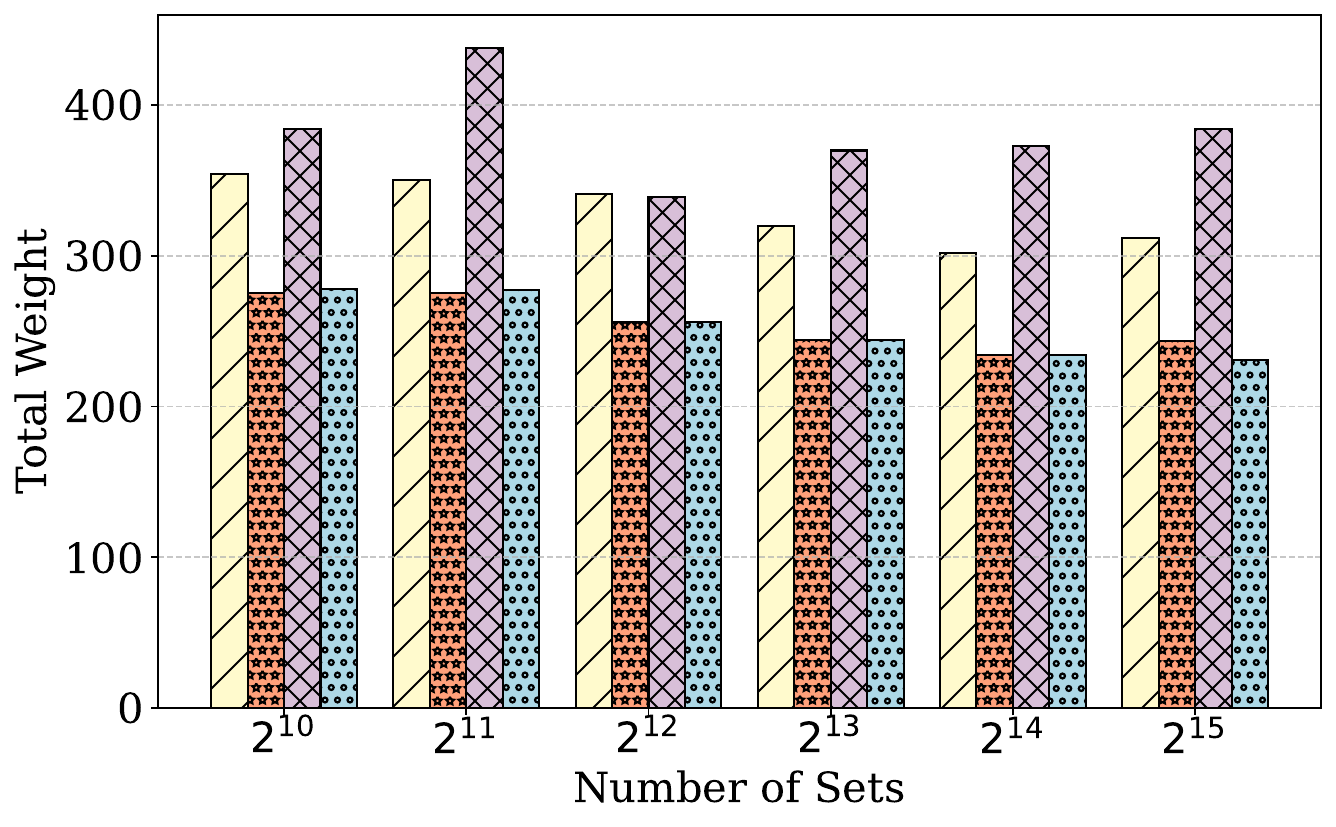}
        \vspace{-2.5em}
        \caption{\small Impact of varying number of sets $n$ on the solution’s total weight, {\sf Stack Overflow}}
        \label{fig:weight_varying_n_stack}
    \end{minipage}
    \hfill
    \begin{minipage}[t]{0.32\linewidth}
        \centering
        \includegraphics[width=\textwidth]{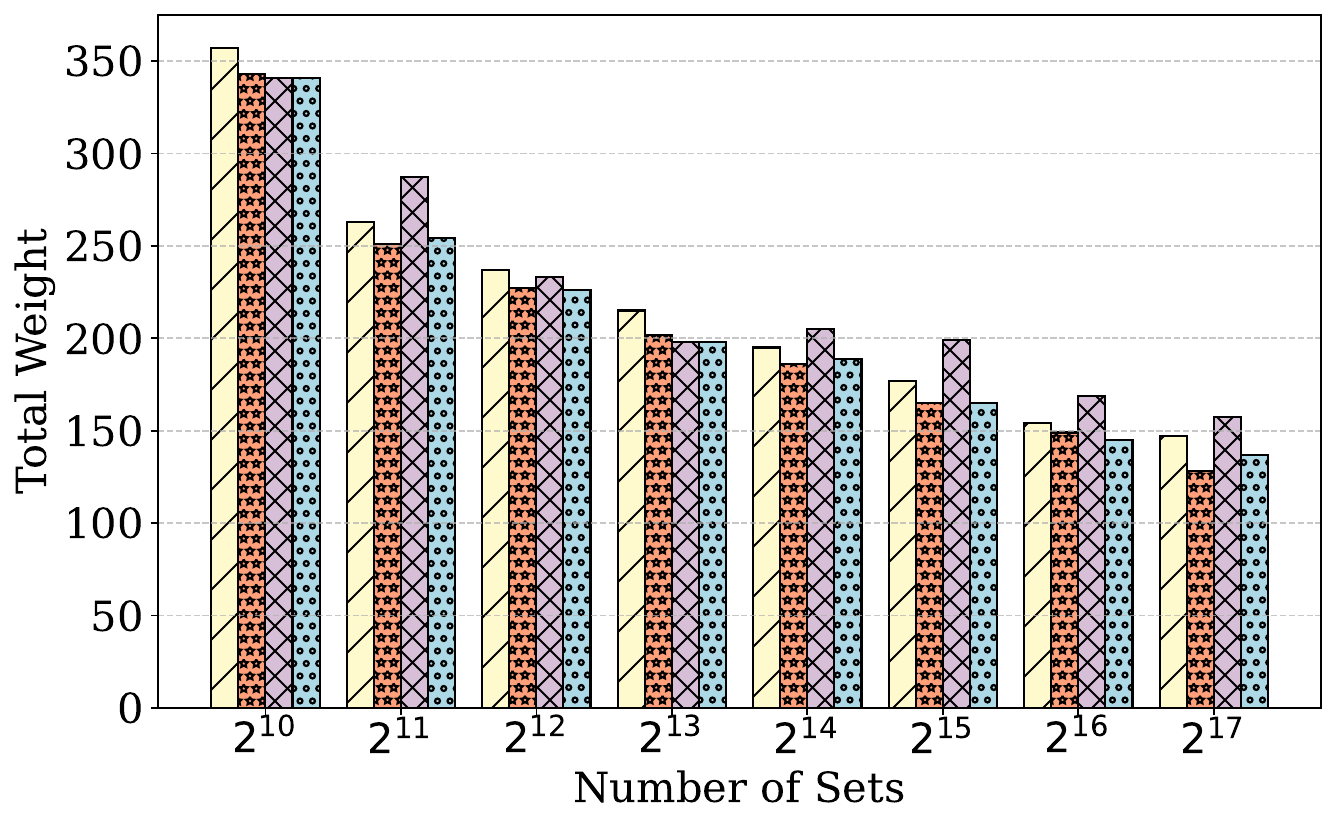}
        \vspace{-2.5em}
        \caption{\small Impact of varying number of sets $n$ on the solution’s total weight, {\sf Yelp}}
        \label{fig:weight_varying_n_yelp}
    \end{minipage}
\end{figure*}

\begin{figure*}
\centering
\includegraphics[width=0.65\textwidth]{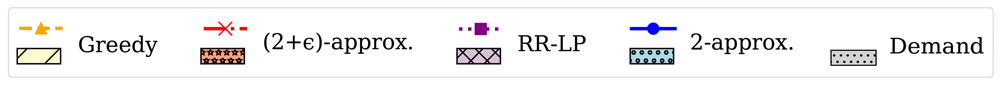}\\
    \begin{minipage}[t]{0.32\linewidth}
        \centering
        \includegraphics[width=\textwidth]{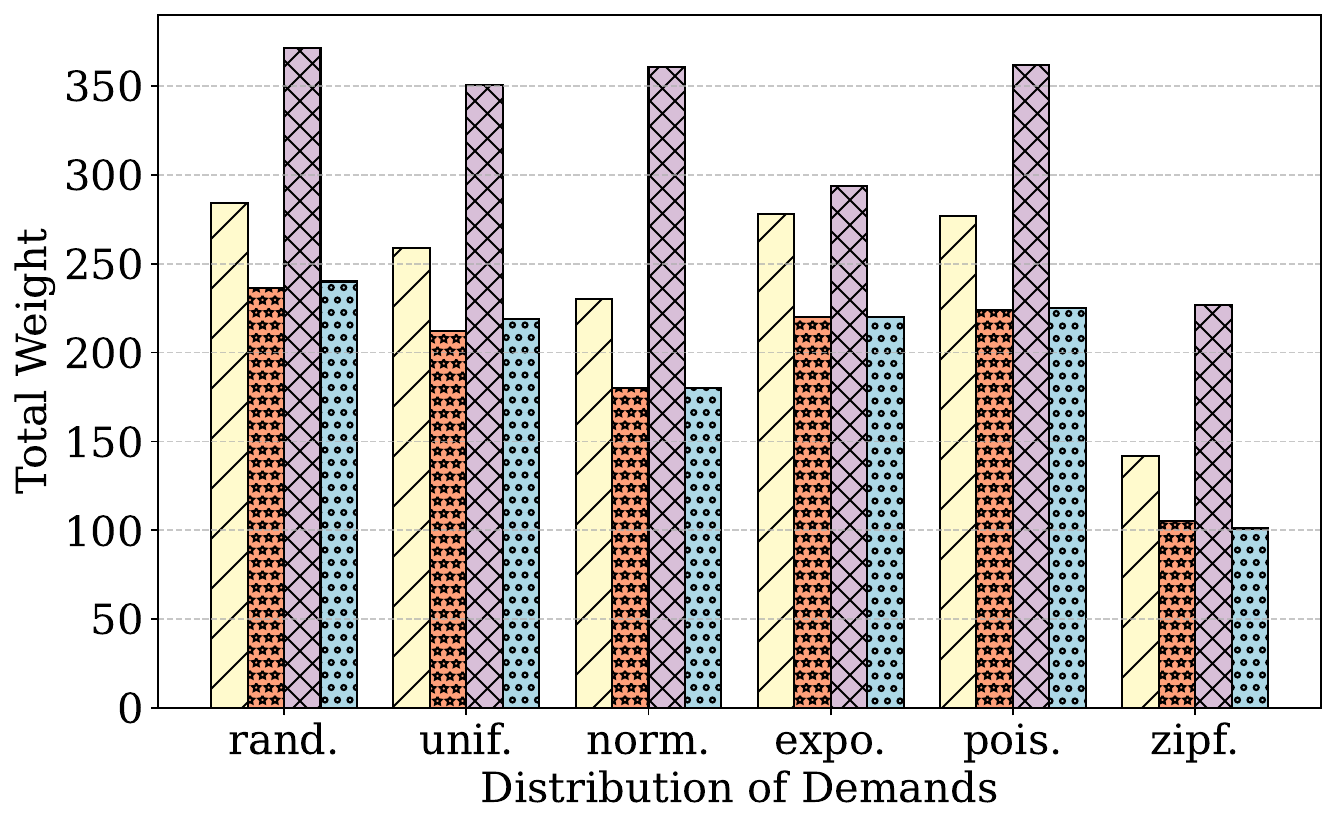}
        \vspace{-2.5em}
        \caption{\small Impact of demands distribution on the solution’s total weight, {\sf Stack Overflow}}
        \label{fig:weight_varying_dist_stack}
    \end{minipage}
    \hfill
    \begin{minipage}[t]{0.32\linewidth}
        \centering
        \includegraphics[width=\textwidth]{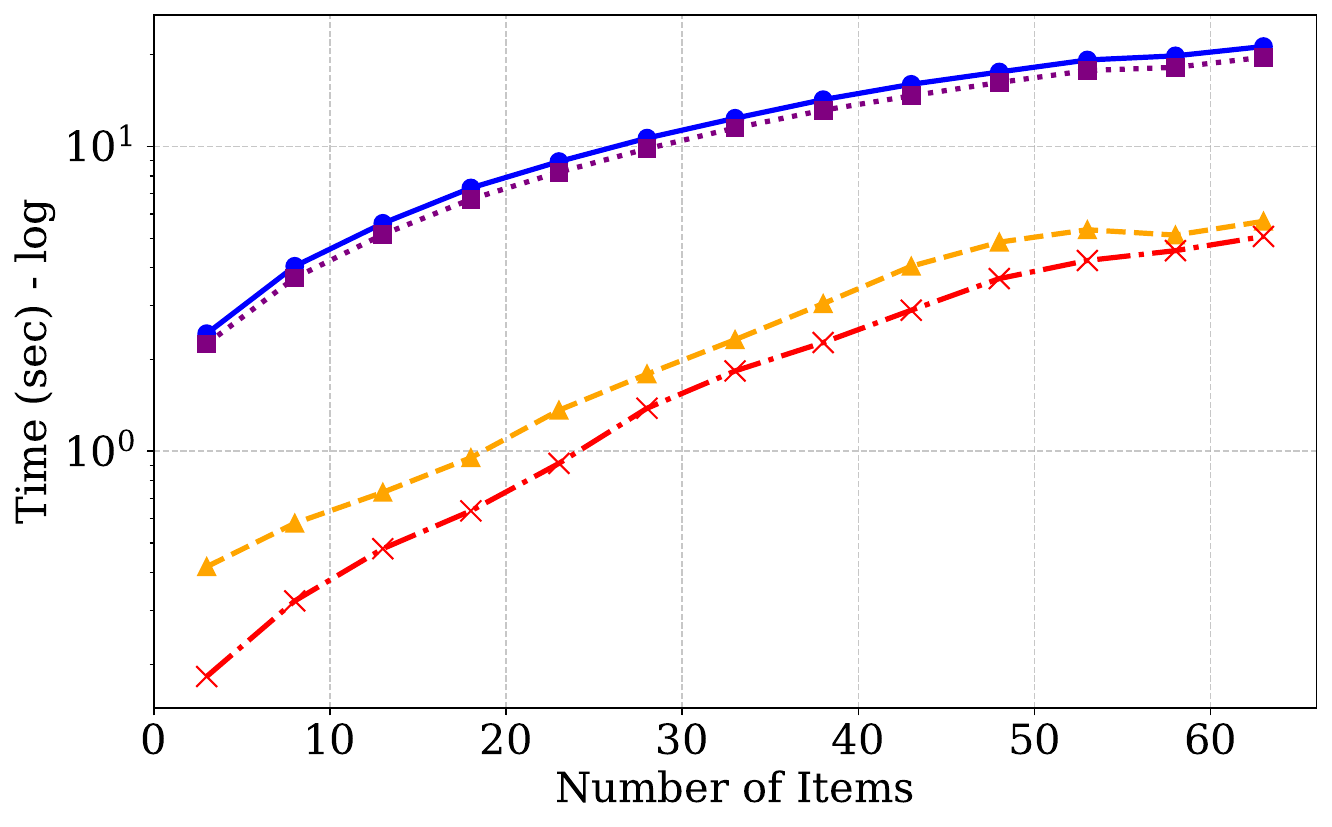}
        \vspace{-2.5em}
        \caption{\small Impact of varying number of items $\ell$ on the running time, {\sf Census}}
        \label{fig:time_varying_l_census}
    \end{minipage}
    \hfill
    \begin{minipage}[t]{0.32\linewidth}
        \centering
        \includegraphics[width=\textwidth]{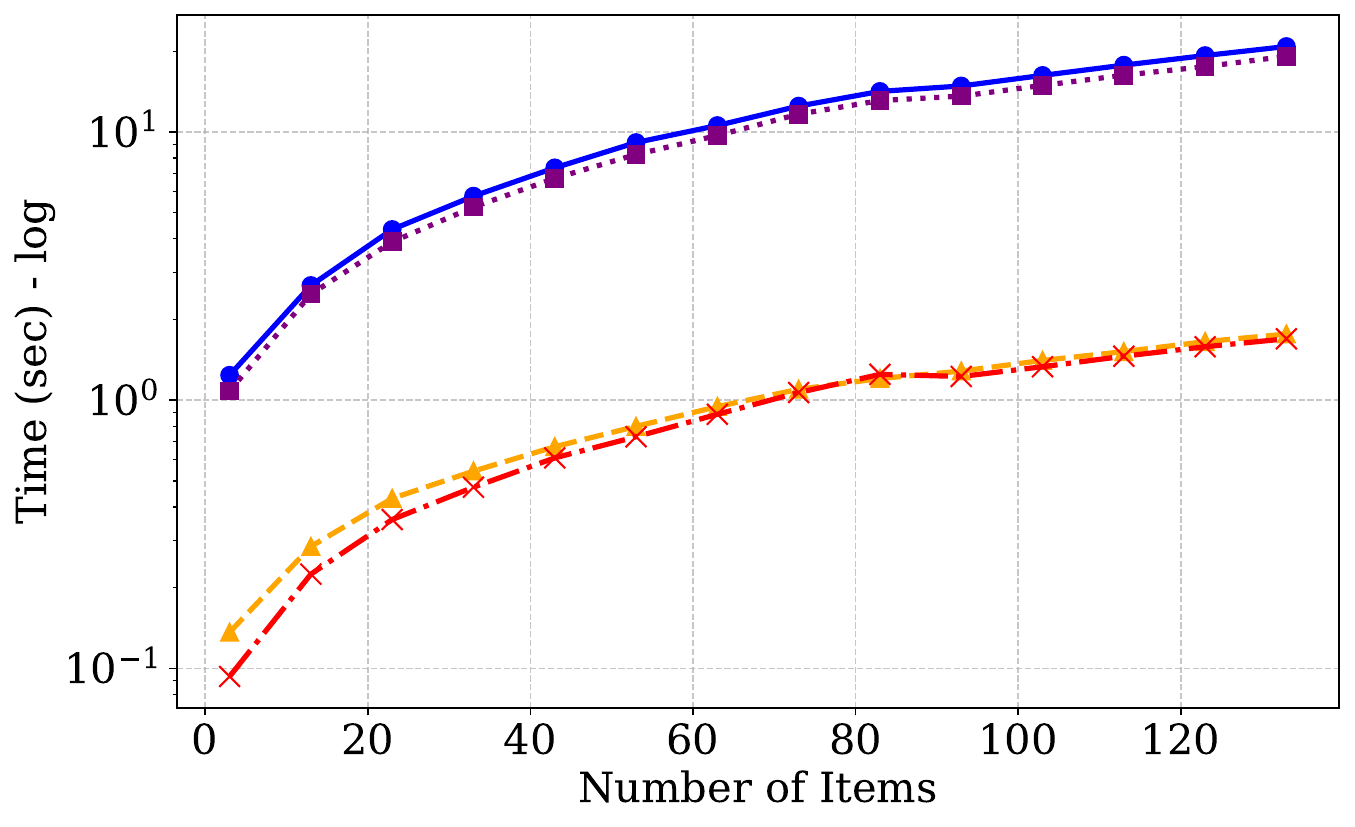}
        \vspace{-2.5em}
        \caption{\small Impact of varying number of items $\ell$ on the running time, {\sf Music}}
        \label{fig:time_varying_l_music}
    \end{minipage}
    \hfill
    \begin{minipage}[t]{0.32\linewidth}
        \centering
        \includegraphics[width=\textwidth]{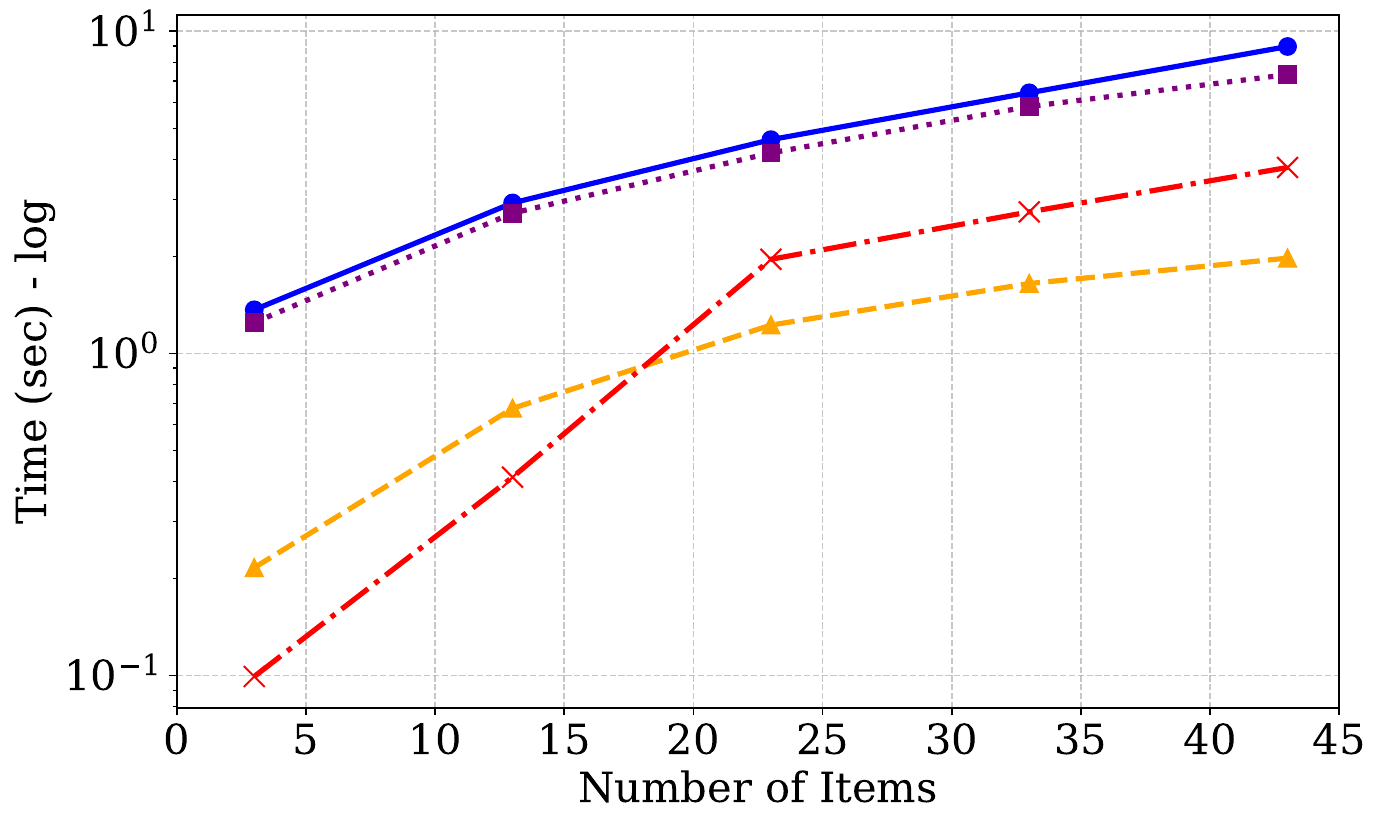}
        \vspace{-2.5em}
        \caption{\small Impact of varying number of items $\ell$ on the running time, {\sf Stack Overflow}}
        \label{fig:time_varying_l_stack}
    \end{minipage}
    \hfill
    \begin{minipage}[t]{0.32\linewidth}
        \centering
        \includegraphics[width=\textwidth]{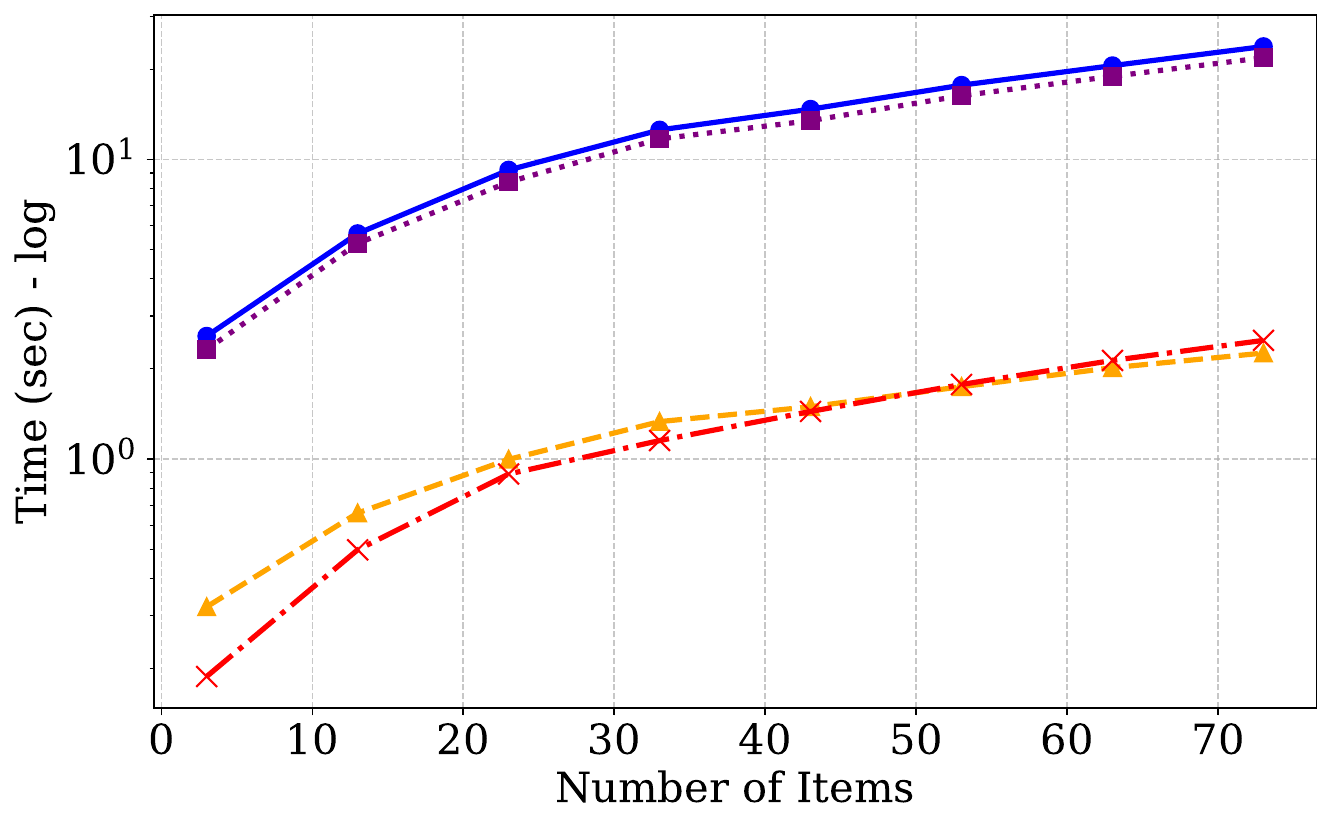}
        \vspace{-2.5em}
        \caption{\small Impact of varying number of items $\ell$ on the running time, {\sf Yelp}}
        \label{fig:time_varying_l_yelp}
    \end{minipage}
    \hfill
    \begin{minipage}[t]{0.32\linewidth}
        \centering
        \includegraphics[width=\textwidth]{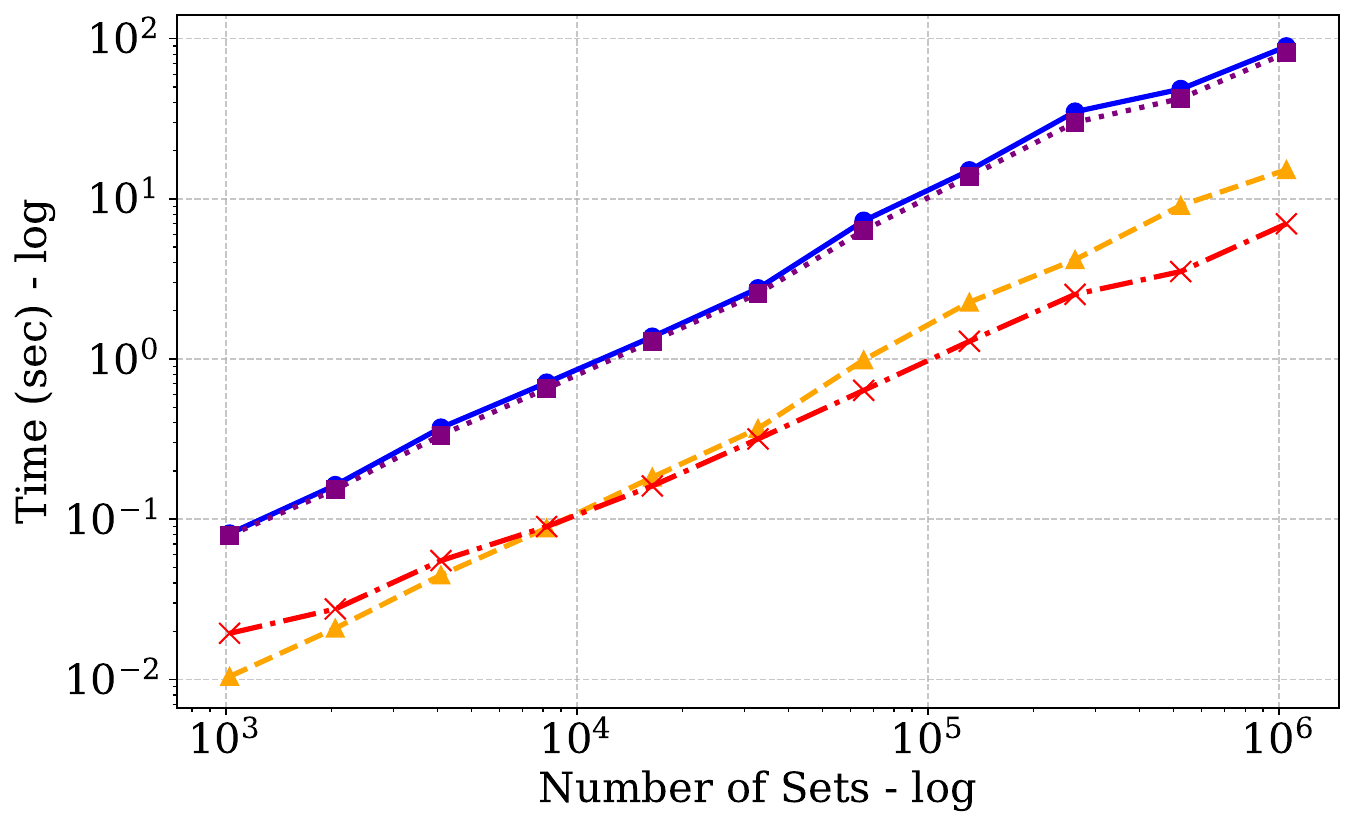}
        \vspace{-2.5em}
        \caption{\small Impact of varying number of sets $n$ on the running time, {\sf Census}}
        \label{fig:time_varying_n_census}
    \end{minipage}
    \hfill
    \begin{minipage}[t]{0.32\linewidth}
        \centering
        \includegraphics[width=\textwidth]{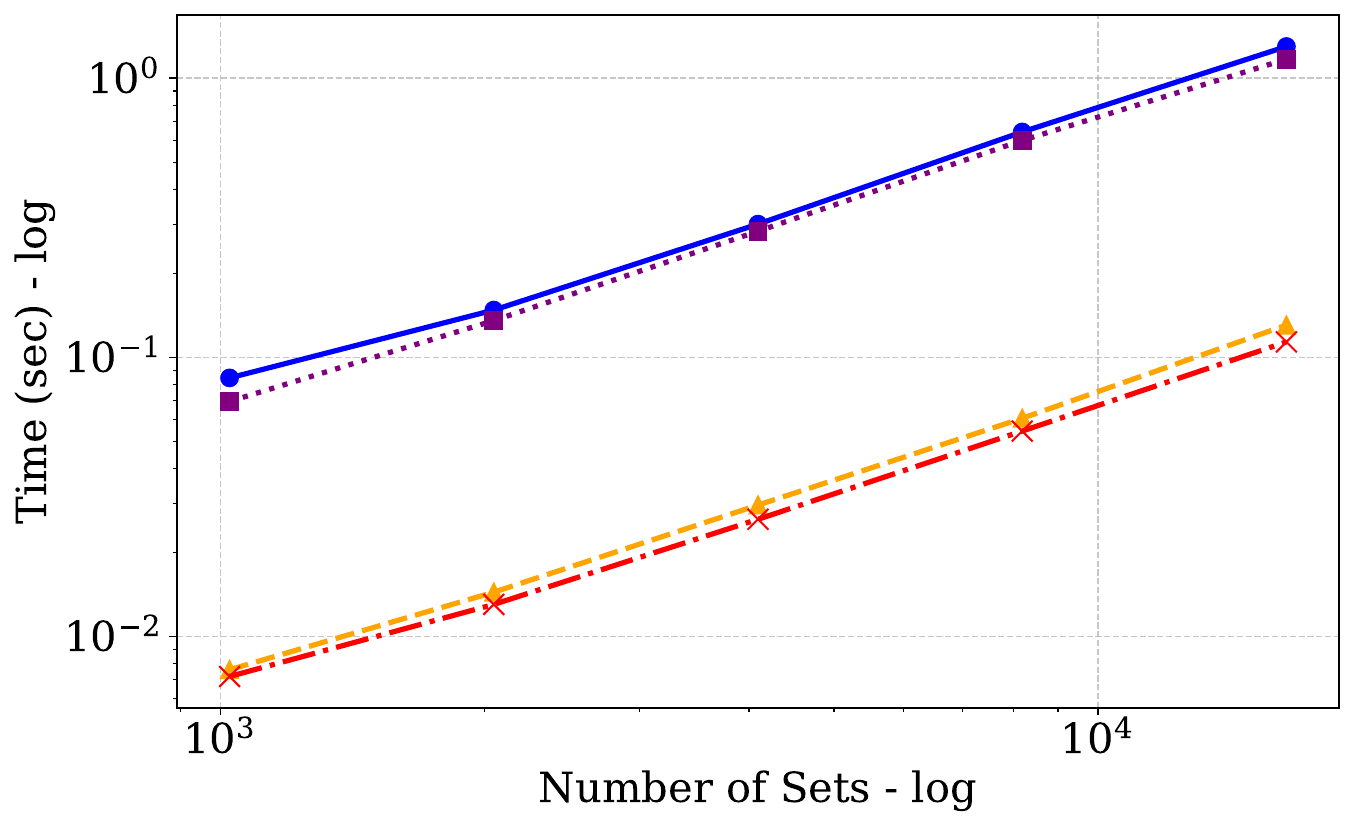}
        \vspace{-2.5em}
        \caption{\small Impact of varying number of sets $n$ on the running time, {\sf Music}}
        \label{fig:time_varying_n_music}
    \end{minipage}
    \hfill
    \begin{minipage}[t]{0.32\linewidth}
        \centering
        \includegraphics[width=\textwidth]{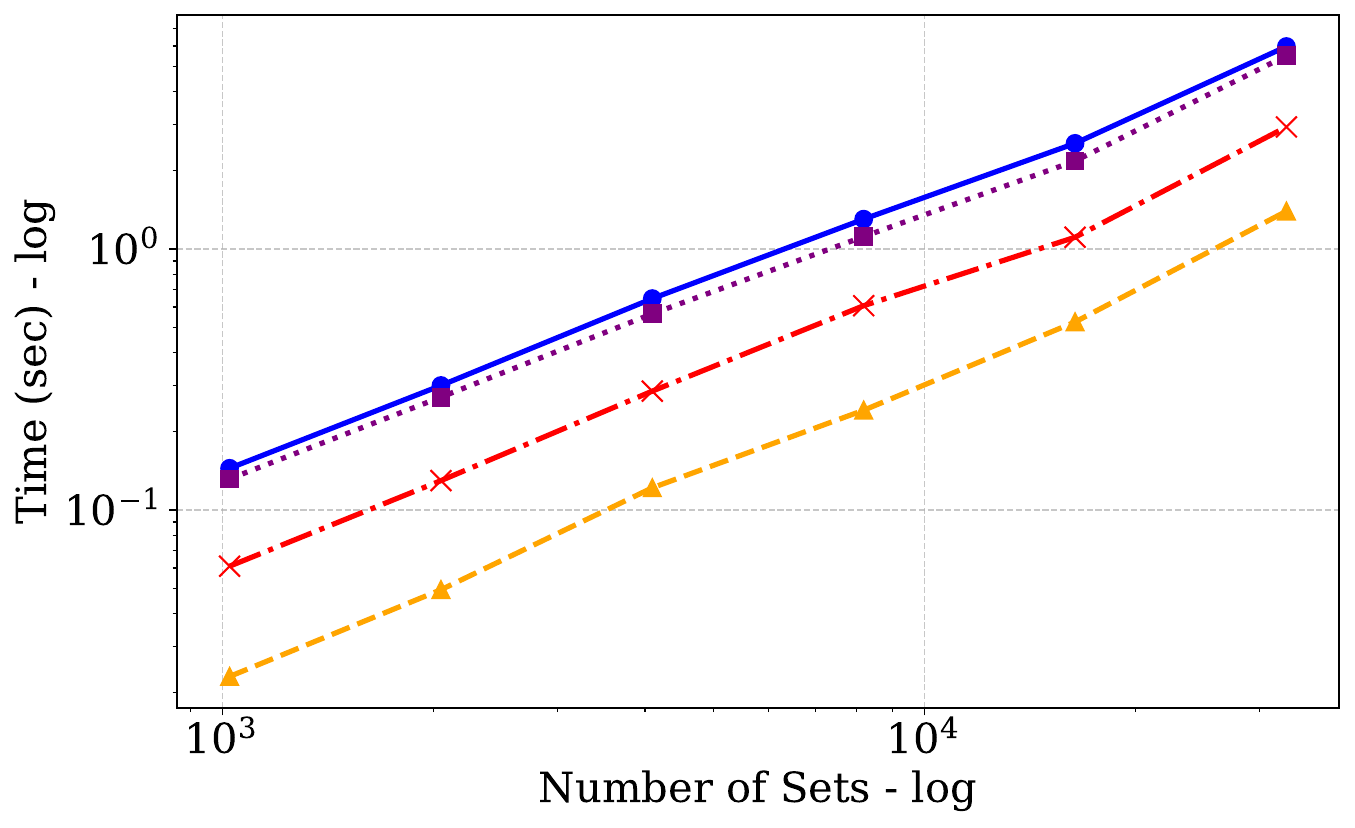}
        \vspace{-2.5em}
        \caption{\small Impact of varying number of sets $n$ on the running time, {\sf Stack Overflow}}
        \label{fig:time_varying_n_stack}
    \end{minipage}
    \hfill
    \begin{minipage}[t]{0.32\linewidth}
        \centering
        \includegraphics[width=\textwidth]{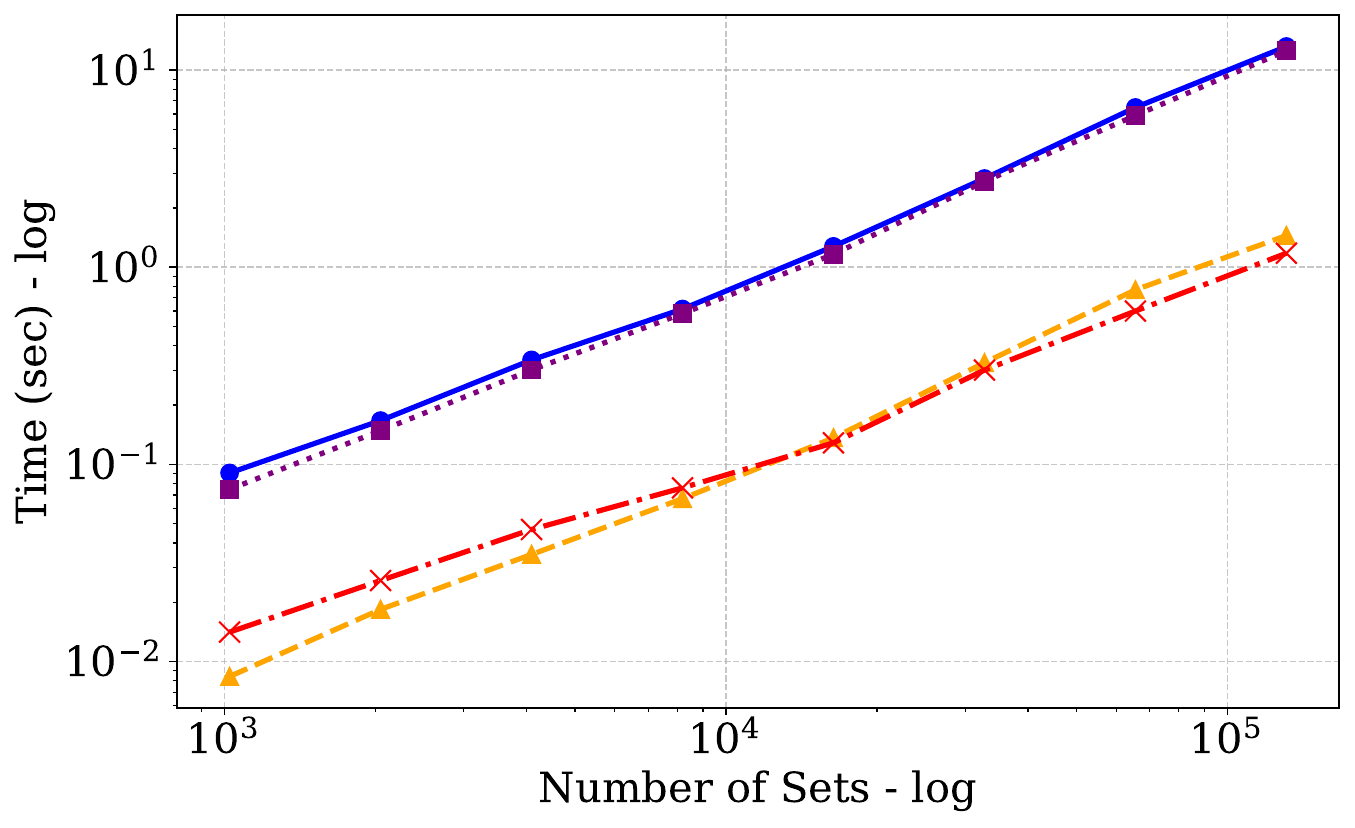}
        \vspace{-2.5em}
        \caption{\small Impact of varying number of sets $n$ on the running time, {\sf Yelp}}
        \label{fig:time_varying_n_yelp}
    \end{minipage}
    \hfill
    \begin{minipage}[t]{0.32\linewidth}
        \centering
        \includegraphics[width=\textwidth]{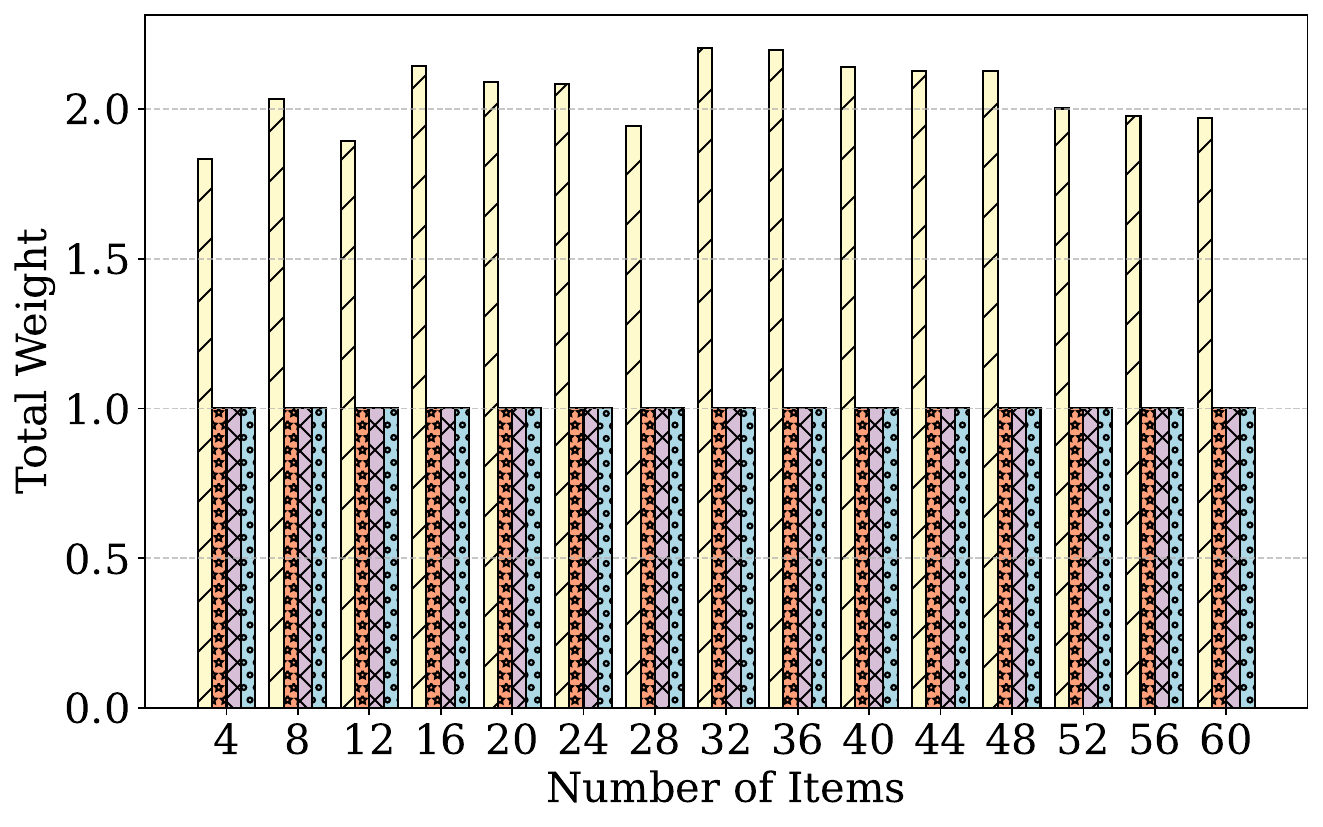}
        \vspace{-2.5em}
        \caption{\small Impact of varying number of items $\ell$ on the solution’s total weight, {\sf Synthetic}}
        \label{fig:case_study_1}
    \end{minipage}
    \hfill
    \begin{minipage}[t]{0.32\linewidth}
        \centering
        \includegraphics[width=\textwidth]{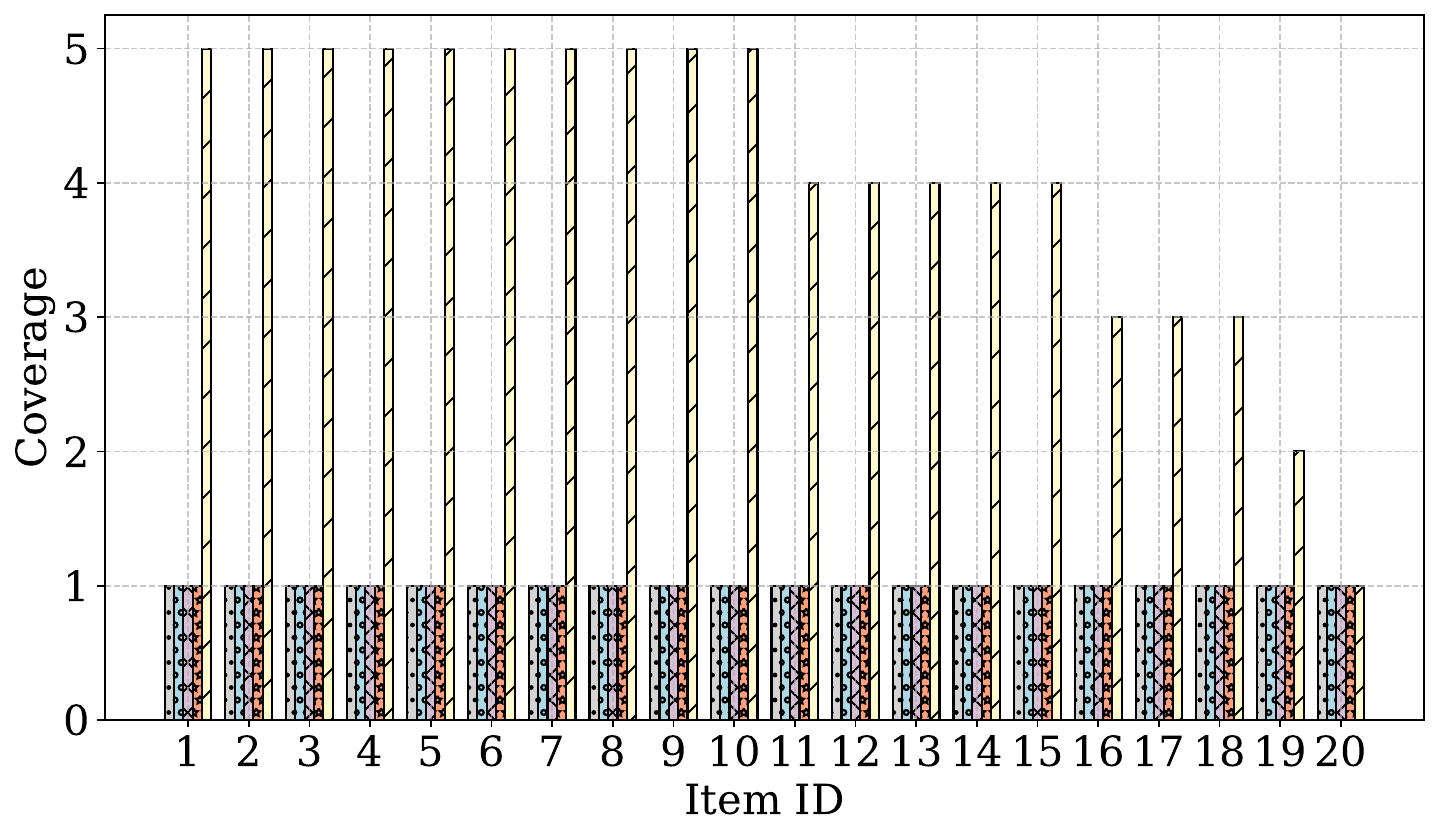}
        \vspace{-2.5em}
        \caption{\small Comparison of the actual demands vs. the item coverage, {\sf Synthetic}, $\ell=20$\protect\footnotemark.}
        \label{fig:case_study_2}
    \end{minipage}
    \hfill
    \begin{minipage}[t]{0.32\linewidth}
        \centering
        \includegraphics[width=\textwidth]{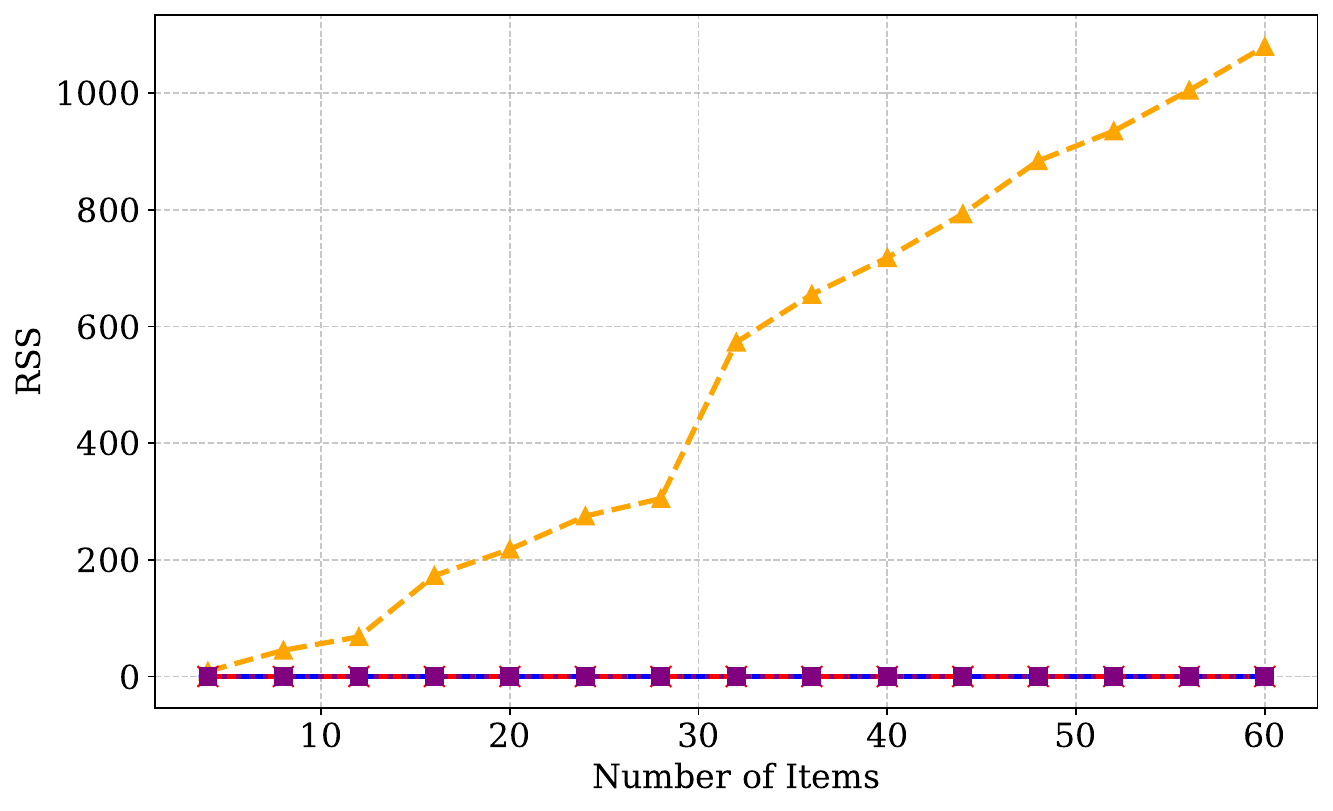}
        \vspace{-2.5em}
        \caption{\small Impact of varying number of items $\ell$ on deviation from demands, {\sf Synthetic}}
        \label{fig:case_study_3}
    \end{minipage}
\end{figure*}
\begin{figure*}[!tb] 
\centering
    \begin{minipage}[t]{0.48\linewidth}
        \centering
        \includegraphics[width=\textwidth]{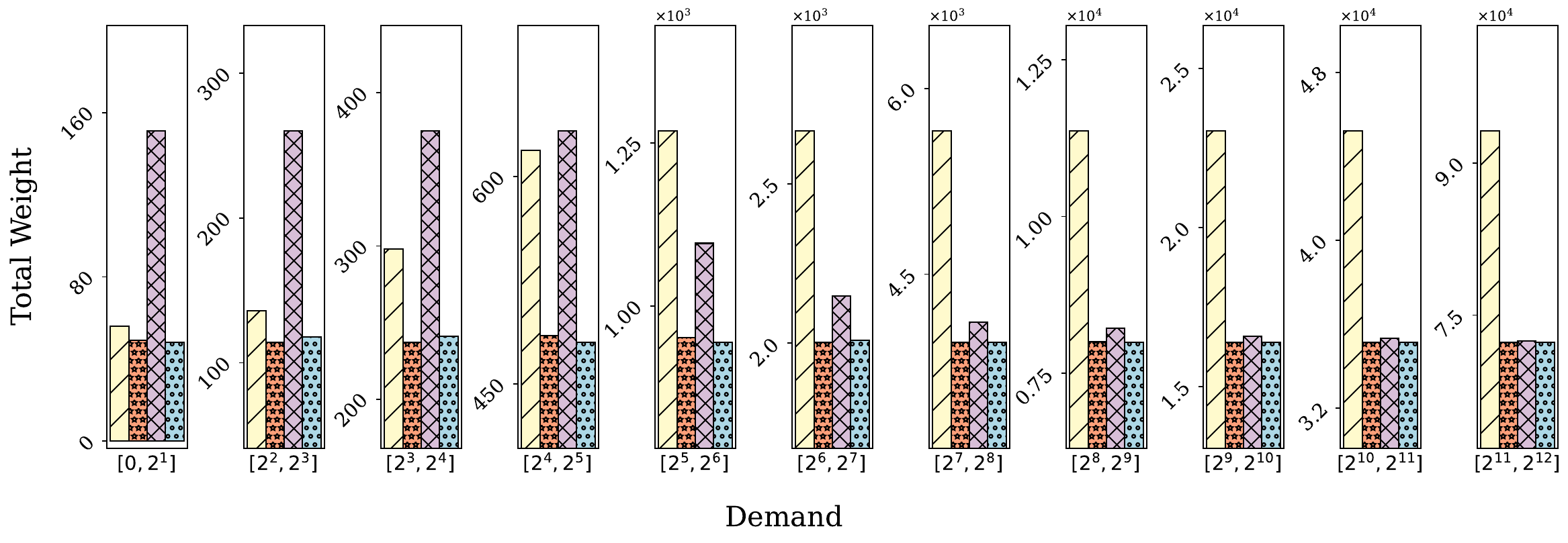}
        \vspace{-2.5em}
        \caption{\revtwo{\small Impact of varying demands on the solution’s total weight, {\sf Stack Overflow}}}
        \label{fig:weight_varying_demands_stack}
    \end{minipage}
    \hfill
    \begin{minipage}[t]{0.48\linewidth}
        \centering
        \includegraphics[width=\textwidth]{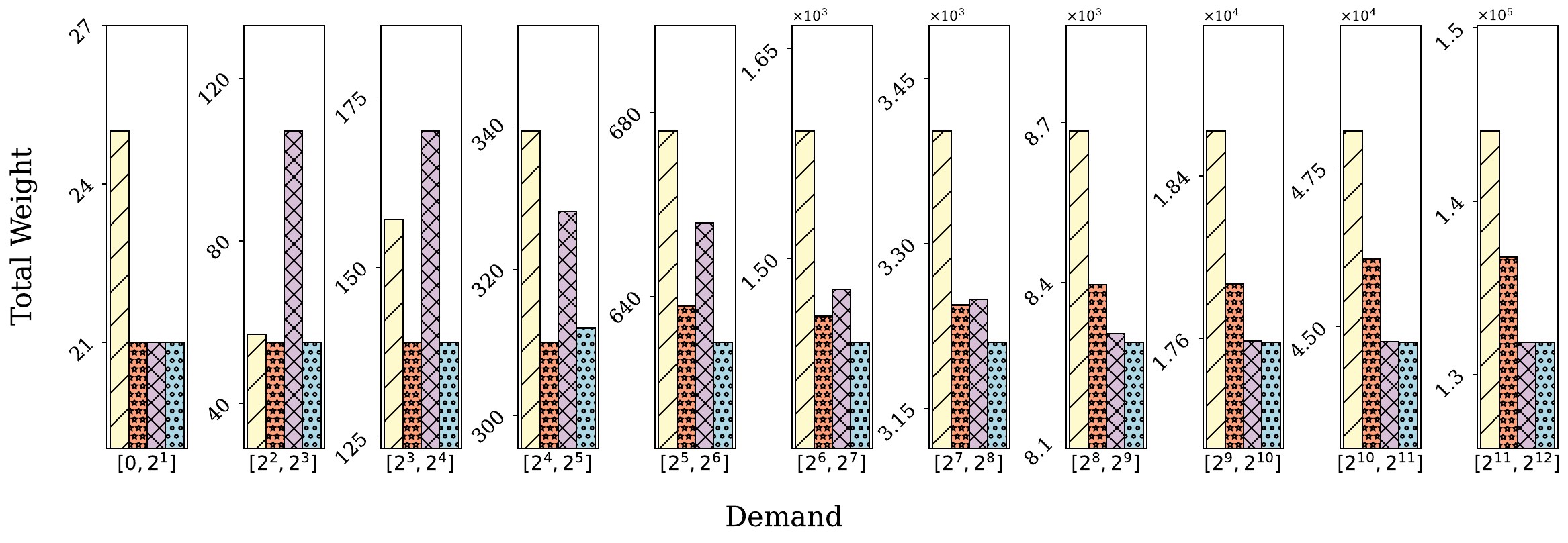}
        \vspace{-2.5em}
        \caption{\revtwo{\small Impact of varying demands on the solution’s total weight, {\sf Yelp}}}
        \label{fig:weight_varying_demands_yelp}
    \end{minipage}
\end{figure*}

\begin{figure*}
    \vspace{-1.5em}
    \begin{minipage}[t]{0.32\linewidth}
        \centering
        \includegraphics[width=\textwidth]{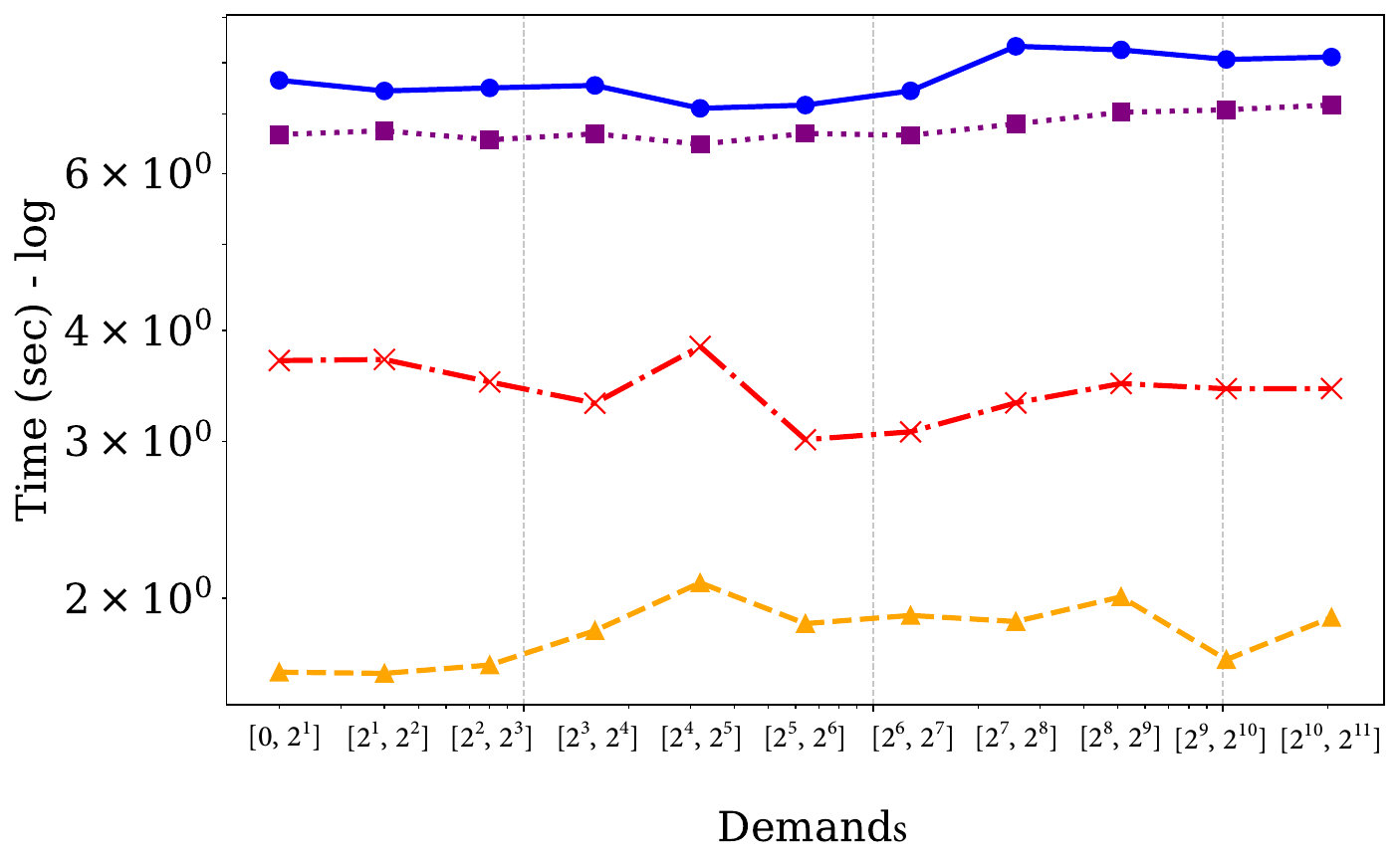}
        \vspace{-2.5em}
        \caption{\revtwo{\small Impact of varying demands on the running time, {\sf Stack Overflow}}}
        \label{fig:time_varying_demands_stack}
    \end{minipage}
    \hfill
    \begin{minipage}[t]{0.32\linewidth}
        \centering
        \includegraphics[width=\textwidth]{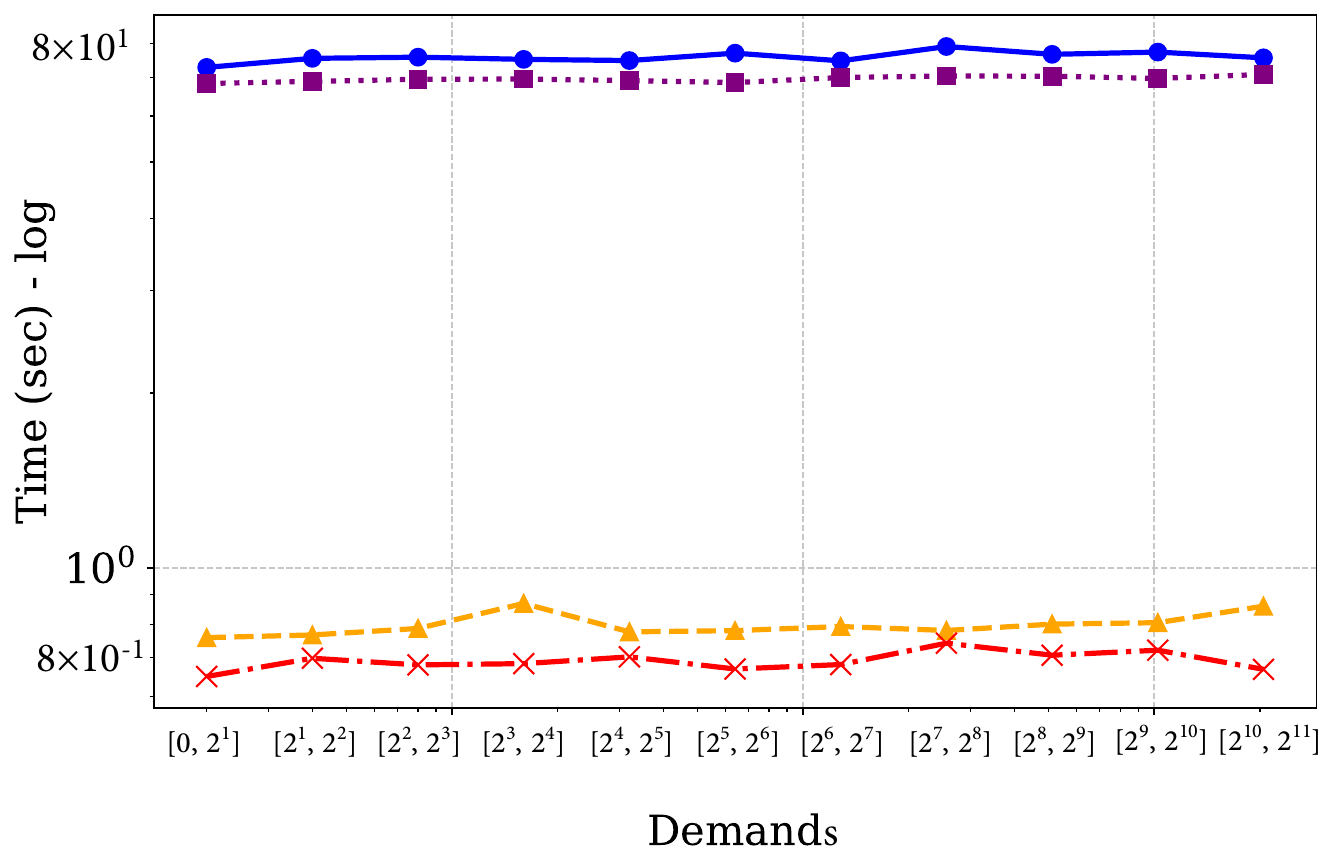}
        \vspace{-2.5em}
        \caption{\revtwo{\small Impact of varying demands on the running time, {\sf Yelp}}}
        \label{fig:time_varying_demands_yelp}
    \end{minipage}
    \hfill
    \begin{minipage}[t]{0.32\linewidth}
        \centering
        \includegraphics[width=\textwidth]{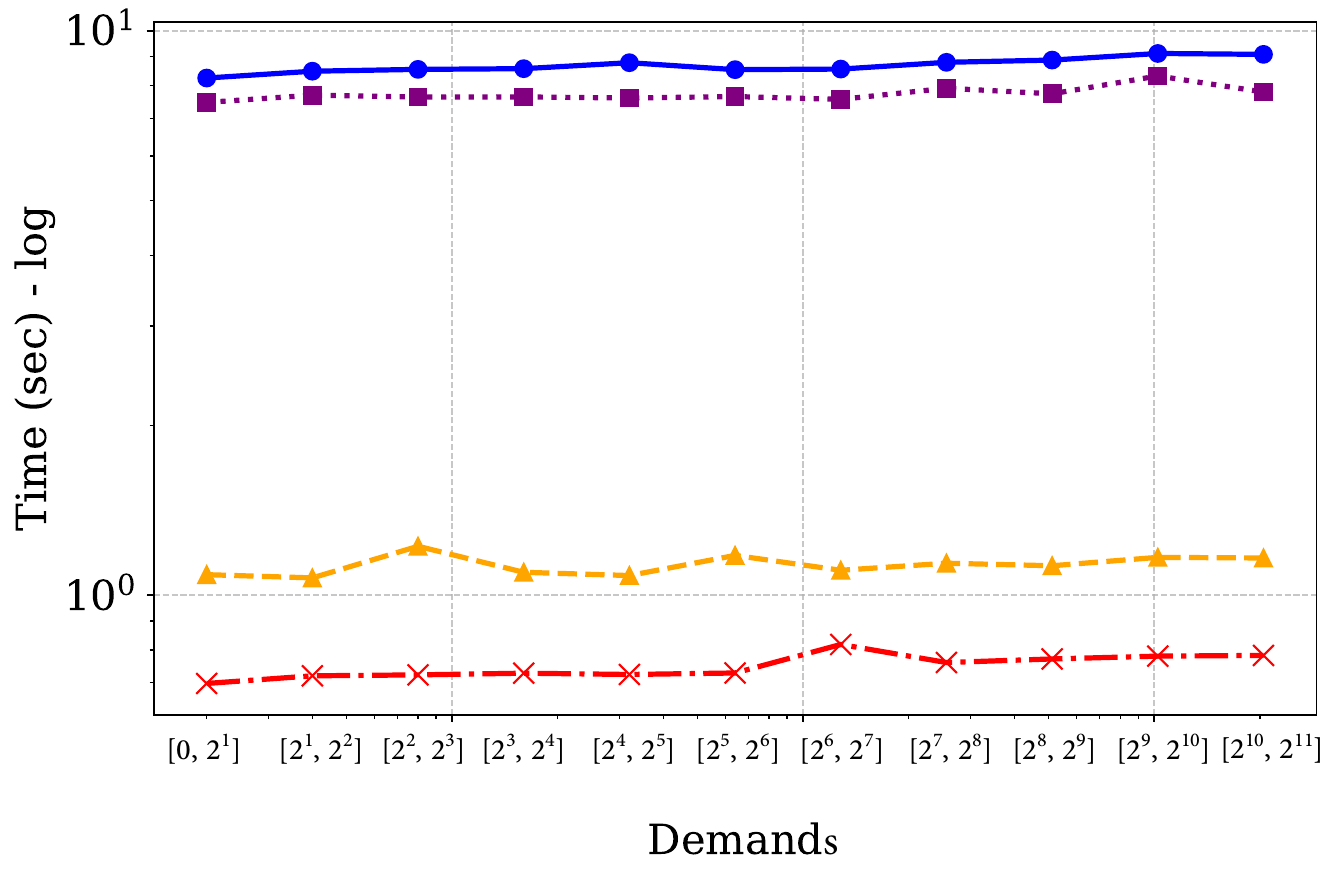}
        \vspace{-2.5em}
        \caption{\revtwo{\small Impact of varying demands on the running time, {\sf Census}}}
        \label{fig:time_varying_demands_census}
    \end{minipage}
    \hfill
    \begin{minipage}[t]{0.32\linewidth}
        \centering
        \includegraphics[width=\textwidth]{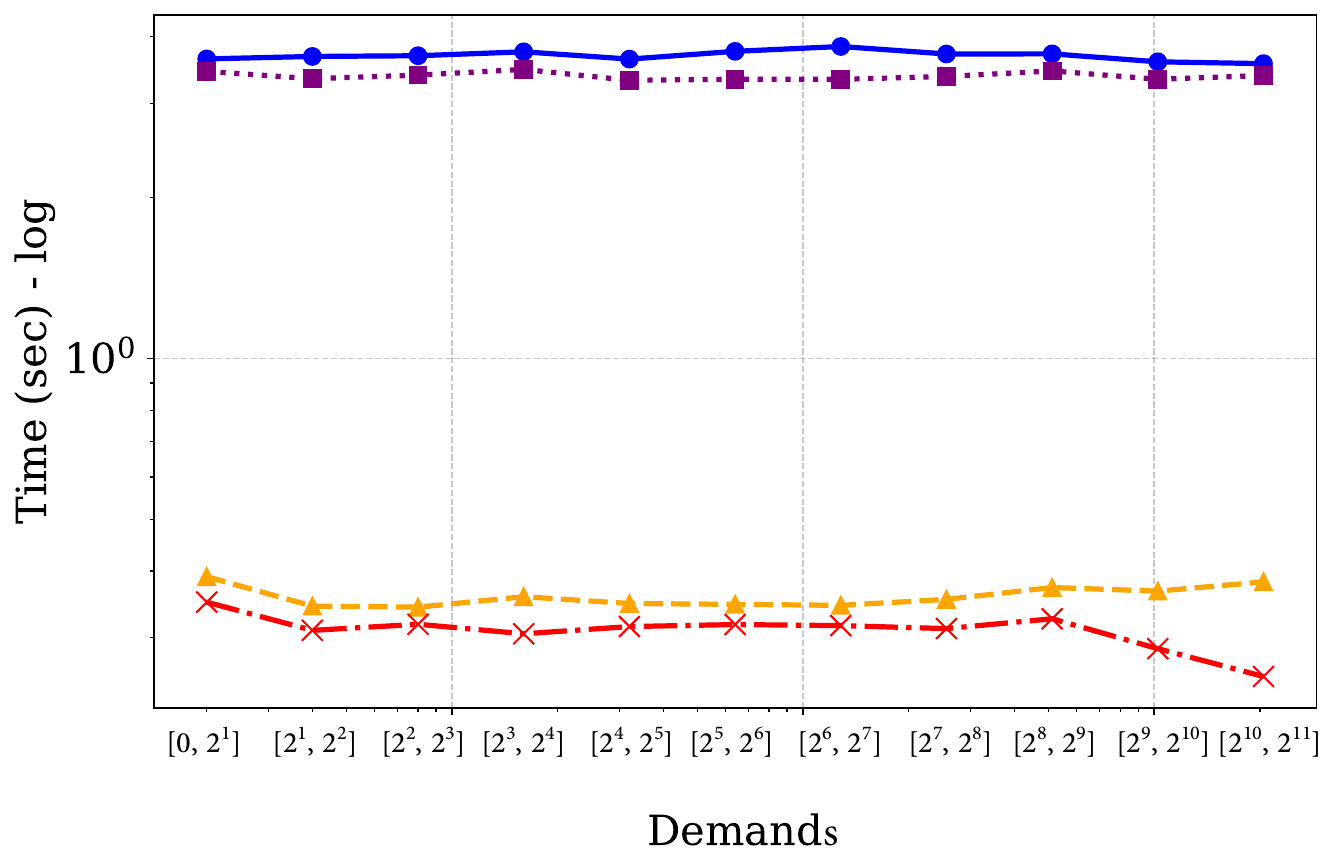}
        \vspace{-2.5em}
        \caption{\revtwo{\small Impact of varying demands on the running time, {\sf Music}}}
        \label{fig:time_varying_demands_music}
    \end{minipage}
\end{figure*}
\vspace{-2mm}
\subsection{Results}
\label{sec:ExpRes}
In this section, we present the results and key findings from our experiments on real datasets, compared to the two baselines. 
\vspace{-0.7em}
\subsubsection{Effectiveness}
We begin by evaluating the effectiveness of the algorithms as the number of items increases. The results for the different datasets are presented in Figures~\ref{fig:weight_varying_l_census}, \ref{fig:weight_varying_l_music}, \ref{fig:weight_varying_l_stack}, and~\ref{fig:weight_varying_l_yelp}. The key observation is that the solution weight increases as the number of items grows. This is intuitive, since a larger number of items requires selecting more sets to satisfy all demands, thereby increasing the total weight. Our algorithms consistently outperform or match the two baselines across all datasets while largely coinciding with each other. \revone{Notably, as shown in Appendix~\ref{appndx:experiments}, our results closely align with those obtained by the exact DP algorithm.}

Then, we evaluate how much each approach deviates from the requested demands as the number of items increases. This deviation is measured using the residual sum of squares (RSS), defined as the sum of squared differences between the original demands and the covered demands across all items. \revtwo{An example demonstrating why this is an important measure for fairness applications of our work is given in Appendix~\ref{appndx:experiments}}.
The results are presented in Figures~\ref{fig:rss_varying_l_census}, ~\ref{fig:rss_varying_l_music}, ~\ref{fig:rss_varying_l_stack}, and ~\ref{fig:rss_varying_l_yelp}. The results indicate that the greedy algorithm consistently exhibits higher RSS, reflecting its tendency to select more sets than needed from each item compared to our algorithms. We also observe that RR-LP sometimes matches our algorithms but significantly underperforms in others. 


In the next experiment, we examine the impact of varying the number of sets $n$ on the total weight of the solution for each algorithm, noting that the objective is to select the solution with the lowest weight while satisfying the demands. The results for this experiment are illustrated in Figures~\ref{fig:weight_varying_n_census},~\ref{fig:weight_varying_n_music}, ~\ref{fig:weight_varying_n_stack} and ~\ref{fig:weight_varying_n_yelp}. The first key observation is that the greedy baseline consistently underperforms compared to our two algorithms.
In some cases, the Greedy baseline returns a solution with total weight $1.3\times$ larger than the total weight returned by our algorithms. RR-LP, as in the previous experiment, can either match or fall short.
Since our two algorithms are essentially the same—one being a faster variant of the other—they produce nearly identical results across most settings. In a few rare instances, the $(2+\epsilon)$ algorithm achieves slightly higher quality solutions, though the improvement is marginal and practically negligible; in other cases, it performs slightly worse, reflecting the influence of $\epsilon$.
The next observation is a consistent decreasing trend in the total weight as $n$ increases across different datasets. This is primarily because additional sets provide more opportunities to satisfy the demands at lower cost, leading to lower-weight solutions. In the case of {\sf Census} shown in Figure~\ref{fig:weight_varying_n_census}, the initial increasing trend arises because the demands cannot be satisfied if all the demands are generated at random, and hence some of them are manipulated as described in Section~\ref{sec:params}, to keep the problem feasible. 

Next, we examine the impact of demand distributions on the effectiveness of our algorithms. Specifically, we generate demands according to six different distributions: Random, Uniform, Normal, Exponential, Poisson, and Zipfian. Owing to space constraints and the similarity of the results across datasets, we only report the results for the {\sf Stack Overflow} dataset. The distribution parameters are defined as follows:
Random: $\Qu \sim \text{Uniform}(1,10)$, Uniform: $\Qu = 5$, Normal: $\Qu \sim \mathcal{N}(\mu=5, \sigma^2=4)$, Exponential: $\Qu \sim \text{Exp}(\lambda=\tfrac{1}{5})$, Poisson: $\Qu \sim \text{Poisson}(\lambda=5)$, Zipfian: $\Qu \sim \text{Zipf}(s=2)$. 
Our findings indicate that the solution weight is largely oblivious to the choice of demand distribution, as shown in Figure~\ref{fig:weight_varying_dist_stack}. \revtwo{Finally, we study the impact of demand magnitude on the effectiveness of our algorithms. To this end, we generate demand values uniformly at random within a specified range and vary this range from small to large values as previously explained. The results for two of the datasets are shown in Figures~\ref{fig:weight_varying_demands_stack} and~\ref{fig:weight_varying_demands_yelp}. Due to space constraints, the results for the other two datasets are shown in Appendix~\ref{appndx:experiments}, and they follow a similar pattern. It can be observed that as demand values increase, the total weight of the solution increases, since the solution needs to choose more sets to satisfy the greater demands. Moreover, the results show that regardless of demand magnitude, our approaches consistently outperform the greedy baseline in terms of solution quality. The comparison between the greedy algorithm and our approaches seems to remain consistent regardless of demand size. In contrast, while the RR-LP baseline generally produces weaker solutions than our algorithms, its solution's quality improves as demands increase. For very large demand values, RR-LP approaches the quality of our methods, but it remains consistently worse than our 2-approximation algorithm. In some instances with large demands, RR-LP outperforms our $(2+\epsilon)$-approximation algorithm; however, RR-LP is always slower by multiple orders of magnitude, as shown in Figures~\ref{fig:time_varying_demands_stack}, \ref{fig:time_varying_demands_yelp}, \ref{fig:time_varying_demands_census}, and \ref{fig:time_varying_demands_music}.
}

Overall, in every experiment, both of our algorithms outperform the baselines in solution quality, with particularly consistent superiority over the widely used greedy approach. Although the differences may not always be large, it is important to note that the weights in these experiments are mainly random, which works against our approach. In some real-world scenarios, where weights are more likely to correlate meaningfully with the number of items a set contains, this performance gap could be significantly larger. An example of such cases is discussed in Section~\ref{sec:case_study}, showing the large gap between the quality of results returned by the greedy algorithm and our algorithms. As a rule of thumb, when solution quality is crucial, our approaches represent the only viable option.

\vspace{-1em}
\subsubsection{Efficiency}
We next turn our attention to the efficiency of our algorithms and compare them with the two baselines. We begin by examining the impact of increasing the number of items on running time. The results are shown in Figures~\ref{fig:time_varying_l_census}, \ref{fig:time_varying_l_music}, \ref{fig:time_varying_l_stack}, and~\ref{fig:time_varying_l_yelp}. The first observation is that the 2-approximation algorithm and the RR-LP baseline are the slowest, with running times roughly an order of magnitude higher than the other two. The $(2+\epsilon)$-approximation algorithm again closely matches the efficiency of the greedy algorithm. 

Next, we examine how increasing the number of sets affects the running time. The results are presented in Figures~\ref{fig:time_varying_n_census}, \ref{fig:time_varying_n_music}, \ref{fig:time_varying_n_stack}, and~\ref{fig:time_varying_n_yelp}. As the number of sets increases, the running time of all algorithms grows. Similar to the previous experiment, our 2-approximation algorithm, followed by RR-LP, are the slowest among the four methods. The $(2+\epsilon)$-approximation algorithm closely matches the efficiency of the greedy; however, the greedy algorithm typically exhibits a steeper growth rate, leading to slightly lower running times for smaller values of $n$ but higher running times as $n$ increases. 

This demonstrates the scalability of our algorithm. In particular, our optimization for the $(2+\epsilon)$-approximation proves highly effective, achieving a substantial reduction in running time—comparable to the greedy approach—while delivering superior solution quality.

\revtwo{Finally, we evaluate the impact of varying the magnitude of demands on the efficiency of the algorithms. The results, shown in Figures~\ref{fig:time_varying_demands_stack}-\ref{fig:time_varying_demands_music}, confirm that the performance of all approaches is almost independent of the demand values over all the datasets.}

\stitle{Summary of results} Our main result—the $(2+\epsilon)$-approximation algorithm—consistently outperforms the baselines across all settings, in terms of effectiveness, efficiency, or both. When effectiveness is comparable to that of the baselines, our approach holds a clear advantage in efficiency. Conversely, when efficiency is similar, our method surpasses the baselines in effectiveness. While the running time of our algorithm is comparable to that of the greedy baseline, it consistently delivers higher solution quality. Conversely, although its solution quality is similar to the RR-LP baseline and the 2-approximation algorithm, our method achieves substantially better running time. In summary, our algorithm offers the best of both worlds, combining high solution quality with efficiency.

\vspace{-1.0em}
\subsection{Case Study: An Adversarial Instance}\label{sec:case_study}
Finally, we highlight a scenario where the weight of the sets is proportional to the number of items they include and demonstrate that our algorithms substantially outperform the greedy baseline in terms of the solution quality. For example, the expected salary of software engineers can be correlated with the number of programming languages they know. \revtwo{The primary goal of these experiments is to highlight the fact that in some scenarios, the quality of the solution produced by the greedy baseline can be significantly worse than what our algorithms achieve.} To this end, we construct a synthetic dataset with $\ell$ items and $n=\ell$ sets. The sets are constructed such that the first set contains only the first item, the second set contains the first and second items, and so forth, until the $\ell$-th set, which includes all items. Additionally, we assign weights to the sets as follows: the first set is assigned a weight of $\tfrac{1}{\ell}$, the second set a weight of $\tfrac{1}{\ell-1}$, the third set a weight of $\tfrac{1}{\ell-2}$, continuing in this manner up to the $\ell - 1$-th set. We set the weight of the $\ell$-th set to $1.01$. We further set the demand of each item to 1. We begin by examining the effect of increasing the number of items (both $\ell$ and $n$ in this case) on the total solution weight. As illustrated in Figure~\ref{fig:case_study_1}, the total weight of the sets selected by the greedy baseline is up to twice that of the solution produced by our algorithms. This is because the greedy algorithm selects sets based on the best cover-to-weight ratio, which implies that the $\ell$-th set is not chosen in the early stages. This not only leads to a suboptimal solution in terms of total weight but also causes the greedy solution to deviate unnecessarily from the required demands by up to a factor of five, as illustrated in Figure~\ref{fig:case_study_2}. In contrast, both of our algorithms and the RR-LP baseline select only the $\ell$-th set as the solution, which in this case is optimal and satisfies all demands without deviation. This further becomes clear by the measured RSS between the required and covered demands, as illustrated in Figure~\ref{fig:case_study_3}.

\footnotetext{Smaller $\ell$ selected for clarity of visualization.}

\section{Related Work}
\label{sec:relwork}
The \emph{set cover} problem is one of the most extensively studied problems in computer science. The classical greedy algorithm achieves an $O(\log \ell)$-approximation~\cite{chvatal1979greedy}, which is asymptotically optimal~\cite{dinur2014analytical}. Numerous variations of set cover have also been explored, including partial set cover~\cite{ran2022approximation}, set multi-cover problem~\cite{dobson1982worst}, and multi-set multi-cover problem~\cite{hua2009exact}. Generally, $O(\log \ell)$-approximation algorithms are known for broad classes of covering and packing problems via integer programming methods~\cite{kolliopoulos2005approximation,chekuri2019approximating,srinivasan1999improved}.

For the \emph{set multi-cover problem} (unbounded universe), which generalizes the \prob\ problem, the greedy algorithm of Dobson~\cite{dobson1982worst} achieves an $O(\log \ell)$-approximation, with the same guarantee holding in the weighted case. In geometric settings, where the universe $\Gee$ consists of points in $\Re^d$ (for constant $d$) and sets correspond to geometric objects such as balls or halfspaces, $O(\log \mathsf{OPT})$- or even constant-factor approximations are known for both the unweighted~\cite{chekuri2012set} and weighted cases~\cite{chekuri2020fast,raman2020improved}.

\newcommand{\poly}{\mathsf{poly}}
Specifically for the \prob\ problem, Bredereck et al.~\cite{bredereck2020mixed} provide an exact algorithm running in $O(h(\ell)\cdot\poly(n))=O(\poly(n))$ time, where $h(\ell)$ is a function exponential on $\ell$, and $\poly(n)$ is a polynomial function with respect to $n$. However, this approach is of limited practical use, as the polynomial factors are not explicitly analyzed and the constants involved appear prohibitively large. On the other hand, the greedy algorithm for weighted set multi-cover (unbounded universe)~\cite{dobson1982worst} directly yields an $O(\log \ell)=O(1)$-approximation for the \prob\ problem in $O(n\log n)$ time.

A closely related problem is \emph{maximum coverage}, where the goal is to select $k$ sets from $\Dee$ that maximize the number of covered elements. The greedy algorithm guarantees a $(1-\tfrac{1}{e})$-approximation~\cite{nemhauser1978analysis}. Improved algorithms exist for set systems of bounded VC-dimension~\cite{badanidiyuru2012approximating}. Several extensions have been investigated, including maximum coverage with constraints~\cite{chekuri2004maximum}, the maximum multi-coverage problem~\cite{barman2022tight}, and the minimum $p$-union problem~\cite{ran2022approximationmin}.

\begin{envrevfour}
Several problems studied in the database community can be naturally modeled as instances of \prob\ and can therefore benefit from our algorithms. Below, we briefly discuss two representative application domains and focus on how prior work differs from our setting.

{\bf Crowd-sourcing and task assignment.}
A large body of work in databases formulates crowd-sourcing and task-assignment problems as selecting a set of workers, each associated with a cost and a set of skills, to satisfy coverage requirements over tasks or skill types~\cite{crowd6, Crowd2, Crowd3, crowd4, crowd5, crowd7}. These formulations correspond to weighted set multi-cover instances, where workers define sets, skills, or task types define universe elements, and demands capture the required number of assignments. In many practical settings, however, the number of skills or task types is small (e.g., domains in Example~\ref{ex:CS} or programming languages in Example~\ref{ex3}), making these problems well-suited for our \prob\ framework.

{\bf Fairness-aware data selection.}
Fairness-aware data management and ML pipelines often require selecting a subset of data that satisfies representation constraints over protected attributes~\cite{kurkure2024faster, asudeh2023maximizing, Nargesian2021fair, stoyanovich2018fair, shahbazi2024fairness, chameleon, fair1, fair2, fair3, fair4, fair5}. 
Such problems can be modeled as weighted set multi-cover instances and variations of it.
The number of protected attributes is usually small, which naturally aligns these problems with the bounded-universe setting of \prob. For example, there are 5 standard categories for gender and 7 standard categories for ethnicity in Example~\ref{ex1}.


We emphasize that these are representative examples. Due to the generality of the set cover framework, \prob\ also applies more broadly to other data-driven decision-making tasks such as constrained recommendation problems~\cite{serbos2017fairness, qi2016recommending}.

Most existing approaches in crowd-sourcing, fairness-aware data selection, and package recommendation employ greedy heuristics or LP-based techniques for covering or packing problems \cite{serbos2017fairness, chameleon, shahbazi2023representation}. Some of these approaches \cite{shahbazi2023representation, chameleon} can be directly formulated as instances of \prob, while others consider closely related variants with different objectives or constraints \cite{crowd5, serbos2017fairness}.
Independent of these modeling choices, none of the existing methods exploits a bounded universe size.
To the best of our knowledge, we are the first to study the \emph{weighted set multi-cover} problem on a bounded universe within the database community. While this problem has been studied in the theory community~\cite{bredereck2020mixed}, no efficient $2$-approximation (or better) algorithm was previously known.
Our approximation algorithms crucially leverage the bounded universe size to achieve strictly better approximation guarantees for \prob while remaining efficient. As a result, our approach is fundamentally different from both the classical greedy algorithm and standard LP-based solutions. For instance, although the LP formulation in LP~\ref{FinalLPobfunc} resembles standard formulations such as~\cite{dobson1982worst}, we introduce a novel rounding technique that is efficient only under the bounded-universe assumption. Furthermore, we establish a non-trivial connection to Problem~\ref{LPobfunc}, which enables us to accelerate the algorithm in Section~\ref{subsec:smalllp} by approximating its objective function with a sum of piecewise-linear functions.
\end{envrevfour}


\vspace{-1em}
\section{Conclusion}
In this paper, we study the \prob\ problem and propose three algorithms: an exact dynamic programming algorithm, a $2$-approximation algorithm, and a more efficient $(2+\eps)$-approximation algorithm. Our methods outperform the Greedy and the standard LP-rounding baselines both theoretically (achieving better approximation guarantees with faster running times) and empirically, as demonstrated through experiments on real and synthetic datasets.
Several interesting directions remain open. A key question is whether a polynomial-time $(1+\eps)$-approximation algorithm exists for the \prob\ problem, for any $\eps>0$. Another interesting problem is to study the \prob\ problem in streaming or dynamic settings. We believe our results provide a solid foundation for exploring these extensions in the future.


\bibliographystyle{ACM-Reference-Format}
\bibliography{ref}
\newpage
\clearpage
\appendix


\begin{envrevone}
\section{Missing details from Section~\ref{sec:opt}}
\label{appndx:DPpseudocode}
\revone{The pseudo-code of the exact DP algorithm is shown in Algorithm~\ref{alg:dp}}.

\begin{algorithm}
    \caption{DP Algorithm for {\sc \prob}}
    \label{alg:dp}
    \begin{algorithmic}[1]
        \Require $\Dee, \Gee, \Qu$
        \Ensure result set $\eh$
        \State $\textsf{DP}[0\ldots n][0\ldots n]\ldots[0\ldots n]\gets \infty$ 
        \State $\textsf{DP}[0\ldots n][0]\ldots[0]\gets 0$
        \For{$i=1$ \textbf{to} $n$}
            \For{$q_1=0$ \textbf{to} $n$}
            \State \vdots
                \For{$q_{\ell}=0$ \textbf{to} $n$}
                    \State $\textsf{DP}[i][q_1]\ldots[q_{\ell}] \gets \min\{\mathsf{DP}[i-1][q_1]\ldots[q_\ell],$
             \Statex \hspace{8em}$w_i\!+\!\mathsf{DP}[i\!-\!1][\max\{0,q_1\!-\!\mathbbm{1}\![g_1\!\in\! t_i]\}]\ldots$
             \Statex \hspace{13.3em}$[\max\{0,q_\ell\!-\!\mathbbm{1}\![g_\ell\!\in\!t_i]\}]\}$       
                \EndFor
            \EndFor
        \EndFor
        \State $(j_1^*,\ldots, j_\ell^*)=\argmin_{j_1\geq \Qu_1, \ldots, j_\ell\geq \Qu_\ell}\mathsf{DP}[n][j_1]\ldots[j_\ell]$
        \If{$\textsf{DP}[n][j_1]\ldots[j_\ell]=\infty$}
            \State \Return $-1$ {\tt \quad //no solution exists}
        \Else
            \State $\eh\gets\textsc{selectOptSets($\min(\textsf{DP}[n][j_1]\ldots[j_\ell])$)}$ {\tt \quad //traces opt path and returns the selected sets.}
            \State \Return $\eh$
        \EndIf
         
    \end{algorithmic}
\end{algorithm}
\end{envrevone}

\section{Missing proofs from Section~\ref{subsec:biglp}}
\label{appndx:sec4a}

\subsection{Proof of Lemma~\ref{lem:1}}
\label{appndx:lem1}
    For (i), notice that $\hat{f}_H(x)=f_H(x)$ when $x\in \mathbb{Z}$, while the function $f_H(x)$  is not defined for $x\notin\mathbb{Z}$. Furthermore, the constraints of IP~\eqref{obfunc} and its relaxed Problem~\eqref{LPobfunc} are the same. For every $H\subseteq \Gee$, let $x_H^*$ be the value of the variable $x_H$ in the optimum solution optimum of IP~\eqref{obfunc}. By definition, $\hat{\mathsf{opt}}\leq \sum_{H\subseteq \Gee}\hat{f}_H(x_H^*)=\sum_{H\subseteq \Gee}f_H(x_H^*)=\mathsf{opt}$.

    We show (ii).
    The function $\hat{f}_H$ is continuous since $\lim_{x\rightarrow z^-}\hat{f}_H(x)=\lim_{x\rightarrow z^+}\hat{f}_H(x)=\hat{f}_H(z)$, for every $z\in \mathbb{Z}$. 
    The function $\hat{f}_H$ is piecewise linear because for every interval $z\leq x<z+1$ for a $z\in \mathbb{Z}$, the function is defined as a linear function with respect to $x$.
    Finally, we show that $\hat{f}_H$ is convex showing that the slopes of the linear functions are non-decreasing. We fix $H\subseteq \Gee$. For every $i\in\{1,\ldots, |B(H)|\}$, let $\alpha_i$ be the slope of $\hat{f}_H$ in the $i$-th interval $z\leq x<z+1$, i.e., $i-1\leq x<i$. By definition $\alpha_i=f_H(i)-f_H(i-1)$. Recall the definition of $f_H(i)$: It is the sum of weights of the $i$ sets in $B(H)$ with the smallest weights. Hence $\alpha_i$ is defined as the $i$-th smallest weight among the weights of the sets in $B(H)$, i.e., $a_i=w_{H,i}$ for every $i\in[|B[H]|]$.
    Thus, $\alpha_1\leq \alpha_2\leq \ldots\leq \alpha_{|B(H)|}$, proving that the function $\hat{h}_H$ is convex. Since $\hat{f}_H$ is piecewise linear, convex and every slope $\alpha_i\geq 0$, it also follows that it is non-decreasing. 
    An example of a function $\hat{f}_H$ is shown in Figure~\ref{fig:pl-table-stacked}.

For iii), we show that LP~\eqref{FinalLPobfunc} and Problem~\eqref{LPobfunc} are equivalent.
For each $H\subseteq \Gee$, let $\hat{x}_H$ be the value of variable $x_H$ on an optimum solution of Problem~\eqref{LPobfunc}, and let $\hat{Y}=\sum_{H\subseteq \Gee}\hat{f}_H(\hat{x}_H)$ be the value of the objective function.
We show that there exists a feasible solution in LP~\eqref{FinalLPobfunc} with an objective value equal to $\hat{Y}$.
For every $H\subseteq \Gee$, let $\rho(H)=\lfloor \hat{x}_H\rfloor$. For every $i\in \{1,2,\ldots, \rho(H)\}$ we set $\hat{x}_{H,i}=1$, for every $i\in\{\rho(H)+2,\ldots, |B(H)|\}$ we set $\hat{x}_{H,i}=0$, while for $i=\rho(H)+1$ we set $\hat{x}_{H,\rho(H)+1}=\hat{x}_H-\rho(H)$.
Notice that, by definition, for every $H\subseteq \Gee$, $\hat{x}_H=\sum_{i\in[|B(H)|]}\hat{x}_{H,i}$. Hence, all constraints in LP~\eqref{FinalLPobfunc} hold. Furthermore, for any $H\subseteq\Gee$, we have $\hat{f}_H(\hat{x}_H)=(f_H(\rho(H)+1)-f_H(\rho(H)))\cdot (\hat{x}_H-\rho(H))+f_H(\rho(H))=w_{H,\rho(H)+1}\cdot \hat{x}_{H,\rho(H)+1}+\sum_{i\in[\rho(H)]}w_{H,i}=\sum_{i\in[|B(H)|]}w_{H,i}\cdot \hat{x}_{H,i}$.
Thus, 
$\sum_{H\subseteq \Gee}\sum_{i\in[|B(H)|]}w_{H,i}\cdot \hat{x}_{H,i}=
\sum_{H\subseteq \Gee}\hat{f}_H(\hat{x}_H)=\hat{Y}$.

Next, we show the other direction:
For each $H\subseteq \Gee$ and $i\in[|B(H)|]$, let $\hat{x}_{H,i}$ be the value of variable $x_{H,i}$ \revtwo{an optimum} solution of LP~\eqref{FinalLPobfunc}, and let $\hat{Y}=\sum_{H\subseteq \Gee}\sum_{i\in[|B(H)|]}w_{H,i}\cdot \hat{x}_{H,i}$ be the value of the objective function.
We show that there exists a feasible solution in Problem~\eqref{LPobfunc} with an objective value equal to $\hat{Y}$.
We make the following observation. For every $H\subseteq \Gee$, and for any $i\in[|B(H)|]]$, if $0<\hat{x}_{H,i}<1$ then $\hat{x}_{H,j}=1$ for every $0<j<i$, and $\hat{x}_{H,j}=0$ for every $|B(H)|\geq j>i$. We show it by contradiction. Assume that there exists an index $j'<j$ such that $\hat{x}_{H,j'}<1$. Let $\mu=\min\{1-\hat{x}_{H,j'}, \hat{x}_{H,j}\}$. We set the values $\bar{x}_{H,j'}=\hat{x}_{H,j'}+\mu$, $\bar{x}_{H,j}=\hat{x}_{H,j}-\mu$, and for all the other values $\bar{x}_{H,s}=\hat{x}_{H,s}$ for $s\in[|B(H)|]\setminus\{j',j\}$. We notice that $\vv{\bar{x}}$ satisfies all constrains in LP~\eqref{FinalLPobfunc} since $\sum_{i\in[|B(H)|]}\hat{x}_{H,i}=\sum_{i\in[|B(H)|]}\bar{x}_{H,i}$. Furthermore, $\sum_{i\in [|B(H)|]}w_{H,i}\bar{x}_{H,i}\leq \sum_{i\in [|B(H)|]}w_{H,i}\hat{x}_{H,i}=\hat{Y}$ because $w_{H,j'}\leq w_{H,j}$ and $w_{H,j'}\bar{x}_{H,j'}+w_{H,j}\bar{x}_{H,j}\leq w_{H,j'}\hat{x}_{H,j'}+w_{H,j}\hat{x}_{H,j}$. We conclude that if $w_{H,j'}<w_{H,j}$, then this is a contradiction because $\hat{Y}$ is the value of the optimum solution. If $w_{H,j'}=w_{H,j}$ we can always construct in $O(n)$ time a solution with the same value $\hat{Y}$ satisfying the observation.
Hence from now on, we assume that for every $H\subseteq \Gee$, there exists an index $\rho(H)\in[|B(H)|]$ such that $\hat{x}_{H,i}=1$ for $i\leq \rho(H)$, and $\hat{x}_{H,i}=0$ for $i\geq \rho(H)+2$.
For every $H\subseteq \Gee$, we define $\hat{x}_H=\sum_{i\in[|B(H)|]}\hat{x}_{H,i}$, so by definition, all constrains of Problem~\eqref{LPobfunc} are satisfied.
Furthermore, $\sum_{i\in [|B(H)|]}w_{H,i}\hat{x}_{H,i}=\sum_{i\leq \rho(H)}w_{H,i} + w_{H,\rho(H)+1}\cdot\hat{x}_{H,\rho(H)+1}=\hat{f}_H(\hat{x}_H)$. We conclude that $\sum_{H\subseteq\Gee}\hat{f}_H(\hat{x}_H)=\sum_{H\subseteq \Gee}\sum_{i\in [|B(H)|]}w_{H,i}\cdot\hat{x}_{H,i}=\hat{Y}$. The result follows. 

\subsection{\textbf{Proof of Theorem~\ref{thm:1}}}
\label{appndx:thm2}

We first show that 
   Algorithm~\ref{alg:lprounding} is a 2-approximation algorithm for the \prob\ problem.

    For all $H \subseteq \Gee$, let $\opt_H$ denote the number of sets selected from $B(H)$ in the optimal solution of the \prob\ problem over $\dee$, and let $p_H$ denote the the number of the sets we choose from $B(H)$, in the returned solution $\sol$. The goal is to show that $\sum_{H \in \Gee}f_H(p_H) \leq 2\sum_{H \in \Gee}f_H(\opt_H)$, to prove that Algorithm~\ref{alg:lprounding} returns a $2$-approximation solution. For all $H \subseteq \Gee$, we define $e_H = \max\{\bar{x}_H, \opt_H\}$. It is easy to verify that if we choose $e_H$ sets (with the smallest weight) from every bucket $B(H)$, all the demands will get satisfied, i.e., for all items $g \in \Gee$, we have $\sum_{H \subseteq \Gee, g \in H}e_H \geq Q_g$. This is because for all $H \subseteq \Gee$, we have that $e_H \geq \opt_H$, by definition. Moreover, we have
    \begin{align}\label{algn:ieq}
    \sum_{H \subseteq \Gee}f_H(e_H) \leq \sum_{H \subseteq \Gee}f_H(\bar{x}_H) + \sum_{H \subseteq \Gee}f_H(\opt_H)
    \leq 2\sum_{H \in \Gee}f_H(\opt_H).
    \end{align}
    Recall that $r = \left\lceil\sum_{H \subseteq \Gee}(\hat{x}_H - \bar{x}_H)\right\rceil$. For all $H \subseteq \Gee$, we define $e'_H = \min\{e_H, \bar{x}_H + r\}$. By definition, we have $e'_H \leq e_H$, for all $H \subseteq \Gee$, and hence $\sum_{H \subseteq \Gee}f_H(e'_H) \leq \sum_{H \subseteq \Gee}f_H(e_H)$. Therefore, by (\ref{algn:ieq}), we get 
    \begin{align}\label{algn:ieq2}
        \sum_{H \subseteq \Gee}f_H(e'_H) \leq 2\sum_{H \subseteq \Gee}f_H(\opt_H).
    \end{align}
    Next, we show that for all $H \subseteq \Gee$, we have that $\bar{x}_H$ and $e'_H$ are close to each other and hence the described rounding scheme Algorithm~\ref{alg:lprounding}, will consider the vector $\vv{e'}_{H\subseteq \Gee}$ as a candidate solution. 
    \begin{lemma}\label{lem:helper1}
        For all $H \subseteq \Gee$, it holds that $0 \leq e'_H - \bar{x}_H \leq r$.
    \end{lemma}
    \begin{proof}
        Recall that by definition, $e_H = \max\{\bar{x}_H, \opt_H\}$ and $e'_H = \min\{e_H, \bar{x}_H + r\}$. Thus, we have $e'_H = \min\{\max\{\bar{x}_H, \opt_H\}, \bar{x}_H + r\}$. Hence, we have $e'_H \geq \bar{x}_H$ and $e'_H \leq \bar{x}_H + r$, for all $H \subseteq \Gee$. The lemma follows.
    \end{proof}
    
    The following lemma shows that if we choose $e'_H$ sets from the bucket $B(H)$, for all $H \subseteq \Gee$, then all the demands are satisfied.
    \begin{lemma}\label{lem:helper2}
        For all items $g \in \Gee$, it holds that $\sum_{H \subseteq \Gee, g \in H}e'_H \geq Q_g$.
    \end{lemma}
    \begin{proof}
        By definition, both the fractional vector $\vv{\hat{x}}_{H \subseteq \Gee}$ and the optimal solution vector $\vv{\opt}_{H \subseteq \Gee}$ satisfy all the demands. More formally, for any $g \in \Gee$ we have 
            $\sum_{H \subseteq \Gee, g \in H}\hat{x}_H \geq Q_g$,
        and 
\begin{align}\label{algn:opt}
            \sum_{H \subseteq \Gee, g \in H}\opt_H \geq \revtwo{Q_g}.
        \end{align}
        Recall that $r = \left\lceil\sum_{H \subseteq \Gee}(\hat{x}_H - \bar{x}_H)\right\rceil$, and since $\vv{\hat{x}}_{H \subseteq \Gee}$ covers all the demands, for any $g \in \Gee$ we have
            $r + \sum_{H \subseteq \Gee, g \in H}\bar{x}_H \geq Q_g$.
        Intuitively, if we take $\bar{x}_H$ sets from $B(H)$ for all $H \subseteq \Gee$, then for any item $g \in \Gee$, we need at most $r$ additional sets to satisfy its demand.
        By definition, $e'_H = \min\{\max\{\bar{x}_H, \opt_H\}, \bar{x}_H + r)\}$. Thus, for all $H \subseteq \Gee$, we have
        \begin{align}\label{algn:egeqx}
            e'_H \geq \bar{x}_H.
        \end{align}
        Moreover, by definition, we have 
        \begin{align}\label{algn:cases}
        \forall H \subseteq \Gee:
            \begin{cases}
                e'_H = \bar{x}_H + r, & \text{if } \opt_H > \bar{x}_H + r, \\[4pt]
                e'_H \geq \opt_H, & \text{if } \opt_H \leq \bar{x}_H + r.
            \end{cases}
        \end{align}
        Let $g \in \Gee$, be any arbitrary item. There are 2 different cases. 
        
        (i) For all sets $H \subseteq \Gee$, such that $g \in H$, it holds that $\opt_H \leq \bar{x}_H + r$. In this case, for all $H \subseteq \Gee$ such that $g \in H$, we have that $e'_H \geq \opt_H$, by the second condition of (\ref{algn:cases}). Therefore, we get
        \begin{align*}
            \sum_{H \subseteq \Gee, g \in H}e'_H \geq \sum_{H \subseteq \Gee, g \in H}\opt_H \geq Q_g,
        \end{align*}
        as desired in this case. The last inequality holds by (\ref{algn:opt}).
        
        (ii) There exists a set $H' \subseteq \Gee$, such that $g \in H'$ and $\opt_{H'} \geq \bar{x}_{H'} + r$. In this case, by the first condition of (\ref{algn:cases}), we have $e'_{H'} = \bar{x}_{H'} + r$. Moreover, for all $H \subseteq \Gee$, it holds that $e'_H \geq \bar{x}_H$, by (\ref{algn:egeqx}). Thus, we get 
            $\sum_{H \subseteq \Gee, g \in H}e'_H \geq r + \sum_{H \subseteq \Gee, g \in H}\bar{x}_H \geq Q_g$,
        as desired. Thus, the lemma follows in both cases. 
    \end{proof}
    
    For all $H \subseteq \Gee$, by Lemma~\ref{lem:helper1} we have that $\bar{x}_H \leq e'_H \leq \bar{x}_H + r$. Therefore, the described rounding scheme Algorithm~\ref{alg:lprounding}, will consider the vector $\vv{e'}_{H\subseteq \Gee}$ as a candidate. By Lemma~\ref{lem:helper2}, we have that the vector $\vv{e'}_{H\subseteq \Gee}$ satisfies all the demands, and thus the condition in line 19 of Algorithm~\ref{alg:lprounding} is true for this vector, and it is qualified as a potential solution.
    Hence, we have $\sum_{H \subseteq \Gee}f_H(p_H) \leq \sum_{H \subseteq \Gee}f_H(e'_H)$. Putting everything together, by (\ref{algn:ieq2}), we get
        $\sum_{H \subseteq \Gee}f_H(p_H) \leq 2\sum_{H \subseteq \Gee}f_H(\opt_H)$,
    as desired.

Next, we bound the running time of Algorithm~\ref{alg:lprounding}.
\begin{lemma}
    Algorithm~\ref{alg:lprounding} runs in $O(\mathcal{L}(n))$ time.
\end{lemma}
\begin{proof}
The LP~\eqref{FinalLPobfunc} has $O(n)$ variables and constraints (Lemma~\ref{lem:1}), so its optimum solution is computed in $O(\mathcal{L}(n))$ time. Then we bound the time for the rounding procedure.
For all $H \subseteq \Gee$, we have $\hat{x}_H - \bar{x}_H < 1$, and hence $r = \left\lceil\sum_{H \subseteq \Gee}(\hat{x}_H - \bar{x}_H)\right\rceil \leq 2^\ell$. Since there are $2^\ell$ different buckets and from each bucket we consider at most $r$ sets (with the smallest weight), the rounding step takes $O(r^{2^\ell}) = O((2^\ell)^{(2^\ell)}) = O(1)$ time. 
Although one might argue that $O(r^{2^\ell})$ represents a large constant that could slow down the rounding procedure in practice, in Section~\ref{sec:exp} we show that $r$ is typically quite small in practice.
Overall, the algorithm runs in $O(\mathcal{L}(n))$ time.
\end{proof}

\section{Missing proofs from Section~\ref{subsec:smalllp}}
\label{appndx:sec4b}
\begin{proofoflemma}{lem:approx}
    Let $x \in [\alpha, \beta]$, be an arbitrary value. Let $i \in [s-1]$, be an index such that $x_i \leq x \leq x_{i+1}$. By convexity of $\ef(\cdot)$, we have that 
\begin{align}\label{algn:base}
    \ef(x) \leq \gi(x).
\end{align}
By definition of $x_{i+1}$, we have that
\begin{align}\label{algn:defx}
    \ef(x_{i+1}) \leq (1 + \epsilon)\ef(x_i).
\end{align}
By definition of $\gi(\cdot)$, we have that $\ef(x_{i}) = \gi(x_{i})$ and $\ef(x_{i+1}) = \gi(x_{i+1})$. By linearity of $\gi(\cdot)$ between $x_i$ and $x_{i+1}$ and since $\ef(\cdot)$ is non-decreasing, we have $\gi(x) \leq \gi(x_{i+1})$. Thus,
\begin{align}\label{algn:gileqef}
    \gi(x) \leq \gi(x_{i+1}) = \ef(x_{i+1}).
\end{align}
Since $\ef(\cdot)$ is non-decreasing, we have
\begin{align}\label{algn:efleqef}
    \ef(x_i) \leq \ef(x).
\end{align}
By (\ref{algn:defx}), (\ref{algn:gileqef}), and (\ref{algn:efleqef}), we get
\begin{align}\label{algn:main}
    \gi(x) \leq \ef(x_{i+1}) \leq (1+\epsilon)\ef(x_i) \leq (1+\epsilon) \ef(x).
\end{align}
From (\ref{algn:base}) and (\ref{algn:main}), we get
    $\ef(x) \leq \gi(x) \leq (1+\epsilon) \ef(x)$,
for all values of $x \in [\alpha, \beta]$.

By construction, it is straightforward that $\gi(\cdot)$ is a piecewise linear function.
Since $\ef(\cdot)$ is convex and non-decreasing, the slope of $\gi$ in $[x_i,x_{i+1}]$ is greater than the slope of $\gi$ in $[x_{i-1},x_{i}]$ for every $i\in\{2,3,\ldots, s\}$. 
We show it formally as follows.
The slope of $\gi$ in $[x_{i-1},x_i]$ is defined as $\frac{\gi(x_i)-\gi(x_{i-1})}{x_i-x_{i-1}}$ and the slope of $\gi$ in $[x_{i},x_{i+1}]$ is defined as $\frac{\ef(x_{i+1})-\ef(x_{i})}{x_{i+1}-x_{i}}$. We show, $\frac{\ef(x_i)-\ef(x_{i-1})}{x_i-x_{i-1}}\leq \frac{\ef(x_{i+1})-\ef(x_{i})}{x_{i+1}-x_{i}}$.
We write $x_i$ as a linear combination of $x_{i-1}$ and $x_{i+1}$, i.e., $x_i=\lambda x_{i+1}+(1-\lambda)x_{i-1}$, so $\lambda=\frac{x_i-x_{i-1}}{x_{i+1}-x_{i-1}}\in (0,1)$.
By convexity, $\ef(x_i)\leq \lambda\cdot \ef(x_{i+1})+(1-\lambda)\ef(x_{i-1})=\frac{x_i-x_{i-1}}{x_{i+1}-x_{i-1}}\ef(x_{i+1})+\frac{x_{i+1}-x_i}{x_{i+1}-x_i}\ef(x_{i-1})\Leftrightarrow \frac{\ef(x_i)-\ef(x_{i-1})}{x_i-x_{i-1}}\leq \frac{\ef(x_{i+1})-\ef(x_{i})}{x_{i+1}-x_{i}}$.
Hence, the function $\gi(\cdot)$ is also convex and non-decreasing.

Next, we analyze the number of linear pieces in $\gi(\cdot)$ and the running time of the procedure. Recall that $s$ denotes the number of steps the algorithm takes and hence $\gi(\cdot)$ consists of $s-1$ different pieces. Moreover, by definition, we have $\ef(x_1) = \gamma$ and $\ef(x_s) = \delta$.
Assume the following equality with respect to an unknown $y$, $\ef(\alpha)(1+\eps)^y=\ef(\beta)\Leftrightarrow\gamma(1+\eps)^y=\delta\Leftrightarrow y=\log(\frac{\delta}{\gamma})\cdot \log^{-1}(1+\eps)$.
By the construction of $\gi(\cdot)$, it follows that $s=\left\lceil \log(\frac{\delta}{\gamma})\cdot \log^{-1}(1+\eps) \right\rceil$, so $s=O\left(\log(\frac{\delta}{\gamma})\cdot \log^{-1}(1+\eps)\right)$.
Next, we analyze the running time of the algorithm to construct $\gi(\cdot)$. Let $T_{\ef}^>$ be the running time of the $\mathsf{LargestVal_\ef}(\cdot,\cdot)$ procedure and let $T_{\ef}^=$ is the time needed to calculate $\ef(x)$.
The algorithm takes $O(\log(\frac{\delta}{\gamma}))$ steps. In each step $i\in[s-1]$, we execute the $\mathsf{LargestVal_\ef}(\cdot,\cdot)$ procedure and we evaluate $\ef(x_{i+1})$. Thus, the algorithm runs in $O(\log(\frac{\delta}{\gamma})\cdot \log^{-1}(1+\eps)\cdot 
T_{\ef}^>\cdot T_{\ef}^=)$ time.
\end{proofoflemma}

\begin{proofoflemma}{lem:ApproxFactor}
Let $x_H^*$ be the value of the variable $x_H$ in the optimum solution of Problem~\eqref{smallProb}, that we obtain from solving its corresponding LP. Recall that $\hat{\mathsf{opt}}$ is the value of the optimum solution in Problem~\ref{LPobfunc}. From Lemma~\ref{lem:approx} it follows that $\sum_{H\subseteq \Gee}\hat{f}_H(x_H^*)\leq (1+\eps)\hat{\mathsf{opt}}$. In the proof of Theorem~\ref{thm:1} we showed that the rounding procedure applied on  the optimum solution of Problem~\eqref{LPobfunc} returns a $2$-approximation solution for the \prob\ problem. This proof can be extended straightforwardly arguing that the rounding procedure applied on a $\xi$-approximation solution of Problem~\eqref{LPobfunc} returns, a $(2\xi)$-approximation solution for the \prob\ problem, for any $\xi>1$. Hence we conclude that $\sum_{H\subseteq \Gee}f_H(\tilde{x}_H)\leq 2(1+\eps)\opt$.
\end{proofoflemma}

\begin{proofoflemma}{lem:runtimefaster}
    A direct implementation of the algorithm in Lemma~\ref{lem:approx} for $\ef=\hat{f}_H$ , for all $H\subseteq \Gee$, would lead to the construction of function $\hat{g}_H$ with $O(\log W)$ linear pieces (recall that we assumed $\eps$ to be a constant) in $O(\log (W)\cdot T_{\hat{f}_H}^>\cdot T_{\hat{f}_H}^=)$ time. We show that in our case, where each function $\hat{f}_H$ is a piecewise linear function, we can spend $O(n\log n)$ preprocessing time to execute the algorithm in Lemma~\ref{lem:approx} in only $O(\log (W))$ time. For every $H\subseteq \Gee$, we sort all sets in $B(H)$ in ascending order of their weights. This step takes $O(|B(H)|\log |B(H)|)$ time for each $H\subseteq \Gee$, so overall it takes $O(n\log n)$ time to sort all sets over all buckets $B(H)$ and construct the functions $\hat{f}_H$. For every step of the algorithm in Lemma~\ref{lem:approx} we execute $\mathsf{LargestVal}_{\hat{f}_H}(x_i,(1+\eps)\hat{f}_H(x_i))$ as follows. Starting from $x_i$ we traverse the linear pieces of $\hat{f}_H$ one by one at the right side of $x_i$ until we find the piece defined on the interval $[x_{j-1}, x_j]$ such that $\hat{f}_H(x_j)>(1+\eps)\hat{f}_H(x_i)\geq \hat{f}_H(x_{j-1})$ or until we reach to the last interval. Notice that because of the preprocessing step, for each $x_\zeta$ we check, we evaluate $\hat{f}_H(x_\zeta)$ in $O(1)$ time.
    If we reach the last interval we straightforwardly construct the last linear piece of function $\hat{g}_H$. In the former more interesting case, we define the line that passes between $(x_{j-1}, \hat{f}_H(x_{j-1}))$ and $(x_{j}, \hat{f}_H(x_{j}))$, and find the exact $x'_i$ such that $\hat{f}_H(x'_i)=(1+\eps)\hat{f}_H(x_i)$ in $O(1)$ time. Notice that the algorithm traverses the points $x_\zeta$ in increasing order so for each $H\subseteq\Gee$, the function $\hat{g}_H$ is constructed in $O(\log (W)+|B(H)|)$ time. There are $O(2^\ell)$ subsets in $\Gee$, so the running time to define Problem~\eqref{smallProb} and its corresponding LP is $O(\sum_{H\subseteq \Gee} (\log(W)+|B(H)|))=O(2^\ell\log(W)+n)=O(\log(W) +n)$. The LP has $O(\sum_{H\subseteq \Gee}\log W)=O(2^\ell\log W)=O(\log W)$ variables and constraints so it is solved in $O(\mathcal{L}(\log W))$ time. Finally, the rounding step is identical to Subsection~\ref{subsec:biglp} and runs in $O(1)$.
\end{proofoflemma}

\begin{envrevone}
\section{Additional Experiment Results}\label{appndx:experiments}

In this section, we discuss the additional experiments regarding the evaluation of the DP algorithm and the impact of varying demands over the remaining datasets. At last, we discuss an example to highlight the importance of avoiding over-satisfaction of demands in fairness applications. 
\paragraph{Exact DP algorithm}
Figures~\ref{fig:dp_weight_varying_n_census}–\ref{fig:dp_time_varying_group_yelp} present a comparison of our main algorithms and baseline methods against the exact DP algorithm. In these experiments, we restricted the ranges of $n$ and $\ell$ so that the DP algorithm remains scalable. 
Specifically, we fix the default value of $n$ to $4096$ while varying $\ell$ from 2 to 9 groups, and fix $\ell$ to 7 while varying $n$ from 16 to 32{,}768. These experiments are basically a smaller version of the experiments shown in Section~\ref{sec:exp}, so that we can compare the algorithms against the exact optimal solution. In this smaller setting, the general pattern of results, when varying the number of sets $n$ and the number of items $\ell$, remains similar to what we discussed in Section~\ref{sec:exp}. The results of all the algorithms are closer to each other than what we discussed in Section~\ref{sec:exp} because of the smaller settings. Interestingly, almost always (with a few exceptions), the total weight of the solution obtained by both of our $2$-approx and $(2+\epsilon)$-approx algorithms matches the solution of the exact DP algorithm (which we know is optimum), and in all the settings, at least one of our approximation algorithms exactly matches the optimum solution. Although this is a rather strict and smaller setting, this observation is a good sign for the performance of our approximation algorithm compared to the exact optimum solution. The greedy baseline usually performs slightly worse, and the difference can get rather large in some settings, as can be seen in Figure~\ref{fig:dp_weight_varying_n_stackoverflow}.
In terms of the running time and efficiency, the results are shown in Figures~\ref{fig:dp_time_varying_n_census}-\ref{fig:dp_time_varying_group_yelp}.
The DP algorithm has the slowest performance among all the other methods, as expected. Moreover, as shown in Figures~\ref{fig:dp_time_varying_group_census}-\ref{fig:dp_time_varying_group_yelp}, the exact DP algorithm quickly becomes intractable as the number of groups increases, since its runtime grows exponentially with $\ell$.
\end{envrevone}

\begin{envrevtwo}

\paragraph{Varying demands}
Figures~\ref{fig:weight_varying_demands_census} and \ref{fig:weight_varying_demands_music} report further results on the impact of demand variation on the other two remaining datasets. The results follow a similar pattern to that discussed for the other datasets in Section~\ref{sec:exp}.
    \paragraph{Oversatisfaction of demands} Recall Example~\ref{ex1}, where a dataset is given and the goal is to select a minimum number of data points subject to demographic demand constraints. In this example, let the set of demographic groups to be covered be denoted by 
    $$\Gee = \{\text{male}, \text{female}, \text{young}\},$$
    and the collection of images be given by
\[
\Dee = \big\{\{\text{male}, \text{young}\}, \{\text{female}, \text{young}\}, \{\text{male}\}, \{\text{female}\}\big\}.
\]
Each image represents the set of demographic attributes present in that image. Moreover, let the demand requirements be specified as
\[
Q_{\text{male}} = 1, \quad Q_{\text{female}} = 2, \quad Q_{\text{young}} = 1,
\]
and let all data points have unit weight.
Therefore, the objective is to select a minimum-cardinality subset of images such that all demand constraints are satisfied, that is, the selected images collectively contain at least one male individual, at least two female individuals, and at least one young individual. Two feasible solutions are
\begin{align*}
S_1 &= \big\{ \{\text{female}, \text{young}\},\,
               \{\text{male}\},\,
               \{\text{female}\} \big\}, \\
S_2 &= \big\{ \{\text{male}, \text{young}\},\,
               \{\text{female}\},\,
               \{\text{female}, \text{young}\} \big\}.
\end{align*}
Both $S_1$ and $S_2$ are valid solutions since $S_1, S_2 \subseteq \Dee$ and each satisfies all demand constraints. However, solution $S_2$ covers two young individuals, exceeding the required demand for the \emph{young} demographic. As a result, we have $RSS(S_1) = 0$ and $RSS(S_2) = 1$.

Selecting $S_2$ therefore leads to an \emph{over-representation} of the \emph{young} demographic. This over-representation is quantified by the function $RSS(\cdot)$, which measures the sum of squared deviations between the required demands and the actual coverage induced by a solution. Although minimizing $RSS(\cdot)$ is not explicitly part of the objective function of our problem, the experimental results presented in Section~\ref{sec:exp} demonstrate that our approach consistently yields solutions with substantially lower RSS values compared to baseline methods. This behavior is primarily due to the fact that our algorithms tend to find higher-quality solutions with smaller total weight. In contrast, solutions with larger weights often require selecting additional images beyond what is strictly necessary, which in turn increases over-coverage and leads to higher RSS values. This issue becomes important in the fairness applications.

In the above example, consider the behavior of the greedy algorithm. In the first iteration, it may select the image $\{\text{male}, \text{young}\}$ since it simultaneously covers two demographic groups. After this choice, in order to satisfy the demand $Q_{\text{female}} = 2$, the algorithm is forced to select both $\{\text{female}\}$ and $\{\text{female}, \text{young}\}$ images. This sequence of selections results in the solution $S_2$, which unnecessarily over-represents the \emph{young} demographic.

\end{envrevtwo}

\begin{figure*}[!tb] 
\centering
\includegraphics[width=0.65\textwidth]{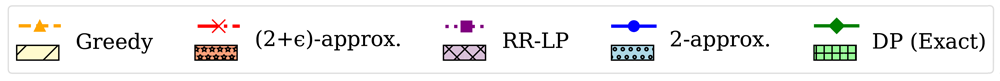}\\
    \begin{minipage}[t]{0.24\linewidth}
        \centering
        \includegraphics[width=\textwidth]{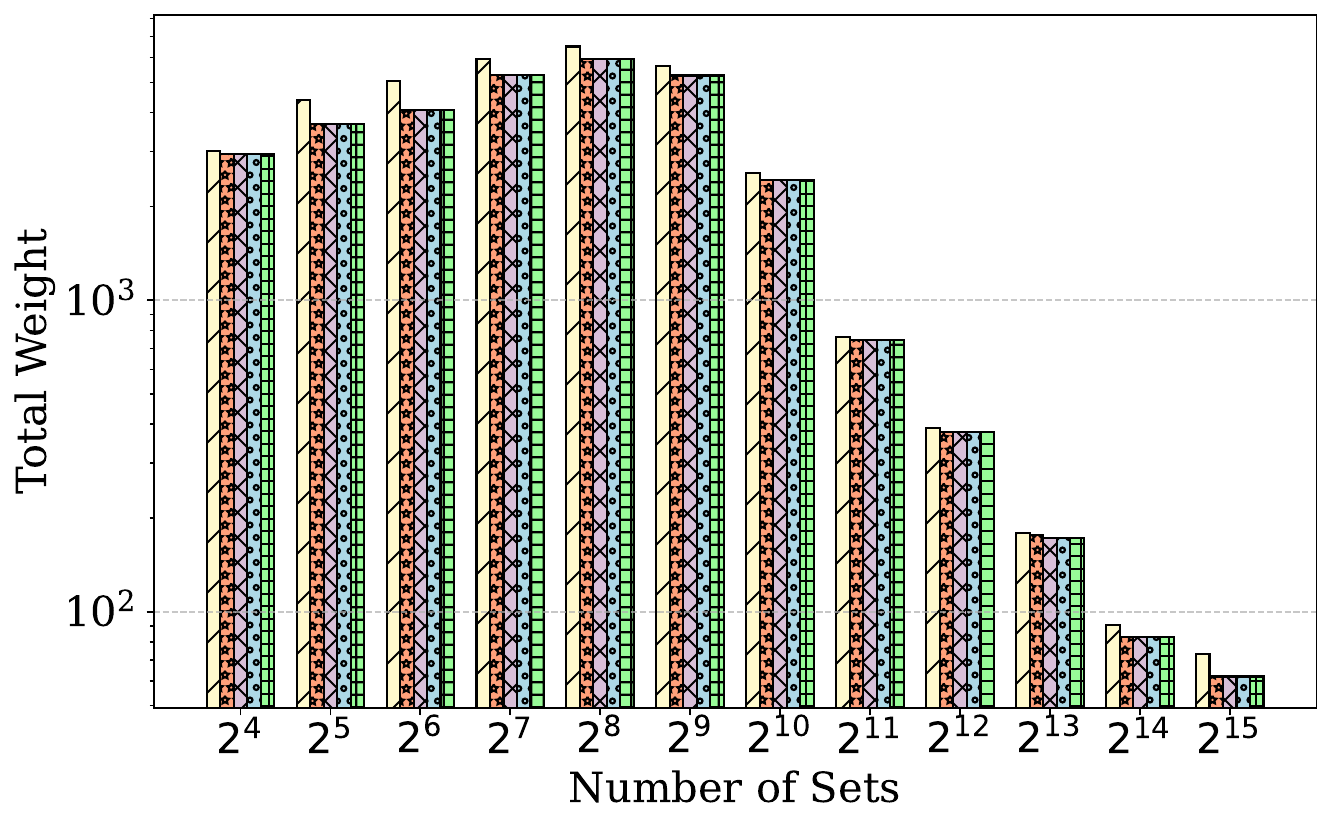}
        \vspace{-2.5em}
        \caption{\revone{\small Impact of varying number of sets $n$ on the solution’s total weight, {\sf Census}}}
        \label{fig:dp_weight_varying_n_census}
    \end{minipage}
    \hfill
    \begin{minipage}[t]{0.24\linewidth}
        \centering
        \includegraphics[width=\textwidth]{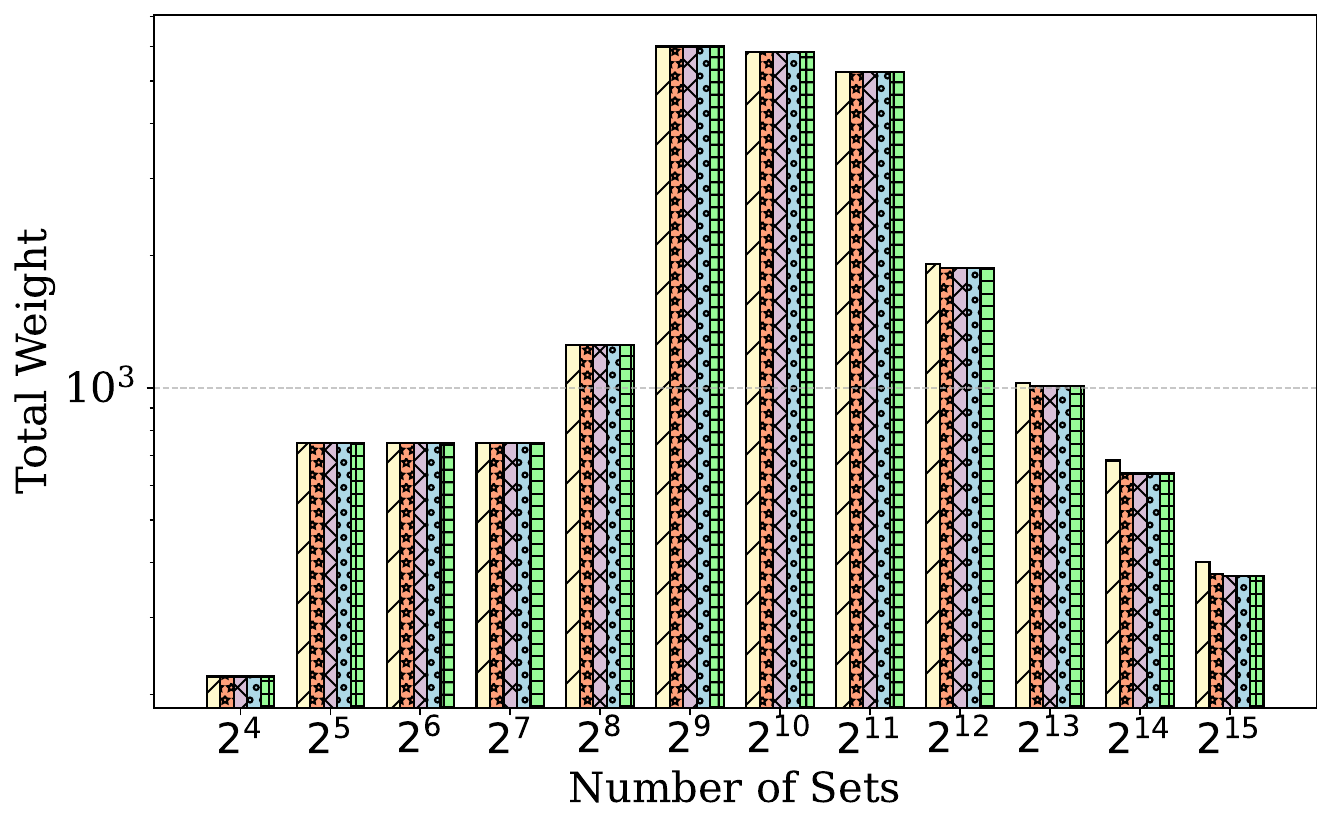}
        \vspace{-2.5em}
        \caption{\revone{\small Impact of varying number of sets $n$ on the solution’s total weight, {\sf Music}}}
        \label{fig:dp_weight_varying_n_music}
    \end{minipage}
    \hfill
    \begin{minipage}[t]{0.24\linewidth}
        \centering
        \includegraphics[width=\textwidth]{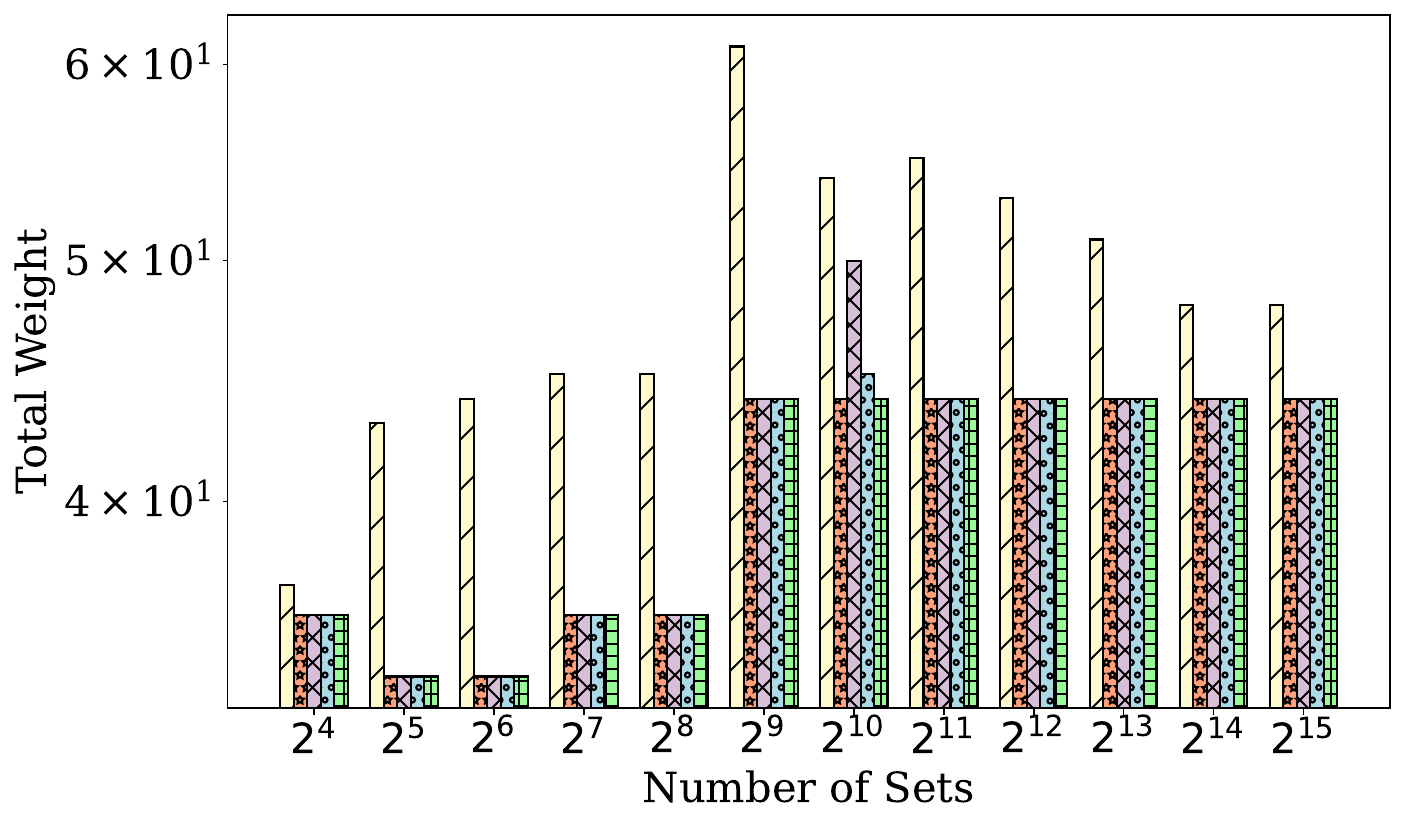}
        \vspace{-2.5em}
        \caption{\revone{\small Impact of varying number of sets $n$ on the solution’s total weight, {\sf Stack Overflow}}}
        \label{fig:dp_weight_varying_n_stackoverflow}
    \end{minipage}
    \hfill
    \begin{minipage}[t]{0.24\linewidth}
        \centering
        \includegraphics[width=\textwidth]{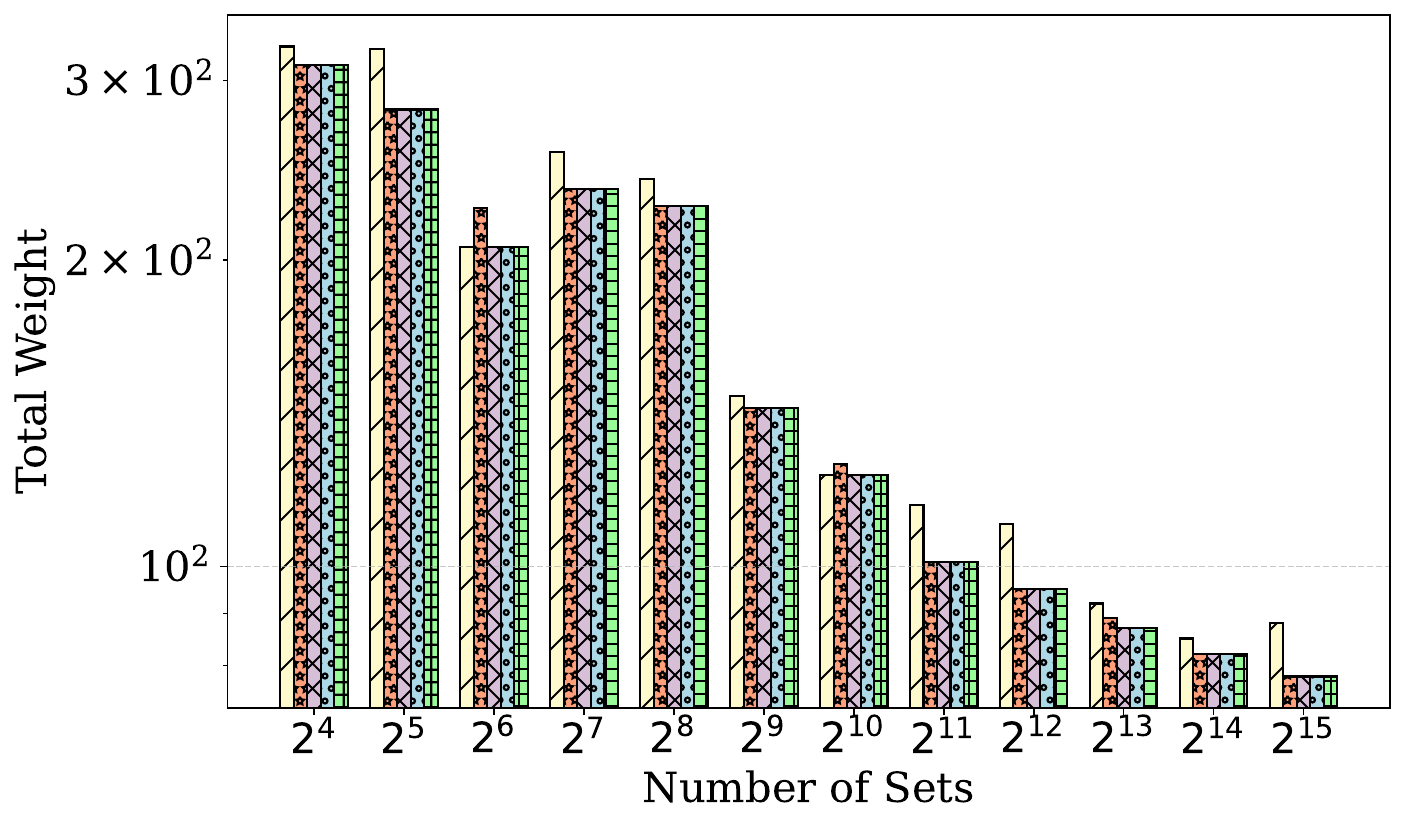}
        \vspace{-2.5em}
        \caption{\revone{\small Impact of varying number of sets $n$ on the solution’s total weight, {\sf Yelp}}}
        \label{fig:dp_weight_varying_n_yelp}
    \end{minipage}
    \hfill
    \begin{minipage}[t]{0.24\linewidth}
        \centering
        \includegraphics[width=\textwidth]{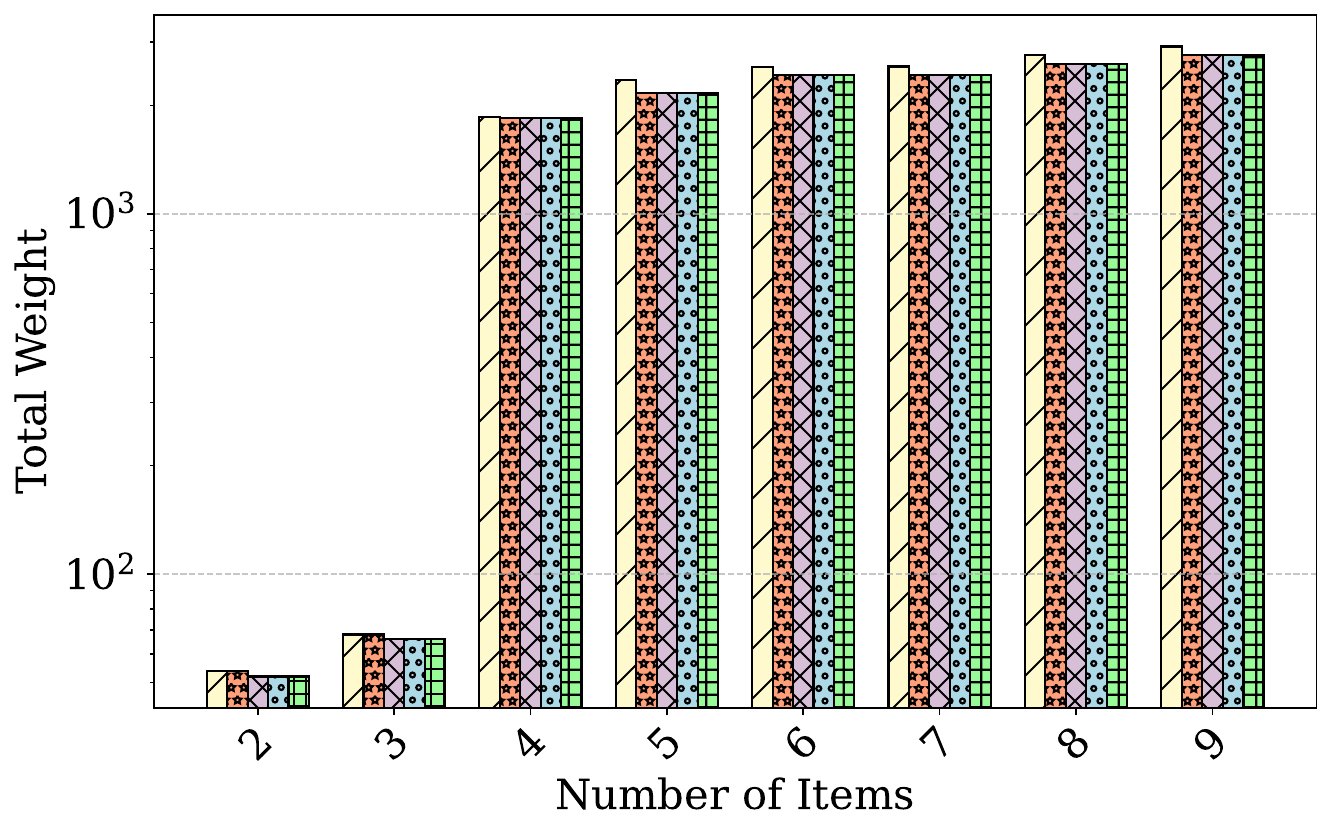}
        \vspace{-2.5em}
        \caption{\revone{\small Impact of varying number of items $\ell$ on the solution’s total weight, {\sf Census}}}
        \label{fig:dp_weight_varying_group_yelp}
    \end{minipage}
    \hfill
    \begin{minipage}[t]{0.24\linewidth}
        \centering
        \includegraphics[width=\textwidth]{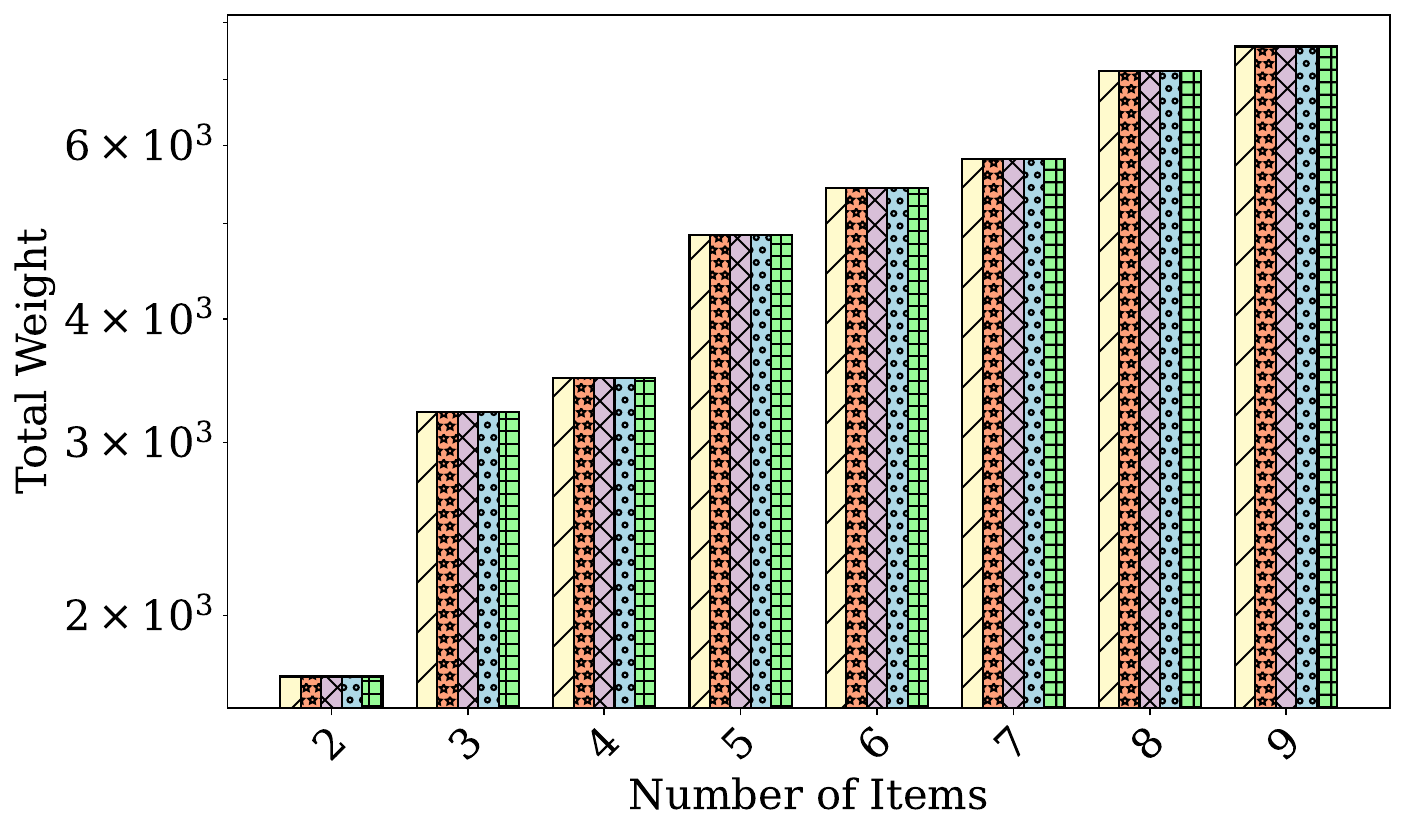}
        \vspace{-2.5em}
        \caption{\revone{\small Impact of varying number of items $\ell$ on the solution’s total weight, {\sf Music}}}
        \label{fig:dp_weight_varying_group_music}
    \end{minipage}
    \hfill
    \begin{minipage}[t]{0.24\linewidth}
        \centering
        \includegraphics[width=\textwidth]{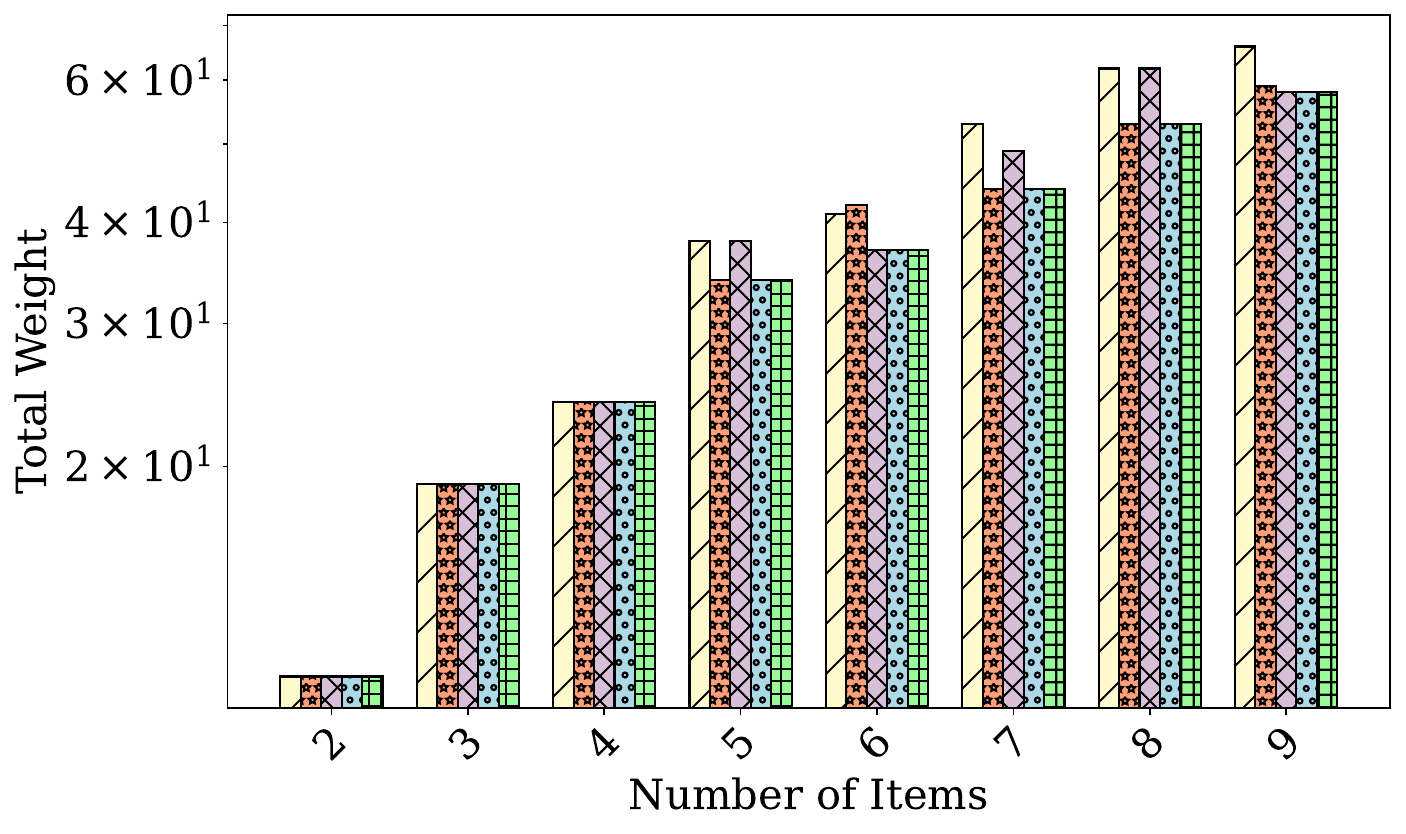}
        \vspace{-2.5em}
        \caption{\revone{\small Impact of varying number of items $\ell$ on the solution’s total weight, {\sf Stack Overflow}}}
        \label{fig:dp_weight_varying_group_stackoverflow}
    \end{minipage}
    \hfill
    \begin{minipage}[t]{0.24\linewidth}
        \centering
        \includegraphics[width=\textwidth]{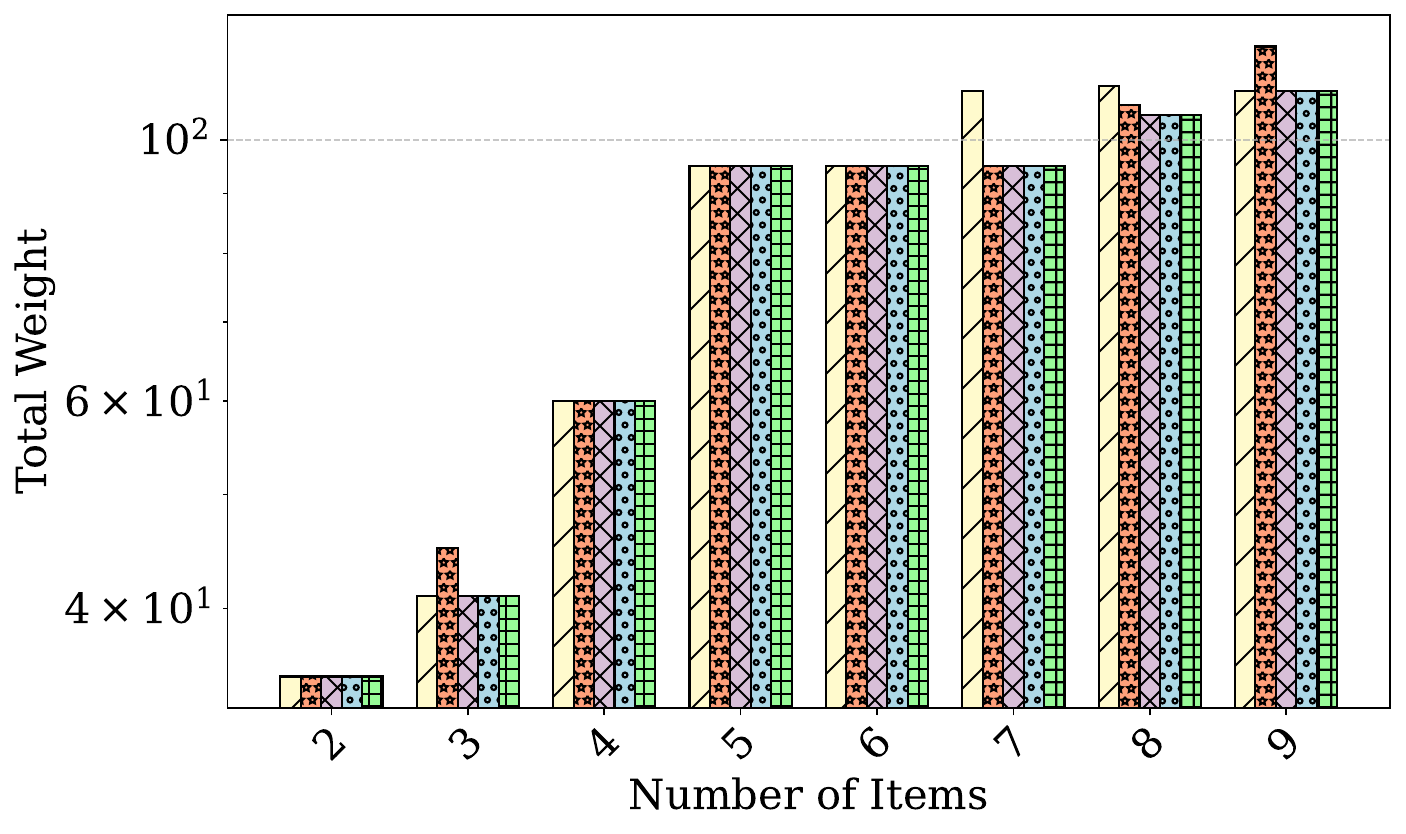}
        \vspace{-2.5em}
        \caption{\revone{\small Impact of varying number of items $\ell$ on the solution’s total weight, {\sf Yelp}}}
        \label{fig:dp_weight_varying_group_yelp}
    \end{minipage}
    \hfill
    \begin{minipage}[t]{0.24\linewidth}
        \centering
        \includegraphics[width=\textwidth]{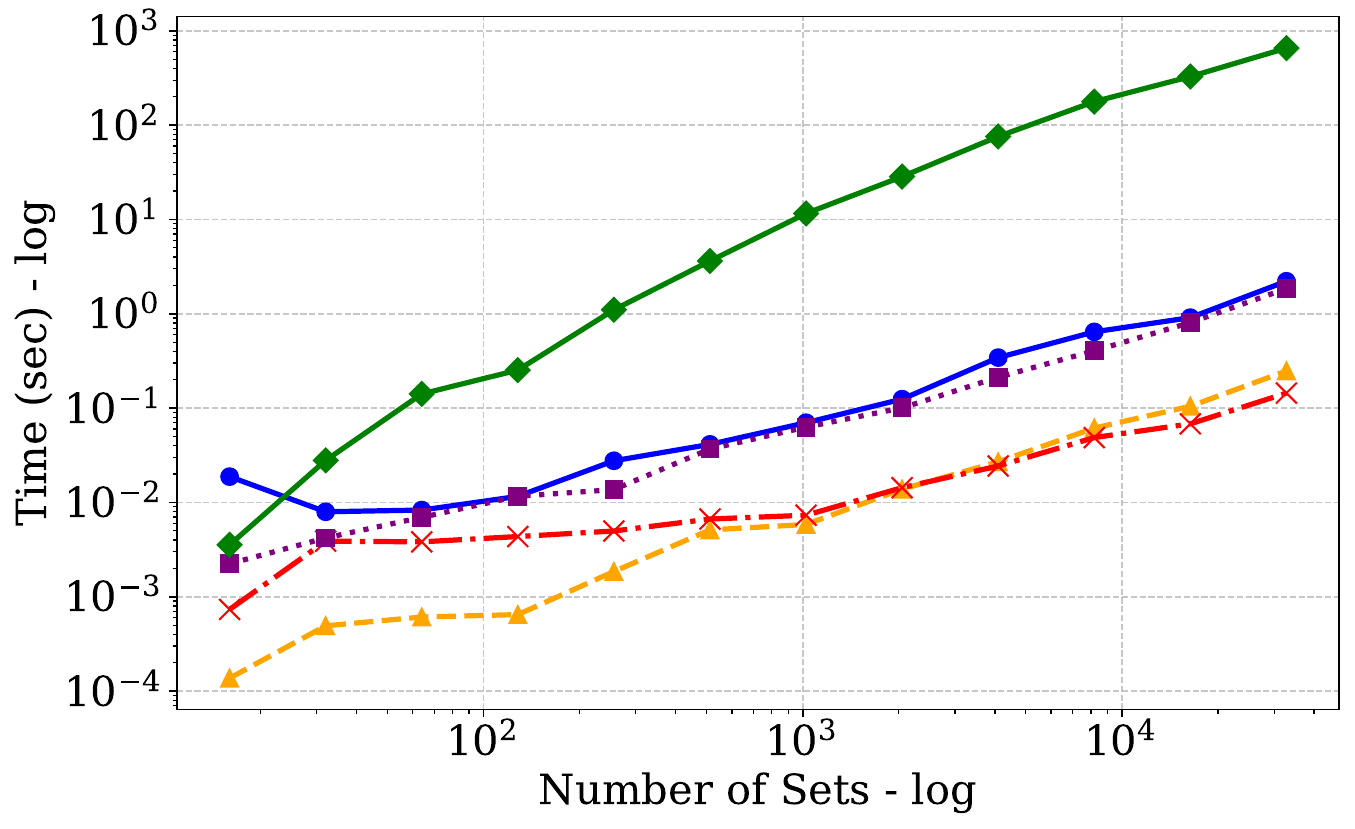}
        \vspace{-2.5em}
        \caption{\revone{\small Impact of varying number of sets $n$ on the running time, {\sf Census}}}
        \label{fig:dp_time_varying_n_census}
    \end{minipage}
    \hfill
    \begin{minipage}[t]{0.24\linewidth}
        \centering
        \includegraphics[width=\textwidth]{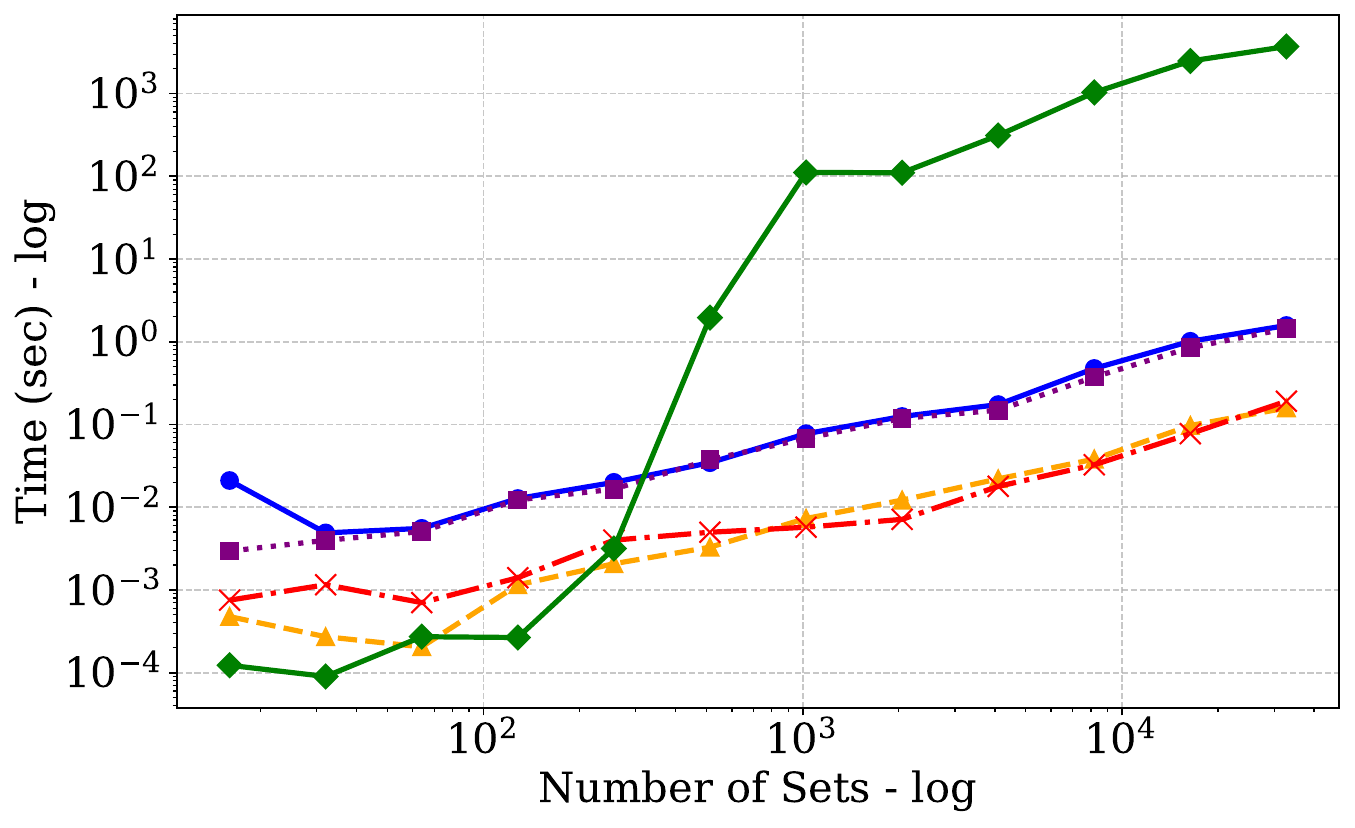}
        \vspace{-2.5em}
        \caption{\revone{\small Impact of varying number of sets $n$ on the running time, {\sf Music}}}
        \label{fig:dp_time_varying_n_music}
    \end{minipage}
    \hfill
    \begin{minipage}[t]{0.24\linewidth}
        \centering
        \includegraphics[width=\textwidth]{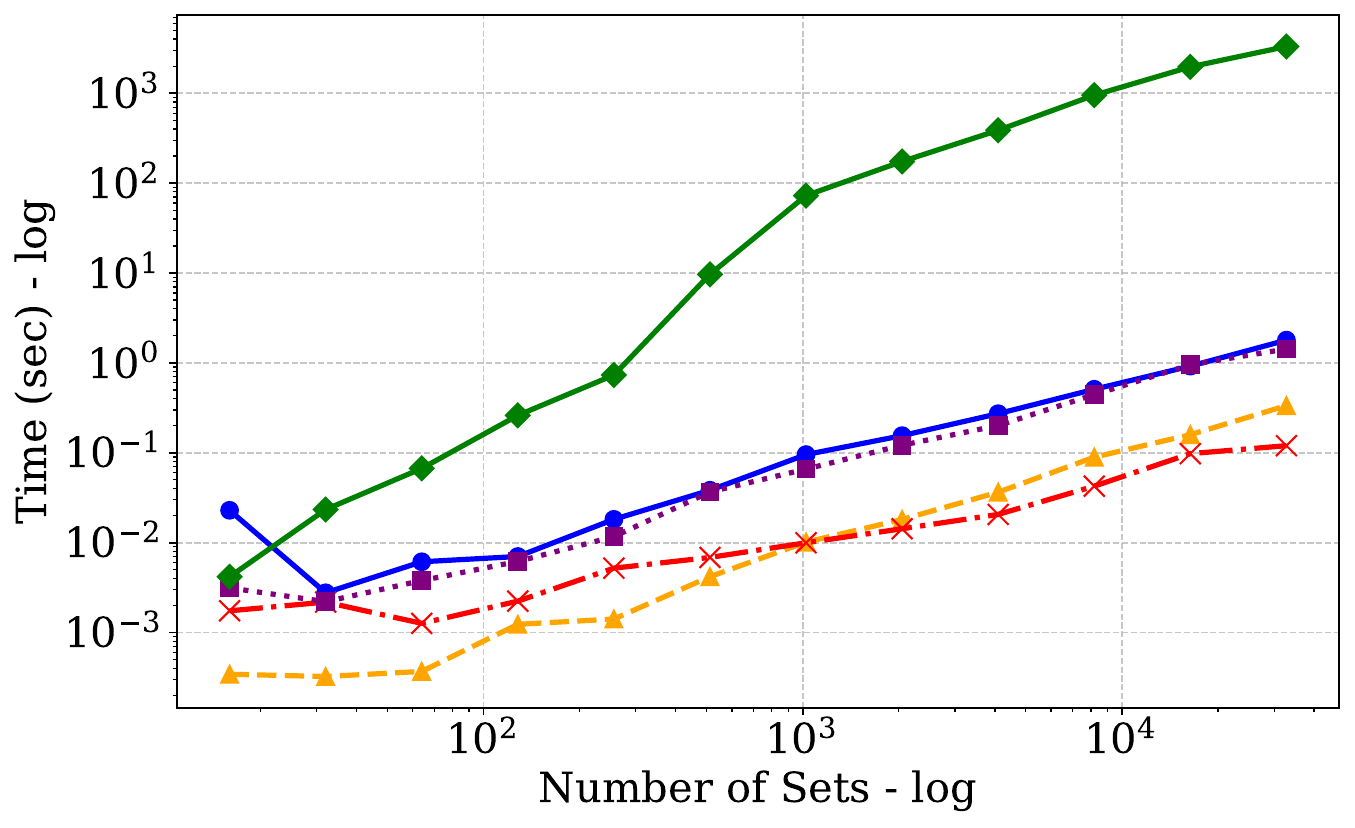}
        \vspace{-2.5em}
        \caption{\revone{\small Impact of varying number of sets $n$ on the running time, {\sf Stack Overflow}}}
        \label{fig:dp_time_varying_n_stackoverflow}
    \end{minipage}
    \hfill
    \begin{minipage}[t]{0.24\linewidth}
        \centering
        \includegraphics[width=\textwidth]{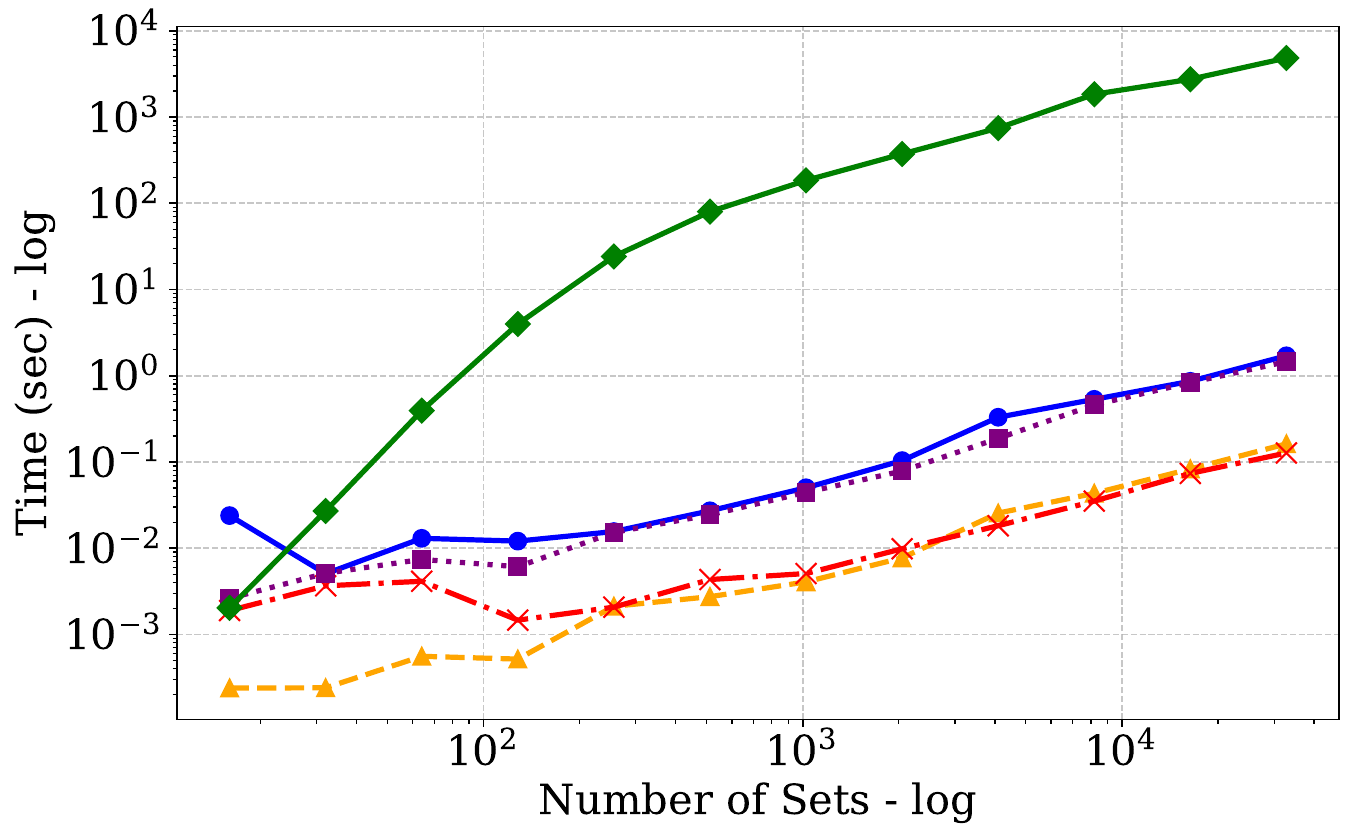}
        \vspace{-2.5em}
        \caption{\revone{\small Impact of varying number of sets $n$ on the running time, {\sf Yelp}}}
        \label{fig:dp_time_varying_n_yelp}
    \end{minipage}
    \hfill
    \begin{minipage}[t]{0.24\linewidth}
        \centering
        \includegraphics[width=\textwidth]{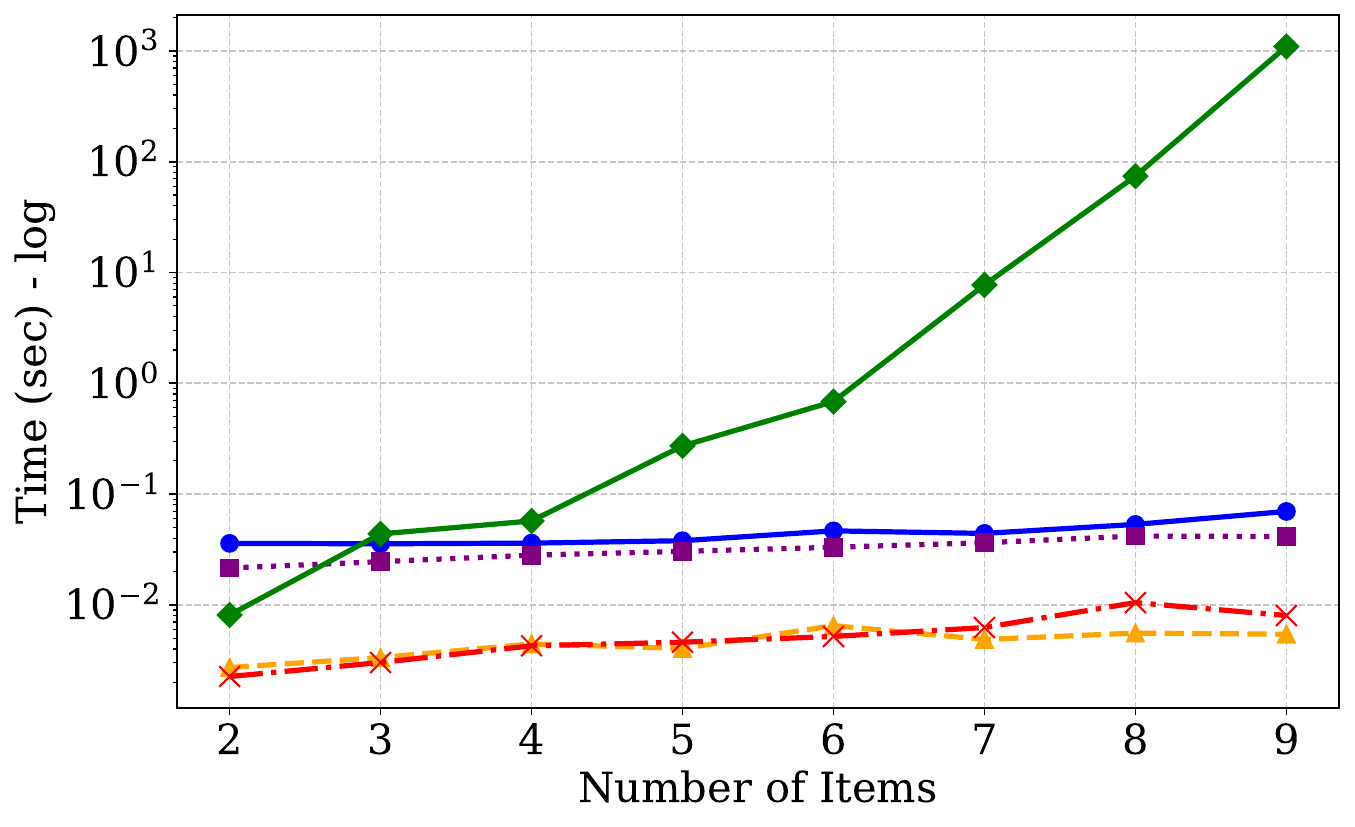}
        \vspace{-2.5em}
        \caption{\revone{\small Impact of varying number of items $\ell$ on the running time, {\sf Census}}}
        \label{fig:dp_time_varying_group_census}
    \end{minipage}
    \hfill
    \begin{minipage}[t]{0.24\linewidth}
        \centering
        \includegraphics[width=\textwidth]{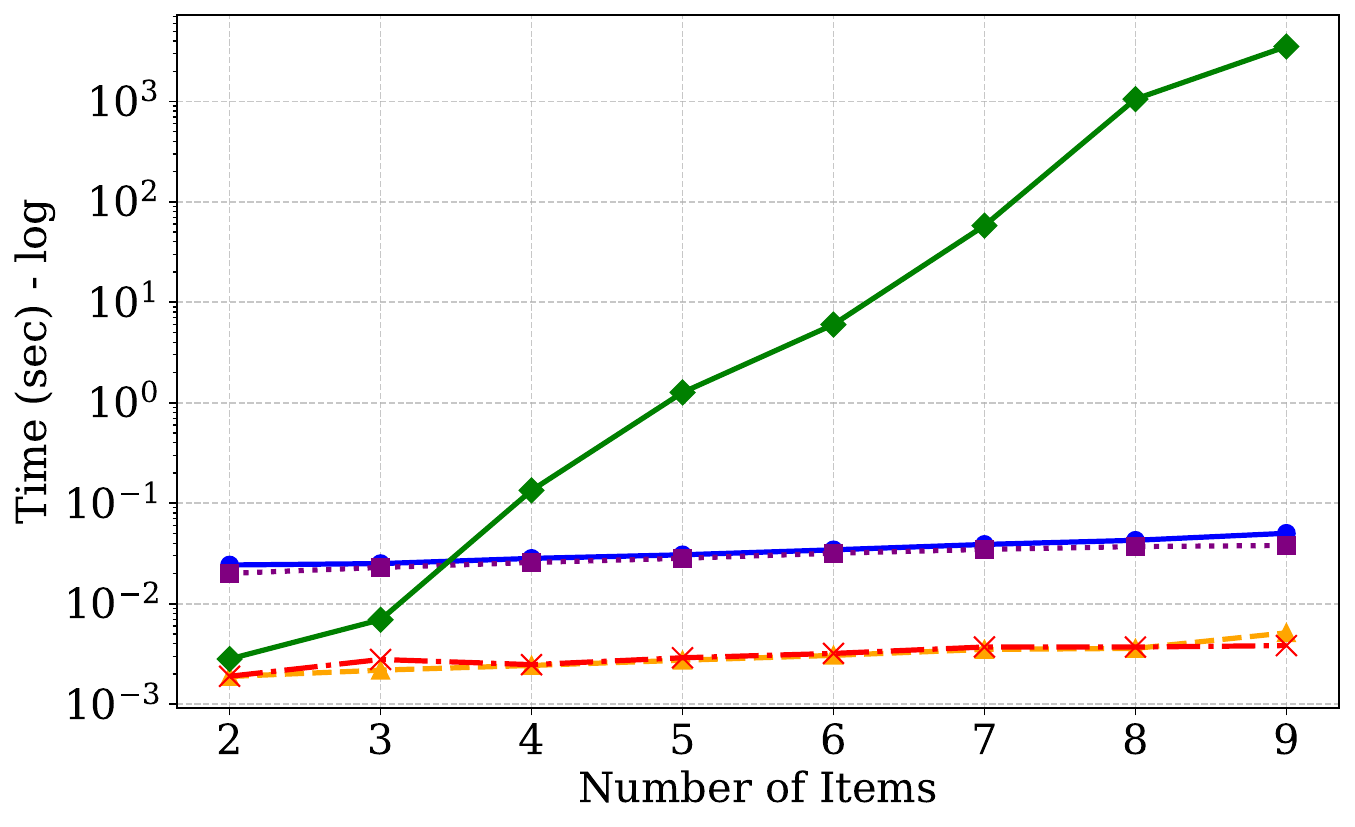}
        \vspace{-2.5em}
        \caption{\revone{\small Impact of varying number of items $\ell$ on the running time, {\sf Music}}}
        \label{fig:dp_time_varying_group_music}
    \end{minipage}
    \hfill
    \begin{minipage}[t]{0.24\linewidth}
        \centering
        \includegraphics[width=\textwidth]{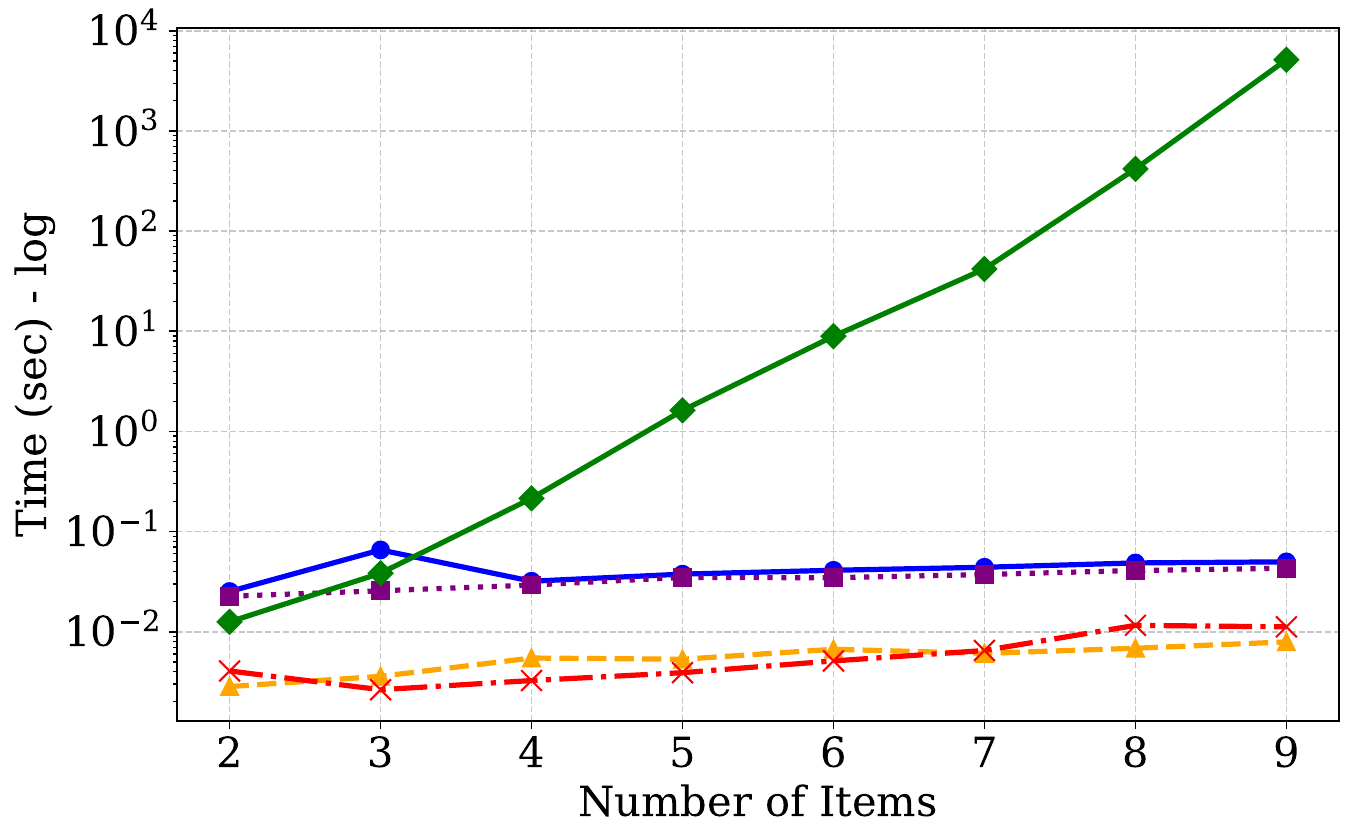}
        \vspace{-2.5em}
        \caption{\revone{\small Impact of varying number of items $\ell$ on the running time, {\sf Stack Overflow}}}
        \label{fig:dp_time_varying_group_stackoverflow}
    \end{minipage}
    \hfill
    \begin{minipage}[t]{0.24\linewidth}
        \centering
        \includegraphics[width=\textwidth]{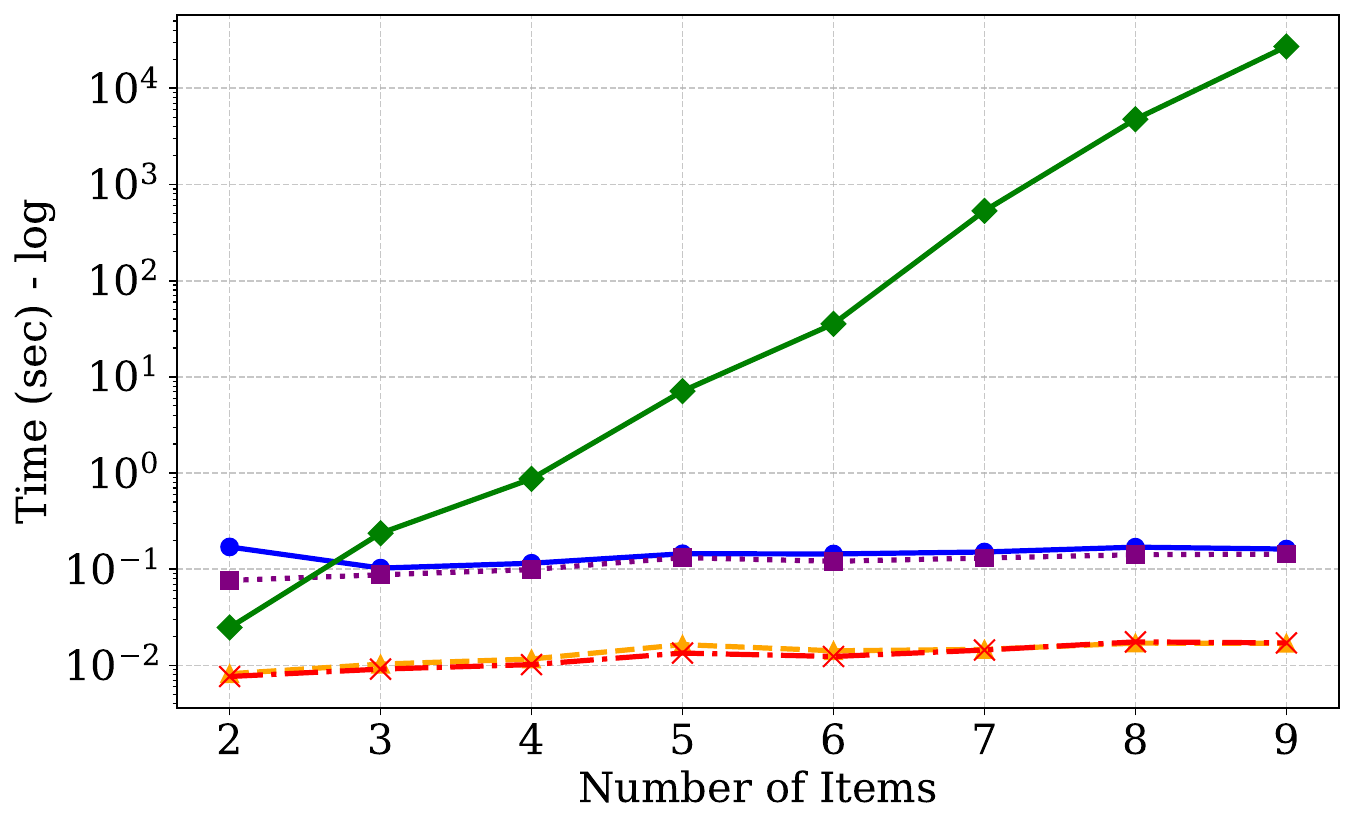}
        \vspace{-2.5em}
        \caption{\revone{\small Impact of varying number of items $\ell$ on the running time, {\sf Yelp}}}
        \label{fig:dp_time_varying_group_yelp}
    \end{minipage}
    \hfill
    \begin{minipage}[t]{0.48\linewidth}
        \centering
        \includegraphics[width=\textwidth]{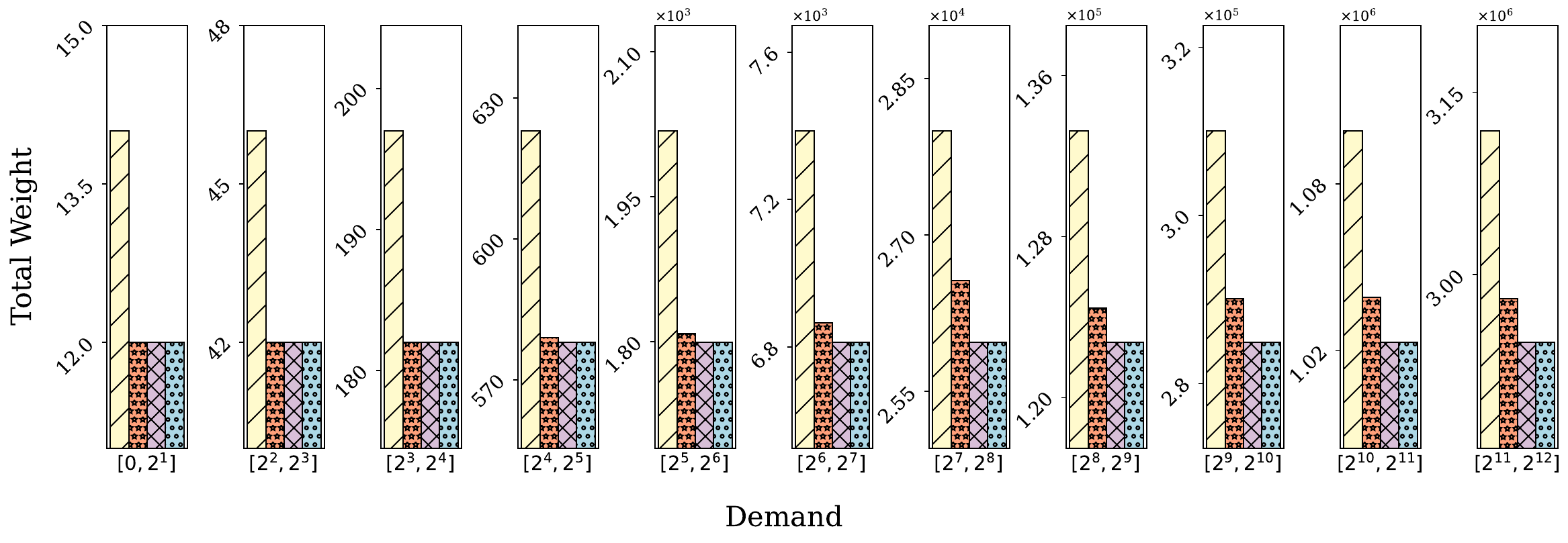}
        \vspace{-2.5em}
        \caption{\revtwo{\small Impact of varying demands on the solution’s total weight, {\sf Census}}}
        \label{fig:weight_varying_demands_census}
    \end{minipage}
    \hfill
    \begin{minipage}[t]{0.48\linewidth}
        \centering
        \includegraphics[width=\textwidth]{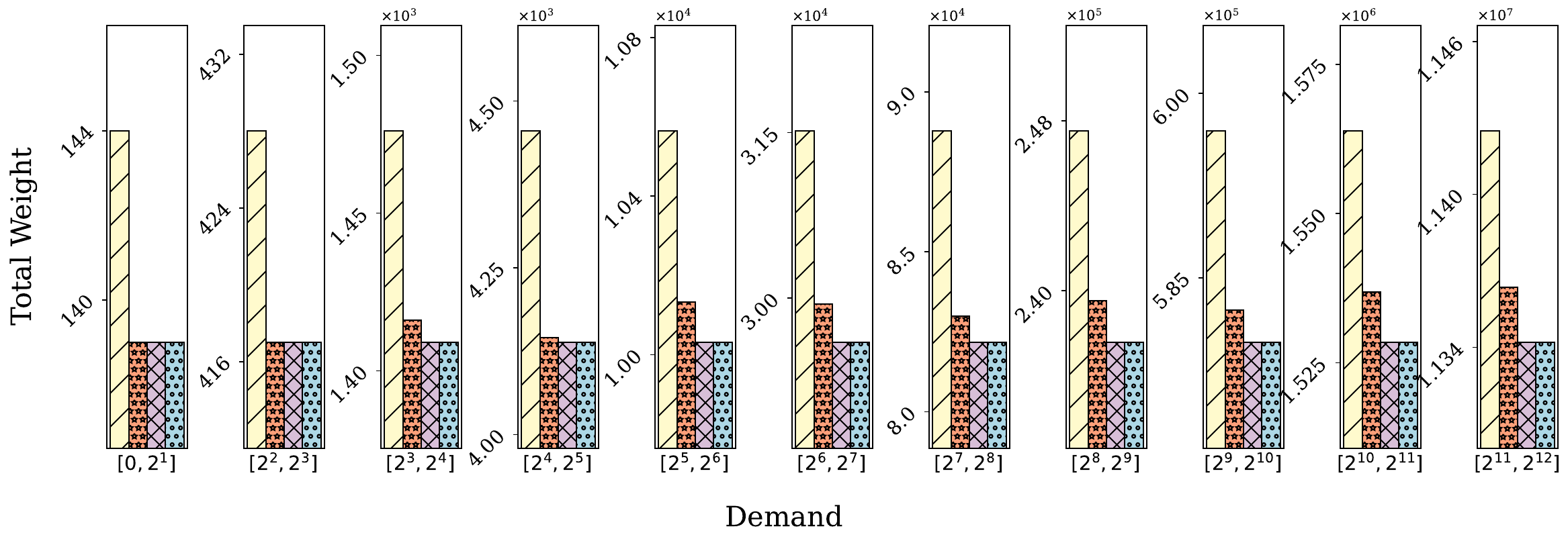}
        \vspace{-2.5em}
        \caption{\revtwo{\small Impact of varying demands on the solution’s total weight, {\sf Music}}}
        \label{fig:weight_varying_demands_music}
    \end{minipage}
\end{figure*}


\end{document}